%% file: sparseDOArev.tex
\documentclass[journal]{IEEEtran}

\ifCLASSINFOpdf
\else
   \usepackage[dvips]{graphicx}
\fi

\usepackage{cite}
\usepackage{amsmath,amssymb,amsfonts,amsthm}
\usepackage{graphicx}
\usepackage{textcomp}

\usepackage{xcolor}
\usepackage[colorlinks=true,  linkcolor=black, citecolor=blue, breaklinks=true, urlcolor=blue]{hyperref}

\usepackage{graphicx}

\usepackage[ruled,vlined,linesnumbered]{algorithm2e}
\SetKwRepeat{Do}{do}{while}%

\usepackage{xcolor}
\usepackage{booktabs}
\usepackage{url}
\usepackage[latin1]{inputenc}

\definecolor{darkgreen}{rgb}{0.2, 0.6, 0.05}
\usepackage{enumitem}

\usepackage{subcaption}

\newcounter{cremark}
\newtheorem{remark}[cremark]{Remark}
\newcounter{clemma}
\newtheorem{lemma}[clemma]{Lemma}

\newcommand{\Minv}{\boldsymbol{\Theta}}
\newcommand{\iter}{\mathsf{t}}

\usepackage{pgfplots}
\usepackage{tikz}
\usepgfplotslibrary{groupplots}
\pgfplotsset{compat=1.13}
\newlength\fwidth

\usepackage{algorithmic}
\let\oldnl\nl
\newcommand{\nonl}{\renewcommand{\nl}{\let\nl\oldnl}}

\usetikzlibrary{arrows,petri,topaths}

\newcommand{\E}{\mathbb{E}}                  
\newcommand{\bo}[1]{\mathbf{#1}}              
\newcommand{\bom}[1]{\boldsymbol{#1}}    

\newcommand{\supp}{\mathsf{supp}}
\newcommand{\pdim}{N}
\newcommand{\ndim}{L}
\newcommand{\Gam}{\boldsymbol{\Gamma}}

\renewcommand{\Pi}{{\mathcal P}}
\newcommand{\hop}{\mathsf{H}}        

\newcommand{\beq}{\begin{equation}}
\newcommand{\eeq}{\end{equation}}
\newcommand{\bmat}{\begin{pmatrix}}
\newcommand{\emat}{\end{pmatrix}}

\newcommand{\R}{\mathbb{R}}

\newcommand{\X}{{\bf X}}
\newcommand{\SCM}{\hat{\boldsymbol{\Sigma}}}

\newcommand{\A}{\mathbf{A}}

\newcommand{\Y}{\mathbf{Y}}
\newcommand{\mE}{\mathbf{E}}

\renewcommand{\S}{\hat{\boldsymbol{\Sigma}}} 

\newcommand{\kdim}{K}

\newcommand{\sigmaw}{\sigma^2}

\newcommand{\C}{\mathbb{C}} 
\newcommand{\rsupp}{\mathsf{supp}} 

\newcommand{\var}{\mathrm{var}}

\renewcommand{\b}{\mathbf{b}}

\newcommand{\x}{\bo x}

\newcommand{\e}{\mathbf{e}}
\newcommand{\y}{\bo y}

\renewcommand{\a}{{\bo a}}

\newcommand{\M}{\bom \Sigma}

\DeclareMathOperator{\tr}{tr}

\DeclareMathOperator{\diag}{diag}
\DeclareMathOperator{\cov}{cov}

\newcommand{\gam}{\boldsymbol{\gamma}}

\begin{document}

\title{Sparse signal recovery and source localization via covariance learning}

\author{Esa~Ollila,~\IEEEmembership{Senior Member,~IEEE}
\thanks{Esa Ollila is with the Department of Information and Communications Engineering, Aalto University, FI-00076 Aalto, Finland (e-mail: esa.ollila@aalto.fi).}}

\maketitle

\begin{abstract} In the Multiple Measurements Vector (MMV) model, measurement vectors are connected to unknown, jointly sparse signal vectors through a linear regression model employing a single known measurement matrix (or dictionary).  Typically, the number of  atoms (columns of the dictionary) is greater than the number measurements and the sparse signal recovery problem is generally ill-posed.  In this paper,  we treat the signals and measurement noise as independent Gaussian random vectors with unknown signal covariance matrix and noise variance, respectively, and  characterize the solution of the likelihood equation in terms of fixed point equation,  thereby enabling the recovery of the sparse signal support (sources with non-zero variances) via a  block coordinate descent (BCD) algorithm that leverage  the FP characterization of the likelihood equation.  Additionally,  a  greedy pursuit method, analogous to popular simultaneous orthogonal matching pursuit (OMP), is introduced.   Our numerical examples demonstrate  effectiveness of the proposed covariance learning (CL) algorithms both in classic sparse signal recovery as well as in direction-of-arrival (DOA) estimation problems where they perform favourably  compared to the state-of-the-art algorithms under a broad variety of settings.  
\end{abstract}

\IEEEpeerreviewmaketitle

\section{Introduction}
\label{sec:intro}

\IEEEPARstart{I}{n the}  \emph{multiple measurements vector} (MMV) \cite{duarte2011structured} model, the measurement vectors  $\y_i \in \mathbb{C}^N$ follow the generative model 
\beq \label{eq:MMV}
\y_{l} = \A \x_{l}  + \mathbf{e}_{\ell} , \quad l = 1,\ldots,  L,
\eeq 
where $\A=(\a_1  \ \cdots \ \a_{M}) \in \C^{\pdim \times M}$ is a fixed (known) overcomplete matrix ($M > N $) 
and   $\mathbf{e}_l  \in \C^{\pdim}$ is an unobserved zero mean white noise random vector,  i.e., $\cov(\e_l)=\sigmaw \mathbf{I}$. 
Matrix $\A$ is  called the \emph{dictionary} or \emph{measurement matrix}, and its column vectors $\mathbf{a}_i$ are referred to as \emph{atoms}. 
The unobserved random  signal vectors $\x_l = (x_{1l}, \ldots,x_{Ml})^\top$, $l=1,\ldots,L$ are assumed to be sparse, {\it i.e.,} most of their elements are zero.

Letting $\mathbf{Y} = ( \y_1 \, \cdots \,  \y_{L} ) \in \C^{\pdim \times L}$ to denote the matrix of measurement vectors, we can write the model  \eqref{eq:MMV} in matrix form as  
 \beq \label{eq:MMVmodel} 
\mathbf{Y}  = \A \mathbf{X} + \mathbf{E},
\eeq 
where matrices $\X  = (x_{ml}) \in \C^{M \times L}$ and  $\mE  = (e_{nl})  \in \C^{N \times L} $  contain the  signal  and error vectors as columns, respectively. 
The key assumption underlying the  MMV model  is that the signal matrix $\mathbf{X}$ is $\kdim$-rowsparse, {\it i.e.,} at most  $\kdim$ rows of $\X$ contain
non-zero entries. The rowsupport  of  $\X \in \C^{M \times L}$ is the index set of rows  containing non-zero  elements:
\[
\mathcal M = \rsupp(\X) 
= \{   i \in  [\! [   M ] \!]   :  x_{ij} \neq 0 \,  \mbox{ for some $j \in [\! [ L ] \!] $} \} 
\]
where $[\! [  M  ] \!]  = \{1,\ldots, M \}$. 
Since $\X$ is $\kdim$-rowsparse,  i.e.,  $| \rsupp(\X)| = \kdim$, 
joint estimation can lead  both to computational advantages and increased reconstruction accuracy \cite{tropp2006algorithms,tropp:2006,chen_huo:2006,eldar_rauhut:2010,duarte2011structured,blanchard2014greedy}. 
 In sparse signal recovery (SSR) problems, the object of interest is identifying the support $\mathcal M$, given only the data $\Y$, the dictionary $\A$, and the sparsity level $K$. 
 
In this paper, the source signal vectors $\mathbf{x}_l$  are modelled as  i.i.d. circular Gaussian random vectors with independent elements and zero mean.  Additionally, $\mathbf{x}_l$-s are assumed statistically independent of noise $\e_l$, $l=1,\ldots,L$. This implies that  
$ \y_l \sim \mathcal{C N}_{\pdim}(\mathbf{0}, \M)$, where the positive definite Hermitian (PDH) $\pdim \times \pdim$ covariance matrix $\M=\cov(\y_l)$ has the form 
\beq \label{eq:M}
\M = \A \Gam \A^\hop + \sigma^2 \mathbf{I}  =  \sum_{i=1}^M \gamma_i \a_i \a_i^\hop + \sigma^2 \mathbf{I} . 
\eeq
where the signal covariance matrix $\Gam = \cov(\x_l) = \diag(\gam)$  is a diagonal matrix and $\gam \in \mathbb{R}^M_{\geq 0}$ is a vector of signal powers with only $K$ non-zero elements. Hence, $\mathcal M = \rsupp(\X)=\supp(\gam)$,  and  \emph{covariance learning} (CL-)based support recovery algorithms can be constructed by  
minimizing the negative log-likelihood function (LLF) of the data $\Y$, defined by (after scaling by $1/\ndim$ and ignoring additive constants) 
\beq
\begin{aligned}  \label{eq:SML}
\ell( \boldsymbol{\gamma}, \sigmaw \mid \Y, \A) = 
 \tr(  (\A& \diag(\gam)\A^\hop + \sigmaw \mathbf{I})^{-1} \SCM )  \\ &+  \log |  \A \diag(\gam) \A^\hop + \sigmaw \mathbf{I} |  
\end{aligned}
\eeq
where $\SCM$ is  the sample covariance matrix (SCM),  
\[
\SCM = \frac{1}{\ndim} \sum_{l=1}^{\ndim} \y_{l} \y_l^\hop = L^{-1} \mathbf{Y} \mathbf{Y}^\hop, 
\]
and $\tr(\cdot)$ and $| \cdot |$ denote the matrix trace and determinant. 

Approaches for tackling the MMV problem can be categorized into modeling and non-modeling methods. The non-modelling category involves treating $\x_l$-s as  deterministic latent signal sequences. Then $\X$ and its rowsupport are determined by either using convex relaxation based optimization methods (e.g., $\ell_1$-minimization) or greedy pursuit algorithms, such as simultaneous orthogonal matching pursuit (SOMP) \cite{tropp2006algorithms} or  the simultaneous normalized iterative hard thresholding (SNIHT)  \cite{blanchard2014greedy}. 
In the modeling category, the variables $\mathbf{x}_l$ are considered random, and the widely assumed Gaussianity of both source signals and noise gives rise to stated maximum likelihood (ML) estimation problem. The objective is to determine the minimizer of \eqref{eq:SML} along with the identification of the support of non-zero source powers. 
A popular Bayesian approach,  sparse Bayesian learning (SBL) \cite{tipping2001sparse}, regards signal powers ($\gamma_i$) as random hyperparameters governed by a hierarchical prior distribution. An innovative  hybrid approach, known as M-SBL algorithm \cite{wipf2007empirical,wipf2007beamforming},  employs an empirical Bayes method to  construct an EM algorithm for solving \eqref{eq:SML}. This algorithm  proceeds as outlined below, assuming a known value for the noise power ($\sigmaw$): 
\begin{enumerate} 
\item Initialize:  Set $\Gam^{(0)}=\diag(\gam^{(0)})$ for some initial $\gam^{(0)} \in \mathbb{R}^M_{\geq 0}$, and  $\iter =0$. 
\item {\bf E-step}: 
\begin{align*}
\M^{(\iter)} &=  \A \Gam^{(\iter)}  \A^\hop + \sigmaw \mathbf{I} \\
\X^{(\iter )} &= \Gam^{(\iter)}   \A^\hop (\M^{(\iter)})^{-1} \Y , \\ 
\M_x^{(\iter)} &=  \Gam^{(\iter)}  - \Gam^{(\iter)}  \A^\hop (\M^{(\iter)})^{-1} \A \Gam^{(\iter)}   
\end{align*} 
\item {\bf M-step}: $ 
\Gam^{(\iter+1)} = \diag( L^{-1} \X^{(\iter)} \X^{(\iter)\hop} + \M_x^{(\iter)}) 
$
\item If not converged, then increment $\iter$ and repeat steps 1-3. 
\end{enumerate}
 M-SBL method  is computationally demanding, and is known to have slow convergence,  which often prohibits its uses even when number of atoms $M$  is  only moderately large. Although  some faster fixed point SBL update rules have been proposed ({\it e.g.}, \cite[eq. (19)]{wipf2007empirical})  these often provide much worse performance than the original slower rule.  The CL algorithms proposed in this paper avoid estimation of  $\X$, distinguishing them from methods such as M-SBL, SOMP,  or SNIHT  which all require  iterative updates  of the unknown source signal matrix.    Another disadvantage of M-SBL is the assumption that  $\sigmaw$ is known. Although M-step can modified for joint estimation of $\sigmaw$, as mentioned in  \cite{tipping2001sparse,wipf2004sparse}, they tend to provide poor estimation results, which is also attested in our numerical studies. 
Additionally,  the estimate of $\gam$ is not necessarily  very sparse, but  is post-processed by pruning (setting to zero all but $K$-largest elements of final iterate $\hat{\gam}$). 

The  CL-based iterative FP algorithms proposed in this work share similarity to iterative adaptive approach (IAA) algorithm \cite{yardibi2010source}. IAA solves a weighted LS cost function for source waveforms $\{x_{ml}\}_{l=1}^L$ and iteratively recalculates the source signal waveforms and their powers. The convergence of IAA was not shown but was supported by the observation that IAA iterates of source powers can be interpreted as approximate solutions to  Gaussian LLF  \cite[Appendix A]{yardibi2010source}. 
Another closely related technique is Sparse Asymptotic Minimum Variance (SAMV) method \cite{abeida2012iterative} that finds the source powers and the variance that minimize the asymptotic minimum variance (AMV) criterion.  Third related work is \cite{fengler2021non} where the source powers are solved via  cyclic coordinatewise optimization (CWO) of the LLF in \eqref{eq:lik}. Our CL-based scheme aims at directly solving  likelihood equation (l.e.) for signal powers. This is accomplished by characterizing the MLE as solution to a fixed point (FP) equation. This approach enables the formulation of a practical block-coordinate descent (BCD) algorithm that leverage the FP characterization.  Moreover, it paves the way for the design of  a greedy pursuit algorithm  that is analogous to SOMP in the classic SSR setting with deterministic source waveforms. 

The paper is structured as follows. In \autoref{sec:CL-IC}, we characterize the solution to the l.e. via  fixed point equation,  and develop a BCD algorithm that . A greedy pursuit method  analogous to SOMP is derived in \autoref{sec:OMP}.  The effectiveness of the proposed CL-methods in  SSR are validated in \autoref{sec:SSR} via simulation studies. Application to source localization is addressed in \autoref{sec:source}.  In both problems the proposed CL approaches perform favourably  compared to the state-of-the-art algorithms under a broad variety of settings.  Finally, \autoref{sec:concl} concludes.

{\it Notations:}  $\a_{\mathcal M}$ denotes the components  of $\a$ corresponding to support
set  $\mathcal M$ with $|\mathcal M | = K$.  For any $n \times m$ matrix $\A$  we denote  by   $\A_{\mathcal{M}}$  the $n \times K$ submatrix of $\A$  restricted to the columns of $\A$ indexed by set $\mathcal{M}$.    By $\A_{\mathcal M:} $ we  denote the $K \times m$ submatrix of $\A$ restricted to the rows of $\A$ indexed by set $\mathcal{M}$.  When the boolean parameter  $\mathsf{peak}$ equals  $\mathsf{true}$ (resp. $\mathsf{false}$), then the operator $H_K( \mathbf{a},  \mathsf{peak})$ selectively retains only the $K$ largest peaks (resp. elements)  of the vector $\mathbf{a}$ and set other elements to zero. Note that  $\supp \big(H_K(\a,\mathsf{true}) \big)$ will  then return the indices of the $K$ largest elements of vector $\a$.  We use $\A^+ = (\A^\hop \A)^{-1} \A^\hop$ to denote  the pseudo-inverse of matrix $\A$ with full column rank, and  $( \cdot)_+ = \max(\cdot,0)$ an operator that keeps the positive part  of its argument. We use $\diag(\cdot)$ in two distinct contexts:  $\diag(\mathbf{A})$  denotes an $N$-vector comprising the diagonal elements of an $N \times N$ matrix $\mathbf{A}$, whereas  $\diag(\mathbf{a})$ signifies an $N \times N$ diagonal matrix with its diagonal elements from the $N$-vector $\a$.

\section{Sparse iterative covariance learning}  \label{sec:CL-IC}

Our objective in this section is  to find a minimizer of  \eqref{eq:SML}  
under the constraint that $\gam$ is $K$-sparse. We consider two forms of sparsity: 
\begin{itemize} 
\item {\bf $K$-sparsity} signifies the search for $K$-largest elements   of $\gam$. 
\item {\bf $K$-peaksparsity}  signifies the search for $K$-largest peaks of $\gam$. 
\end{itemize} 
$K$-sparsity is the conventional form of sparsity occurring in SSR problems while peaksparsity appears in source localization problems, 
e.g., spectral line estimation (SLE) \cite{stoica_moses:1997}  of narrowband sources or SSR-based \emph{Direction-of-Arrival (DOA)} estimation \cite{malioutov2005sparse}  in sensor array processing. 
In the latter case, atoms $\a_i$ are steering vectors at candidate DOAs chosen from a predefined grid in the angle space.  Atoms that reveal high mutual coherence with the steering vector of the true DOA, $\a(\theta_k)$,  form a cluster with similar signal powers.  
The extent of coherence  depends on the density of the angle grid ($M$), the proximity of these clusters (i.e., closeness of true source DOAs),  as well as the sample size ($L$). 
In the case of SLE, off-grid frequencies resulting from sampling lead to the leakage of energy from component signals into adjacent frequency bins, resulting to sidelobes surrounding the frequency of the sources.

\subsection{Fixed-point iteration of $\gam$} \label{subsec:FPiter}

Using
\begin{align}
  &\frac{\partial \M^{-1}}{\partial \gamma_i} = -\M^{-1} \, \frac{\partial \M}{\partial \gamma_i}  \, \M^{-1} 
  = -\M^{-1} \a_i^{}\a_i^\hop \M^{-1}, \label{eq:deriv-inverse}%
\\
  & \frac{\partial\log |\M|}{\partial\gamma_i} 
 = \mathrm{tr} \left({\M^{-1} \frac{\partial \M}{\partial\gamma_i}} \right) 
  =   \a_i^\hop   \M^{-1} \a_i , \\ 
&\frac{\partial \M^{-1}}{\partial \sigma^2} = - \M^{-2}  \quad \mbox{and} \quad  \frac{\partial\log |\M|}{\partial\sigma^2} = \tr(\M^{-1})
\end{align}
the \emph{likelihood equation  (l.e.)} w.r.t. co-ordinates  $\gamma_i$  and $\sigma^2$ are obtained by setting the derivative of LLF $\ell$ in \eqref{eq:SML}  equal to zero  
\begin{align}
0 &=\frac{\partial \ell (\gam,\sigmaw) }{\partial\gamma_i}  = -  \a_i^\hop \M^{-1}  \SCM \M^{-1} \a_i + \a_i^\hop \M^{-1} \a_i    \label{eq:lik1} \\
& \qquad  \qquad \mbox{for }  i = 1,\ldots, M,  \notag \\
\mathcolor{black}{
 0 &= \frac{\partial \ell (\gam,\sigmaw) }{\partial\sigma^2}  = - \tr(\M^{-1}(\hat \M- \M) \M^{-1})} .  \label{eq:lik2}
\end{align} 
Define 
\begin{align} 
\M_{\backslash i} &= \sum_{j \neq i}  \gamma_j   \a_j \a_j^\hop + \sigmaw \mathbf{I}  = \M - \gamma_i \a_i \a_i^{\hop} \label{eq:Mbackslash} \\ 
\mathcolor{black}{\M_{\backslash \sigma^2} &=   \A \diag(\gam) \A^{\hop}  = \M - \sigma^2  \mathbf{I} }
\end{align}
as the covariance matrix of  $\y_l$-s when the contribution from the $i^{\text{th}}$ source signal and noise is removed, respectively. 

Next note that 
\begin{align} 
\a^{\sf H}_i \M^{-1} \b  
&=  \frac{ \a_i^{\hop} \M^{-1}_{\backslash i} \b}{1+\gamma_i \a^{\sf H}_i \M^{-1}_{\backslash i} \a_i}  \label{eq:lemma_apu}    \\
\a^{\sf H}_i \M^{-1}_{\backslash i}  \b  &=   \frac{ \a_i^{\hop} \M^{-1} \b}{1-\gamma_i \a^{\sf H}_i \M^{-1} \a_i} \label{eq:SMformula2}  
\end{align} 
which follows by applying the Sherman-Morrison formula 
\beq \label{eq:SMformula}
(\A + \mathbf{u} \mathbf{v}^\hop)^{-1}=\A^{-1} - \A^{-1} \mathbf{u} \mathbf{v}^\hop \A^{-1}/(1 + \mathbf{v}^\hop \A^{-1} \mathbf{u})
\eeq
to $\M_{\backslash i} + \gamma_i \a_i \a_i^{\hop}$ and $\M - \gamma_i \a_i \a_i^{\hop}$ in equations \eqref{eq:lemma_apu} and  \eqref{eq:SMformula2}, respectively.  Note that these equations  hold for any $N$-vector $\b$.
Next lemma derives the FP equation for $(\gam,\sigma^2)$ parameter. These equations can be traced back to \cite{yardibi2010source,abeida2012iterative}.  To make paper self-contained,  and for allowing better discussion of the algorithm that stems from it, we provide also the proof.  

\begin{lemma} \cite{yardibi2010source,abeida2012iterative} \label{lem} The solution to l.e. in \eqref{eq:lik1}-\eqref{eq:lik2}  can be expressed as a fixed-point equation: 
\beq \label{eq:lik_ver2}
\bmat \gam  \\ \sigma^2 \emat  = \bmat  H(\gam,\sigma^2)  \\  G(\gam,\sigma^2) \emat 
\eeq
where $i^{\text{th}}$  coordinate function of $H$ is 
 \begin{align} 
 H_i(\gam,\sigmaw) 
&=  \frac{\a_i^\hop \M^{-1} \SCM \M^{-1} \a_i }{(\a_i^\hop \M^{-1} \a_i)^2 } -  \frac{1}{ \a_i^{\hop} \M^{-1}_{\backslash i}  \a_i}  , \label{eq:gamma_i_sol1}
\end{align} 
and where
 \beq   \label{eq:gamma_i_sol2}
  \frac{1}{ \a_i^{\hop} \M^{-1}_{\backslash i}  \a_i}  = \frac{1}{\a_i^\hop \M^{-1} \a_i} -  \gamma_i .
\eeq 
The function $G$ is defined by
\beq \label{eq:Gfp}
G(\gam,\sigma^2) = \frac{\tr(\M^{-1}(\hat \M - \M_{\setminus \sigma^2}) \M^{-1}) }{\tr(\M^{-2})}.
 \eeq 
\end{lemma} 

\begin{proof}  
In \cite[Eq. (27)]{yardibi2010source}  it was shown that 
\begin{align}
H_i(\gam,\sigmaw)  &= \frac{ \a_i^{\sf H} \M^{-1}_{\backslash i}  ( \SCM  - \M_{\backslash i}) \M^{-1}_{\backslash i} \a_i}{  (\a_i^{\sf H} \M^{-1}_{\backslash i}  \a_i)^2} \label{eq:gamma_i_soll}\\ 
 &= \frac{ \a_i^{\sf H} \M^{-1}_{\backslash i}   \SCM \M^{-1}_{\backslash i} \a_i}{  (\a_i^{\sf H} \M^{-1}_{\backslash i}  \a_i)^2} - \frac{1}{ \a_i^{\sf H} \M^{-1}_{\backslash i}   \a_i}  \label{eq:gamma_i_sol0}
\end{align} 
 Applying \eqref{eq:SMformula2} to the first term in \eqref{eq:gamma_i_sol0} yields after some simple algebra the stated equation \eqref{eq:gamma_i_sol1}. 
Next notice that by taking inverses of both sides of \eqref{eq:SMformula2} yields the identity
 \beq \label{eq:SMformula3}
\frac{1}{\a^{\sf H}_i \M^{-1}_{\backslash i}  \b}  =   \frac{1}{\a_i^{\hop} \M^{-1} \b} - \gamma_i . 
 \eeq 
Then substituting $\b=\a_i$ in \eqref{eq:SMformula3} yields the identity \eqref{eq:gamma_i_sol2}.  Note that  combining \eqref{eq:gamma_i_sol1} and \eqref{eq:gamma_i_sol2} 
corresponds to \cite[Eq. (31)]{yardibi2010source}.  Noting that $\hat \M - \M = ( \hat \M - \M_{\setminus \sigma^2}) - \sigma^2 \mathbf{I}$ allows to write down the l.e. \eqref{eq:lik2} as 
\[
 0  =  -\tr(\M^{-1}(\hat \M - \M_{\setminus \sigma^2})\M^{-1}   +  \sigma^2 \tr(\M^{-2}) 
 \]
 and solving $\sigma^2$ from this equation yields \eqref{eq:Gfp}. Note that \eqref{eq:Gfp}  corresponds to  \cite[eq. (42)]{abeida2012iterative} but  just expressed in a slightly different form. 
\end{proof} 

The problem with l.e. \eqref{eq:lik_ver2} is that the updates may become negative. To guarantees the non-negativity of signal power updates, one may use update rule:
\begin{align}  
\gamma_i^{(\iter+1)} &= H_i(\gam^{(\iter)},\sigma^{2(\iter)} )  \notag \\ 
&=  \frac{\a_i^\hop \Minv^{(\iter)} \SCM \Minv^{(\iter)} \a_i }{(\a_i^\hop \Minv^{(\iter)} \a_i)^2 } - \left(  \frac{1}{\a_i^\hop \Minv^{(\iter)} \a_i} - \gamma_i^{(\iter)}  \right)_+ \label{eq:gamma_update}  \\
&= \gamma_i^{(\iter)}  +  \frac{\a_i^\hop \Minv^{(\iter)} \SCM \Minv^{(\iter)} \a_i }{(\a_i^\hop \Minv^{(\iter)} \a_i)^2 } - \max\left( \frac{1}{\a_i^\hop \Minv^{(\iter)} \a_i},   \gamma_i^{(\iter)} \right) \notag
\end{align}
for each $i \in [\![ M ]\!] $, where  
\[
\Minv^{(\iter)} =  (\A \Gam^{(\iter)} \A^\hop + \sigma^{2(\iter)}\mathbf{I})^{-1}, 
\]
with $\Gam^{(\iter)} = \diag(\gam^{(\iter)})$, denotes the update of the inverse covariance matrix $\Minv=\M^{-1}$ at $\iter^{th}$ iteration. Note that  
 $\a_i^{\sf H} \M^{-1}_{\backslash i}  \a_i \geq 0$ and $( \cdot)_+$ in \eqref{eq:gamma_update} guarantees non-negativity of the updates.  It is important to notice that 
 the update is different from $ \gamma_i^{(\iter+1)} =  [H_i(\gam^{(\iter)},\sigma^{2(\iter)} )]_+$ which is suggested in \cite[Result~1]{abeida2012iterative}. 
 This update appears to clip many of the source signal power  coordinates to zero already in the beginning of iterations, and in our experiments, performed worse than \eqref{eq:gamma_update} especially for $M > N$ and $M \gg K$ which is usually the case. 

One can express  \eqref{eq:gamma_update} in the form 
\beq \label{eq:gamma_update2} 
\gamma_i^{(\iter+1)} = \begin{cases}  \gamma_i^{(\iter)}  + d_i^{(\iter)}  & \mbox{, if $\gamma^{(\iter)}_i <  \frac{1}{\a_i^\hop \Minv^{(\iter)} \a_i}$  } \\ 
  \frac{\a_i^\hop \Minv^{(\iter)} \SCM \Minv^{(\iter)} \a_i }{(\a_i^\hop \Minv^{(\iter)} \a_i)^2 }  &\mbox{, if $\gamma^{(\iter)}_i  \geq  \frac{1}{\a_i^\hop \Minv^{(\iter)} \a_i}$} 
  \end{cases} 
\eeq 
where 
\beq \label{eq:di_iter}
d_i^{(\iter)} = \frac{\a_i^\hop \Minv^{(\iter)} \S \Minv^{(\iter)} \a_i }{(\a_i^\hop \Minv^{(\iter)} \a_i)^2 }  -   \frac{1}{\a_i^\hop \Minv^{(\iter)} \a_i} . 
\eeq 
Hence the FP update \eqref{eq:gamma_update2}  effectively makes a decision on either to use a coordinate descent update or a power update.  
Observe from l.e. \eqref{eq:lik1}  that 
\[
- \nabla_{\gamma_i}  \ell(\gam^{(\iter)}, \sigma^{2(\iter)})  = \a_i^\hop \Minv^{(\iter)} \S \Minv^{(\iter)} \a_i  - \a_i^\hop \Minv^{(\iter)} \a_i
\]
and thus $d_i^{(\iter)}$  in \eqref{eq:di_iter} can be written in the form
\[
d^{(\iter)}_i  =  \mu_{\iter}  \times -  \nabla_{\gamma_i} \ell(\gam^{(\iter)}, \sigma^{2(\iter)})
\]
 where $\mu_{\iter}  = (\a_i^\hop \Minv^{(\iter)} \a_i)^{-2}$ can be interpreted as adaptive step size. Thus, if  $\gamma^{(\iter)}_i < (\a_i^\hop \Minv^{(\iter)} \a_i)^{-1}$,  then \eqref{eq:gamma_update2} takes step  towards the negative gradient  with stepsize $\mu_{\iter}$. 

 We can find resemblances between the  FP update and at least three distinct methods,  each derived from different perspectives. 
 First, the  update 
\beq  \label{eq:update_SAMV2}
\gamma_i^{(\iter+1)} = \gamma_i^{(\iter)}\left[  \frac{\a_i^\hop \Minv^{(\iter)} \S \Minv^{(\iter)} \a_i }{(\a_i^\hop \Minv^{(\iter)} \a_i) } \right]^b
\eeq
for $b=1$ is used in  SAMV2  \cite{abeida2012iterative}  and derived by minimizing asymptotically minimum variance (AMV) criterion. Furthermore, the SBL-variant  in \cite{nannuru2019sparse}, termed SBL1  is based on $b=1/2$.  These rules  do not necessarily converge to a local minimum of LLF $\ell$ as they are based on an approximation of the  l.e.

Second, the coordinatewise optimization (CWO) of the LLF in 
\eqref{eq:SML} (for known $\sigma^2$)  
 is proposed  in \cite{haghighatshoar2018improved,fengler2021non}. To elucidate the concept, let us 
define a scalar function $\ell_i(d)=\ell(\gam+ d \mathbf{1}_i,\sigma^2 \mid \Y, \A)$, with  $\ell$ in \eqref{eq:SML}, and $\mathbf{1}_i$
denotes the $i^{\text{th}}$ canonical basis vector with a single $1$ at its $i^{\text{th}}$ coordinate and $0$s elsewhere. Then it can be shown that \cite{fengler2021non}:
\[
d_i = \arg \min_d   \ell_i (d) = \frac{\a_i^\hop \Minv \S \Minv \a_i }{(\a_i^\hop \Minv \a_i)^2 }  -   \frac{1}{\a_i^\hop \Minv \a_i} .
\]
Since signal powers $\gamma_i$ are non-negative scalars, the coordinate descent update rule is \cite[Eq. (30)]{yardibi2010source}, \cite{fengler2021non}:
\beq \label{eq:CWO}
\gamma_i \gets  \gamma_i +  \max(d_i, - \gamma_i),
\eeq
which bears similarity to the FP update in  \eqref{eq:gamma_update}. 

Third, we note that the signal power update in IAA algorithm, as presented in \cite[Table~II]{yardibi2010source}, can be expressed in the following form:
\beq \label{eq:gamma_IAA} 
\gamma_i^{(\iter+1)} =  \frac{\a_i^\hop \Minv^{(\iter)} \SCM \Minv^{(\iter)} \a_i }{(\a_i^\hop \Minv^{(\iter)} \a_i)^2 } \mbox{, $i = 1,\ldots, M$}. 
 \eeq 
 Despite originating from a different starting point, namely as the minimizer of a weighted least squares cost function,  this update demonstrates the most similarity to our FP rule.
 When comparing with the FP update rule  \eqref{eq:gamma_update2}, we anticipate that IAA signal power estimate are likely to converge to the same FP, especially when snapshot size is reasonably large.

\subsection{Update of the noise power $\sigmaw$}

A fixed point update rule for noise power  $\sigma^2$  based on \eqref{eq:lik_ver2} however does not lead to convergent FP rule as non-negativity of the update is hard to maintain. This was also noted in \cite{abeida2012iterative} which lead to development of modified rules (SAMV1-SAMV3).   We have observed that  the FP update for $\sigma^2$  (so of the form $\sigma^{2(\iter+1)} = {G}(\gam^{(\iter)},\sigma^{2(\iter)}$) based on Lemma~1 yields  a negative value for the FP iterate often after only 1 iteration whenever $M > N$. Thus alternative approach for updating $\sigma^2$ needs to be looked for. 

Before addressing the estimation of noise variance $\sigmaw$, we remind the reader about the following result  from   \cite{stoica1995concentrated}.

\begin{lemma}    \cite{stoica1995concentrated} \label{lem:bresler}
Assume $K< N$. Then the parameters  $\sigma^2>0$ and $\gam \in \R^K_{\geq 0}$ that minimize $\ell( \gam, \sigma^2 |  \mathbf{Y}, \A)$ in \eqref{eq:SML} are 
\begin{align} 
\hat \sigma^2 = \frac{1}{N- K }\mathrm{tr} \big(  (\mathbf{I} - \A \A^+)  \SCM  \big)  \label{eq:hatsigma2} 
\end{align} 
and 
\beq \label{eq:hat_gam}
\hat \gam =  \diag(\A^+ ( \S -  \hat \sigma^2 \mathbf{I})  \A^{+\hop})
\eeq 
where the latter represents an unconstrained MLE, meaning it aligns with MLE when the non-negativity constraint is met,  i.e., $\hat{\gam} \in  \R^K_{\geq 0}$. 
\end{lemma} 

\begin{remark} \label{rem:bresler} 
In scenarios where $\hat{\gam}$ in \eqref{eq:hat_gam} contains negative elements, one can calculate the constrained solution using \cite[Algorithm I]{bresler1988maximum}.  
A more straightforward method involves setting the negative elements of $\hat{\gam}$ to zero. While this approach does not precisely yield the MLE, it remains consistent in large samples.
\end{remark}

In our iterative algorithm, an update for $\sigmaw$ is  computed in two steps: 
\begin{enumerate} 
\item[1)] Identify the support as $\mathcal M^{(\iter)}  = \supp(H_K(\gam^{(\iter)}, \mathsf{peak}))$ 
\item[2)]  Update $\sigmaw$ as  minimizer of $\ell(\gam,\sigmaw \mid  \Y, \A_{\mathcal M^{(\iter)}} )$:    
\begin{align} \label{eq:update_noise}
\hat \sigma^{2(\iter)} &= \frac{1}{N- | \mathcal M^{(t)} | }\mathrm{tr} \big(  (\mathbf{I} - \A_{\mathcal{M}^{(\iter)}}\A_{\mathcal{M}^{(\iter)}}^+)  \SCM  \big) 
\end{align}
which follows from Lemma~\ref{lem:bresler}. 
\end{enumerate} 
Some other algorithms  use different update rules for noise power.  For instance, SAMV2 \cite{abeida2012iterative}  uses 
\[
\sigma^{2(\iter +1)} =   \frac{\tr \big( (\Minv^{(\iter)})^2  \S\big)}{\tr( (\Minv^{(\iter)})^2 )}.
\]
while SBL variants of  \cite{liu2012efficient} and   \cite{gerstoft2016multisnapshot,nannuru2019sparse} use the update \eqref{eq:update_noise}.

\subsection{BCD algorithm} 

The pseudo-code for block-coordinate descent (BCD) algorithm is provided  in \autoref{alg:SML}.   As name suggest the updates the parameter pair  $(\gam,\sigma^2)$ in blocks.  First,  in line~2 an update of  $\boldsymbol{\gamma}$ is calculated using  FP update \eqref{eq:gamma_update}. This is followed by update of the noise variance  $\sigmaw$ in line~4 using  \eqref{eq:update_noise} while taking into account either the  $K$-sparsity or $K$-peaksparsity of $\gam$ (line~3). Then, if the termination condition (line 5) is not met, one updates the inverse covariance matrix (line 6) and continue iterations until convergence.  
As the termination condition, we use  the relative error  
\beq  \label{eq:terminate_iters}
\| \gam^{\text{new}} - \gam^{\text{old}} \|_{\infty} / \| \gam^{\text{new}} \|_{\infty}<  \epsilon
\eeq  where  $\epsilon$ is a tolerance threshold, e.g., $\epsilon= 0.5e^{-4}$. The algorithm also takes as its input the Boolean variable \textsf{peak} which is set as \textsf{true} when  $K$-peaksparsity is assumed and \textsf{false} when $K$-sparsity is assumed.

\begin{algorithm}[!h]
 \caption{\textsf{CL-BCD}  algorithm}\label{alg:SML}
\DontPrintSemicolon
\SetKwInOut{Input}{Input} 
\SetKwInOut{Output}{Output}
\SetKwInOut{Init}{Initialize}
\SetNlSkip{1em}
\SetInd{0.5em}{0.5em}
\Input{$\mathbf{Y}$, $\A$, $K$, \textsf{peak}} 
\Init{$\S = L^{-1} \mathbf{Y} \mathbf{Y}^\hop$, $\gam = \mathbf{0}$, 
$\Minv = [\pdim/\tr(\S)] \mathbf{I}$} 

  \For{$\iter =1,\ldots,I_{max}$}{

\BlankLine

$\gam \gets (\gamma_i)_{M \times 1}$, $\gamma_i  \gets \frac{\a_i^\hop \Minv \SCM \Minv \a_i }{(\a_i^\hop \Minv \a_i)^2 } + \left( \gamma_i  -   \frac{1}{\a_i^\hop \Minv \a_i} \right)_+$

$ \mathcal{M}  \gets \supp \big( H_K(\gam,\mathsf{peak}) \big)$. 

$\sigmaw \gets  \frac{1}{N-K}\mathrm{tr} \left(  (\mathbf{I} - \A_{\mathcal{M}} \A_{\mathcal{M}}^+)  \S  \right) $ 

\If{ termination condition is met (see text)}{\nonl \textsf{break} } 

$\Minv \gets (\A \diag(\gam) \A^\hop + \sigmaw \mathbf{I})^{-1}$
  
   }
    \Output{$\mathcal{M}$, $ \boldsymbol{\gamma}$,  $\sigmaw$} 

\end{algorithm}

 One benefit of the  FP algorithm is that $L > \pdim$ is not required  but the algorithms can be applied even in a single snapshot ($L=1$) situation and an initial estimate of signal power  vector $\gam$ is not required. 
 Additionally, a small threshold can be defined so that when any $\gamma_i$ becomes sufficiently small (e.g., $10^{-12}$), the corresponding atom $\a_i$ is pruned from the model.  This can be useful when $M$ is excessively large.

\section{Greedy  pursuit covariance learning} \label{sec:OMP} 

We first recall the following result from \cite{faul2001analysis} (single measurement vector (SMV) case: $L=1$) and \cite{yardibi2010source} (MMV case, $L>1$). 

\begin{lemma}  \cite{faul2001analysis}, \cite[Appendix~B]{yardibi2010source} \label{lem:cond_lik} Consider the conditional likelihood of \eqref{eq:SML} where source powers  $\gamma_j$ for $j \neq i$ and the noise variance $\sigma^2$ are known. Then the conditional  negative log-likelihood function for the unknown $i^{\text{th}}$ source, defined as 
\begin{align} \label{eq:epsilon_i_v0}
\ell_i(\gamma &\mid \Y, \A, \{ \gamma_{j} \}_{j \neq i} , \sigma^2 )  \notag \\ 
&=  \tr( (\M_{\setminus i} + \gamma \a_i \a_i^\hop )^{-1}  \S) + \log | \M_{\setminus i} + \gamma \a_i \a_i^\hop | , 
\end{align}
has a unique optimal value 
\beq \label{eq:gamma_i_star}
\gamma_i  =  \max\left( \frac{\a_i^\hop \M_{\setminus i}^{-1}  (\S - \M_{\setminus i}) \M_{\setminus i}^{-1} \a_i }{(\a_i^\hop \M_{\setminus i}^{-1} \a_i)^2 } , 0\right), 
\eeq 
where  $\M_{\setminus i}$ is defined in \eqref{eq:Mbackslash}. 
\end{lemma}

Note that a co-ordinatewise algorithm that is constructed using iterative updates  \eqref{eq:gamma_i_star} for each $\gamma_i$ in turn is equivalent with previously cited CWO algorithm using updates \eqref{eq:CWO}. This equality follows by simply invoking matrix inversion lemma in \eqref{eq:gamma_i_star} \cite[Appendix~B]{yardibi2010source}.
Lemma~\ref{lem:cond_lik} will serve as the foundation for developing the CL-based Orthogonal Matching Pursuit (CL-OMP) algorithm, whose pseudocode is presented in \autoref{alg:SML3}. This algorithm follows a matching pursuit strategy similar to the one outlined in \cite[Table~3.1]{sparse_book:2010}.

{\bf Initialization} phase. Set $k=0$, and  $\gam^{(0)}= \mathbf{0}_{M \times 1}$, $ \sigma^{2(0)}= [\tr(\S)/\pdim ]\mathbf{I}$,  and  $\mathcal M^{(0)}= \mathsf{supp}(\gam^{(0)})=\emptyset$ as initial solutions of signal and noise powers and the signal support, respectively. Then $\M^{(0)} =  \A \diag(\gam^{(0)}) \A^\hop + \sigma^{2(0)} \mathbf{I} = \sigma^{2(0)} \mathbf{I} $ is the initial covariance matrix at the start of iterations.  

{\bf Main Iteration} phase consists of the following steps: 

 {\it 1) Sweep:} Compute the errors 
\begin{align} \label{eq:epsilon_i_v0}
\epsilon_i  &=  \min_{\gamma \geq 0} \ell_i(\gamma  \mid \Y, \A, \{\gamma^{(k)}_j \}_{j \neq i}, \sigma^{2(k)})  
\end{align}
for each $i \in [\![ M]\!] \setminus \mathcal M^{(k)} $ 
using its unique optimal value as given by Lemma~\ref{lem:cond_lik}: 
\beq \label{eq:gamma_i_star2}
\gamma_i  =  \max\Big( \frac{\a_i^\hop (\M^{(k)})^{-1}  (\S - \M^{(k)})  (\M^{(k)})^{-1} \a_i }{(\a_i^\hop (\M^{(k)})^{-1} \a_i)^2 } , 0\Big). 
\eeq 
 The sweep stage is  comprised of  lines 2 and 3 in  \autoref{alg:SML3}. 

 {\it  2) Update support:}  Find a minimizer, $i_k$ of $\epsilon_i$: $\forall i  \not \in \mathcal M^{(k)}$, $\epsilon_{i_k} \leq \epsilon_i$, 
and update the support $\mathcal M^{(k+1)} = \mathcal M^{(k)} \cup \{ i_k\}$.  
This step  corresponds to line 4 in   \autoref{alg:SML3}.   

 {\it 3) Update provisional solution}: 
compute
 \[
 (\hat{\mathbf{g}}, \hat{\sigma}^{2})= \arg \min_{\mathbf{g}, \sigma^2 }  \ell(\mathbf{g},\sigma^{2}  \mid \Y, \A_{\mathcal M^{(k+1)}}), 
 \]
 where $\ell$ is the LLF defined in \eqref{eq:SML}.  These solutions are obtained from Lemma~\ref{lem:bresler} and calculated in lines $5-7$  in   \autoref{alg:SML3}, respectively. 
The obtained signal power is constrained to be non-negative, which is not exactly the MLE. Alternative option is to compute the true MLE (the constrained solution) using \cite[Algorithm I]{bresler1988maximum} as noted in Remark~\ref{rem:bresler}. 

{\it 4) Update the covariance matrix}: Compute $\M^{(k+1)}=  \A_{\mathcal M} \diag(\hat{\mathbf{g}})  \A^\hop_{\mathcal M} + \hat \sigma^{2} \mathbf{I}$.  This stage is implemented by line 8  in \autoref{alg:SML3}.

 {\it  5) Stopping rule}: stop after $K$ iterations, and otherwise increment $k$ by 1 and repeat steps 1)-4). 
It is worth noting that alternative stopping criteria can be employed, such as halting the process when the variance $ \sigma^{2(k)}$ falls below a predefined threshold. This criterion is particularly valuable in applications where the noise level can be accurately estimated or is known a priori but number of sources $K$ is unknown.

Note that the sweep stage gives  (after some simple algebra) the following value for the error in \eqref{eq:epsilon_i_v0}:
\beq \label{eq:epsilon_i}
\epsilon_i = c+ \log (1 + \gamma_i \a_i^\hop \M^{-1} \a_i )  - \gamma_i  \frac{\a_i^\hop \M^{-1}  \S  \M^{-1} \a_i }{1+ \gamma_i \a_i^\hop \M^{-1} \a_i}
\eeq 
where  we have for simplicity of notation written $\M = \M^{(k)}$ and where $c$ 
denotes an irrelevant constant that is not dependent on $\gamma_i$. Thus w.lo.g. we set $c=0$. 
Then using that $\gamma_i$  is given by  \eqref{eq:gamma_i_star2} we can write
\begin{align*}
 1+ \gamma_i \a_i^\hop \M^{-1} \a_i &= 1+  \frac{ \a_i^\hop \M^{-1}  (\S - \M)  \M^{-1} \a_i }{ (\a_i^\hop \M^{-1} \a_i)^2 } \a_i^\hop \M^{-1} \a_i   \\
 &= 1 + \frac{\a_i^\hop \M^{-1}  \S \M^{-1} \a_i  - \a_i^\hop \M^{-1} \a_i}{\a_i^\hop \M^{-1} \a_i} \\
 &= \frac{\a_i^\hop \M^{-1}  \S \M^{-1} \a_i}{\a_i^\hop \M^{-1} \a_i}  
 \end{align*} 
 in the case that $\gamma_i>0$. 
Substituting this into the denominator of the last term in \eqref{eq:epsilon_i}, we obtain 
\beq \label{eq:epsilon_i_v2}
\epsilon_i = \log (1 + \gamma_i \a_i^\hop \M^{-1} \a_i )  -  \gamma_i \a_i^\hop \M^{-1} \a_i .
\eeq
If $\gamma_i = 0$, then  $\epsilon_i=0$. This explains line 3 in   \autoref{alg:SML3}.   

In the provisional solution update,  one finds the minimizer of  $\ell(\gam,\sigma^{2}  \mid \Y, \A_{\mathcal M})$ 
  where for notational simplicity we wrote $\mathcal{M} = \mathcal{M}^{(k+1)}$ for the current support with $|\mathcal M | = k$.  Thus the problem to be solved is 
\begin{align*}
\underset{\gam \in \mathbb{R}^k_{\geq 0},\sigma^2>0}{\mbox{minimize}}  &\tr(  (\A_{\mathcal M} \diag(\gam) \A^\hop_{\mathcal M}  + \sigma^2 \mathbf{I})^{-1} \SCM )  \\ 
&\qquad +  \log |  \A_{\mathcal M}  \diag(\gam) \A^\hop_{\mathcal M}  + \sigma^2 \mathbf{I} |  .
 \end{align*} 
 As already noted, taking the derivative of this equation w.r.t. $\gamma_i$ and setting it to zero, gives the l.e.: 
\[
0 =   \a_i^\hop \M^{-1}  (\S - \M) \M^{-1} \a_i  \quad \forall i \in \mathcal M. 
\]
This in turn implies that in the next iteration $\gamma_i$ in \eqref{eq:gamma_i_star2} 
will take value $\gamma_i \approx 0$ for $i \in \mathcal M$. Thus unlikely these columns will be chosen again for the support in the next iterations. This explains the name of this  greedy pursuit algorithm: the method is similar to  conventional OMP \cite{davis1997adaptive} or SOMP \cite{tropp2006algorithms}  where no atom is ever chosen twice in the sweep stage.

\begin{algorithm}[!h]
 \caption{\textsf{CL-OMP} algorithm}\label{alg:SML3}
\DontPrintSemicolon
\SetKwInOut{Input}{Input} 
\SetKwInOut{Output}{Output}
\SetKwInOut{Init}{Initialize}
\SetNlSkip{1em}
\SetInd{0.5em}{0.5em}
\Input{$\mathbf{Y}$, $\A$, $K$} 
\Init{$\S = L^{-1} \mathbf{Y} \mathbf{Y}^\hop$,  $\M=[\tr(\S)/p] \mathbf{I}$, $\mathcal M =  \emptyset$}

 \For{$k =1,\ldots,M$}{

$ \boldsymbol{\gamma} =  (\gamma_i)_{M \times 1}$, $\gamma_i \gets \max\Big( \frac{\a_i^\hop \M^{-1}  (\S - \M)  \M^{-1} \a_i }{(\a_i^\hop \M^{-1} \a_i)^2 }, 0  \Big)$

\BlankLine
$\boldsymbol{\epsilon}=(\epsilon_i)   \gets  \big(  \log (1 + \gamma_i \a_i^\hop \M^{-1} \a_i )  -  \gamma_i \a_i^\hop \M^{-1} \a_i  \big)_{M \times 1}$

\BlankLine
$\mathcal M   \gets \mathcal M  \cup \{ i_k \}$ with $i_k  \gets   \arg \min_{i \not \in \mathcal M} \epsilon_i $

\BlankLine
$\sigma^2 \gets  \frac{1}{\pdim-k}\mathrm{tr} \big(  (\mathbf{I} - \A_{\mathcal{M}} \A_{\mathcal{M}}^+)  \S  \big) $

 \BlankLine
 $\gam_{\mathcal M}  \gets  \max\big(\diag\big(\A_{\mathcal{M}}^+ ( \S -  \sigma^2 \mathbf{I})  \A_{\mathcal{M}}^{+\hop} \big),0\big)$

 \BlankLine
 $\gam_{ \mathcal M^\complement}   \gets  \mathbf{0}$ 
 
 \BlankLine
$ \M \gets \A \diag(\gam) \A^\hop + \sigma^{2} \mathbf{I}$

\BlankLine
\If{stopping rule is met (see text)}{\nonl \textsf{break}}

   }
 \Output{$\mathcal{M}$, $ \boldsymbol{\gamma}$,  $\sigmaw$} 
\end{algorithm}

It is interesting to draw some parallel to a greedy strategy proposed in \cite{tipping2003fast}. The algorithm in \cite{tipping2003fast} was developed for the SMV case, 
with main motivation of  providing computationally lighter algorithm for solving the marginal likelihood of SBL \cite{tipping2001sparse}. 
Its generalization to MMV model is not straightforward. In particular, its benefits (computational speed) is hard to maintain  in  MMV model. 
The algorithm iteratively considers candidate atoms that are  to be added to (or deleted from) the current model in support $\mathcal M$ either in random or  in sequential order from $m=1,2,\ldots,M$. The order of the model (number of atoms included) can increase at each iteration, and requires continuous updates of statistics such as \cite[Eq. (23)-(24)]{tipping2003fast}. 
  Such addition/deletion approach can provide some benefits, albeit increase computation times. Any exploration of this approach is left for future work.

\section{Sparse signal recovery: simulation studies} \label{sec:SSR} 

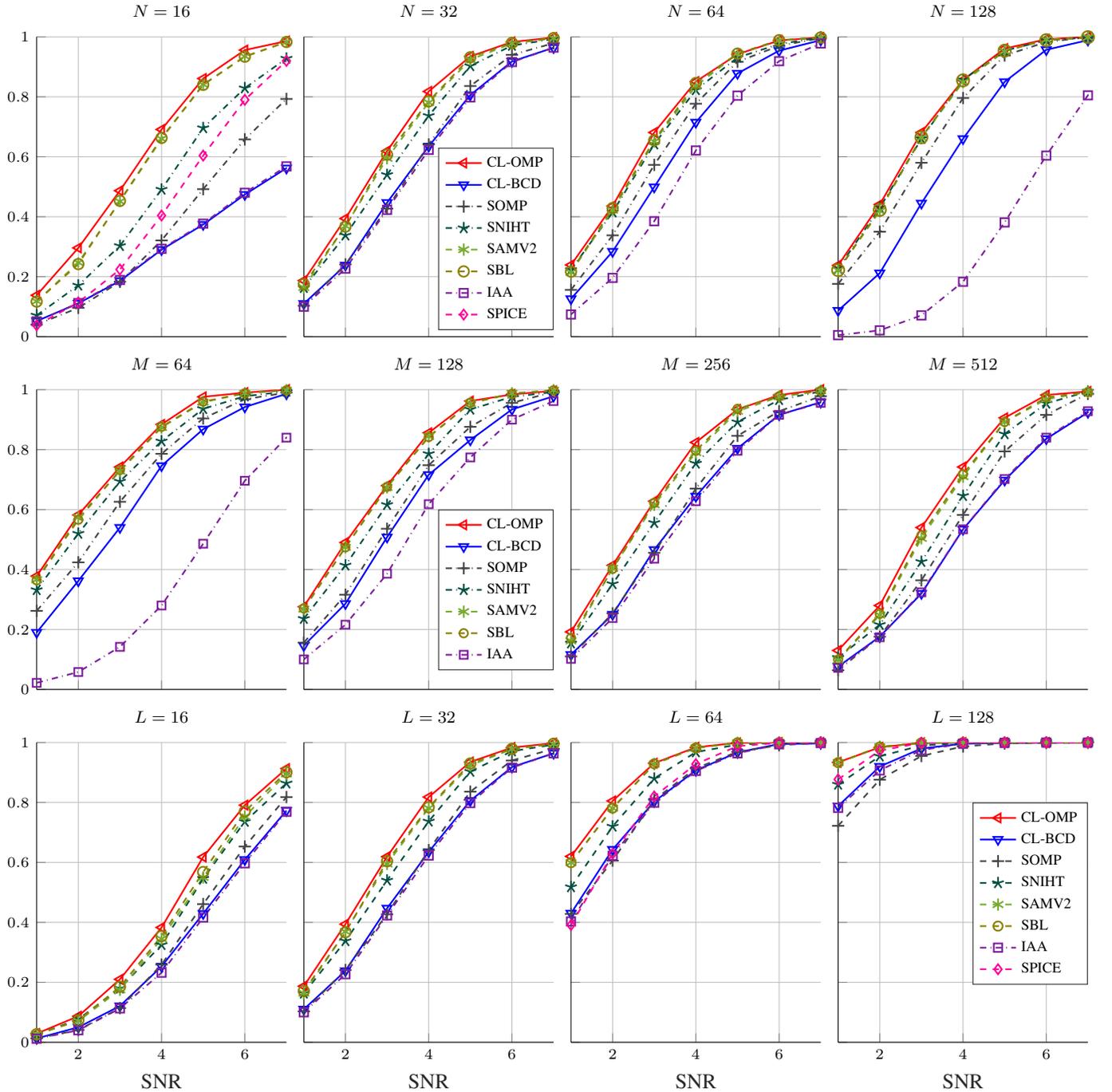
\begin{figure*}[!t]
\centerline{\input{tikz_new/SAMV_simul2_N=16to128_K=4_L=32_M=256to1000_NRSIM=1_GaussianS=49_1-May-2024.tex}}
\centerline{\input{tikz_new/SAMV_simul2_N=32_K=4_L=32_M=64to512_NRSIM=500_GaussianS=1_11-May-2024}}
\centerline{\input{tikz_new/SAMV_simul2_N=32_K=4_L=16to128_M=256to1000_NRSIM=1_GaussianS=49_2-May-2024.tex}}
\caption{PER rates vs. SNR Top: number of measurements varies from $N = 16$ to $N= 128$ ($L=32$, $M=256$);  Middle: grid size varies from $M=64$ to $M=512$ ($N=32$, $L=32$). Bottom: sample size increases from $L=16$ to $L=128$ ($N=32$, $M=256$).}  \label{fig:Ksparse_vs_REC}
\end{figure*}

\begin{figure*}[!t]
\centerline{\input{tikz_new/SAMV_simul2_cpu_N=16to128_K=4_L=32_M=256_NRSIM=500_GaussianS=1_04-Jun-2024.tex}}
\centerline{\input{tikz_new/SAMV_simul2_cpu_N=32_K=4_L=32_M=64to512_NRSIM=500_GaussianS=1_04-Jun-2024.tex}}
\caption{Running times vs. SNR Top: number of measurements varies from $N = 16$ to $N= 128$ ($L=32$, $M=256$);  Bottom: number of atoms varies from $M=64$ to $M=512$ ($N=32$, $L=32$).
}  \label{fig:runtime_vs_REC}
\end{figure*}
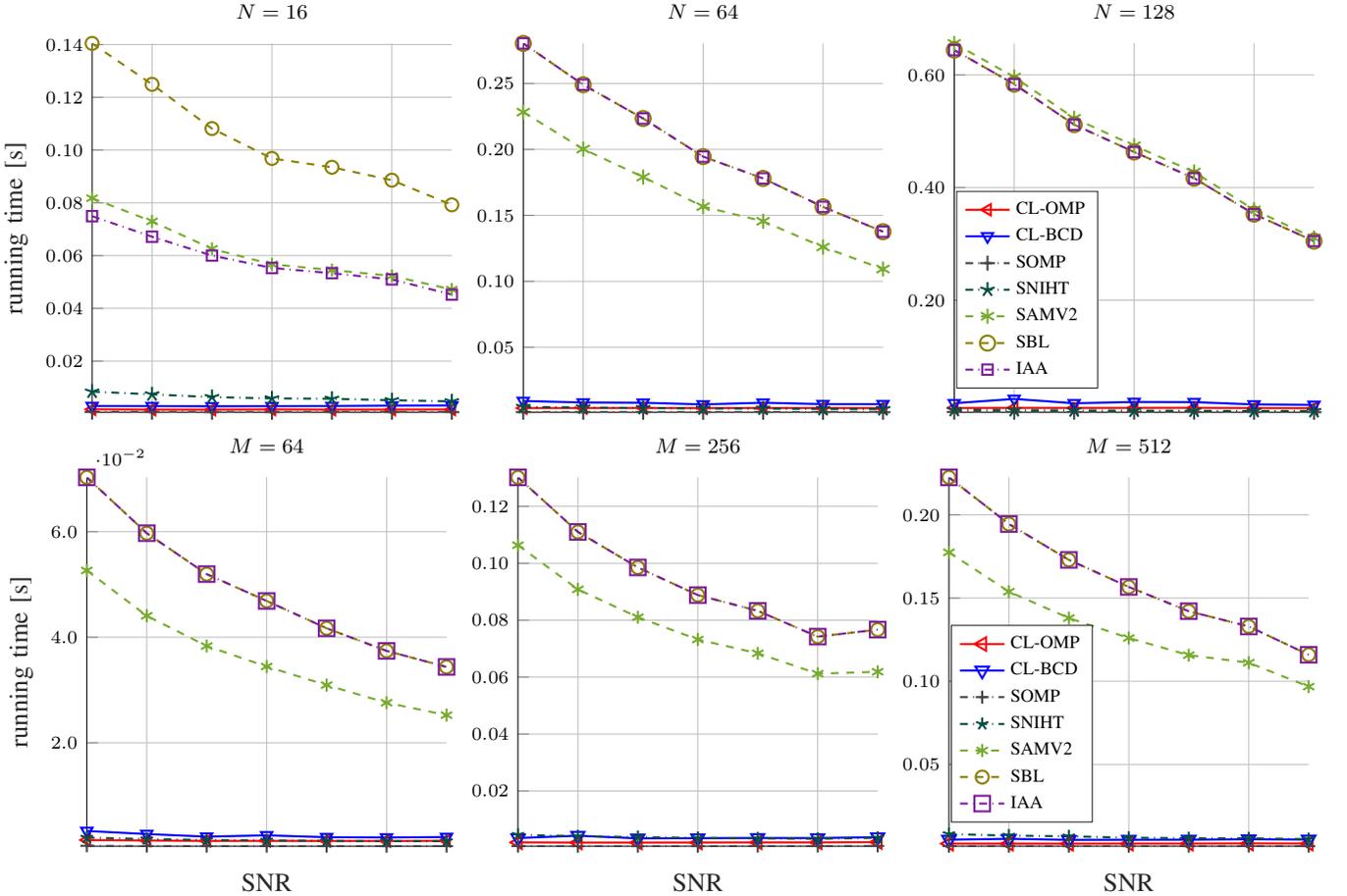

Next  we compare the sparse signal recovery performance of proposed methods to traditional greedy sparse signal recovery  algorithms,  
the {\bf SOMP} \cite[Algorithm~3.1]{tropp2006algorithms} and  the {\bf SNIHT}  \cite[Algorithm~1]{blanchard2014greedy}. 
The boolean variable \textsf{peak} is thus set to \textsf{false} as we assume $K$-sparsity in this example. 
Both SOMP and SNIHT algorithms are designed for MMV model, and return an estimated support $\mathcal M$ of  $K$-rowsparse signal matrix given the measurement matrix $\Y$, the dictionary $\A$,  and the desired rowsparsity level $K$. We also compare with  {\bf IAA}(-APES) \cite[Table~2]{yardibi2010source}, 
{\bf SAMV2} \cite[Table I]{abeida2012iterative},  {\bf SBL} variant of  \cite{nannuru2019sparse} (using $b=1$ in \eqref{eq:update_SAMV2} in signal power update). 
and (enhanced) {\bf SPICE}, also referred to as  SPICE+ in \cite[eq. (44)-(46)]{stoica2010spice}. 
 Note that SPICE assumes that $L> N$, and thus it is included in the study only  when this assumption holds.     
 For all  CL-methods (except for  greedy CL-OMP which terminates after $K$ iterations) we use \eqref{eq:terminate_iters} as the stopping criterion and set $500$ as maximum number of iterations. 
Matrix $\A$ is a Gaussian random measurement matrix, i.e.,  the elements of $\A$ are drawn from $\mathcal{CN}(0,1)$ distribution and the columns are unit-norm normalised as is common in compressed sensing.   To form the $\kdim$-rowsparse  source signal matrix $\X \in \mathbb{C}^{M \times L}$, support $\mathcal M =\mathrm{supp}(\X)$ is randomly chosen from $\{1,\ldots,M\}$   
without replacement for each Monte-Carlo (MC-)trial. The noise $\mathbf{E} \in \mathbb{C}^{N \times L} $ have elements that are  i.i.d. circular complex Gaussian with unit variance (i.e., $\sigma^2 = 1$).  

As performance measure we use  the  empirical \emph{probability of exact recovery},   
$
\mbox{PER}  = \frac{1}{T}\sum_{t=1}^T  \mathrm{I} \big(  \hat{\mathcal{M}}^{(t)} = \mathcal{M}^{(t)} \big), 
$
where $\mathrm{I}(\cdot)$ denotes the indicator function,  and $\hat{\mathcal{M}}^{(t)}$ 
denotes the  estimate of the true signal support $\mathcal M^{(t)}$  for $t^{\text{th}}$ MC trial. 
The number of MC trials is $T=500$  and the sparsity level is $\kdim=4$. 
Let $\mathcal M = \{i_1,\ldots, i_K\}$ be the true support set where $K=4$, and let  $\sigma^2_1= \gamma_{i_1}$ denote the  power of the $1^{\text{st}}$ Gaussian source signal.  
 Define the SNR of the first non-zero source sequence as $10 \log_{10} \sigma^2_1/\sigma^2$. For the second, third, and fourth Gaussian sources, we set their SNR levels to be 1 dB, 2 dB, and 4 dB lower than that of the first source, respectively.

{\bf PER vs dimensionality of the measurements  $N$}. Top panel of \autoref{fig:Ksparse_vs_REC} displays the PER rates  when $N$ varies  from $N = 16$ to $N= 128$ while $L$ and $M$ are fixed ($L=32$,  $M=256$). 
 As can be noted, CL-OMP has the best performance in all cases and SNR levels while having similar performance with SAMV2 and SBL when $N$ increases to $N=128$.  Curiously, the  performance of IAA  decays when the number of measurements increases: it offers clearly the worst   performance  in high dimension (large $N$). When comparing SNIHT and SOMP, it is observed that SNIHT outperforms SOMP, yet SOMP gradually improves, approaching the performance level of SNIHT as $N$ increases. While CL-BCD is not performing that well for small $N$, it also shows significant improvement as $N$ increases.

{\bf PER vs number of atoms $M$.}  The middle panel of \autoref{fig:Ksparse_vs_REC} displays the PER rates when $M$ increases  
from   $M=64$ to $M=512$ while $N$ and $L$ are fixed ($N=32$, $L=32$). The conclusions are similar as previously. CL-OMP again attains the top performance for all SNR levels while SAMV2 and SBL are the next best performing methods. The conventional SSR methods, SOMP and SNIHT, perform worse than these three CL methods. IAA  performs poorly for small $M$  but similarly with BCD when $M$ is large. 

{\bf PER vs number of measurements vectors $L$.}  The bottom panel of \autoref{fig:Ksparse_vs_REC} displays the PER rates as $L$ varies while $N$ and $M$ are fixed ($N=16$, $M=256$). 
Again it can be noted that CL-OMP  uniformly outperforms the other methods for all SNR levels. Overall CL-OMP is distinctively more robust to low SNR than conventional greedy methods, SOMP and SNIHT, especially when number of measurements vectors are limited.   
CL-BCD performs similarly to IAA in this setting. 

{\bf SNR versus running time.}  \autoref{fig:runtime_vs_REC} displays the running time of different algorithms in the first two settings when $M$ and $N$ varies, respectively. As can be noted, 
the greedy methods CL-OMP, SNIHT and SOMP are clearly the best performing methods. Their running time is also independent of the SNR unlike some of the iterative methods. CL-BCD also stands out in running time compared to other iterative algorithms. It has much faster convergence than other iterative algorithms (SAMV2, SBL and IAA).  Moreover,  running time of  SAMV2, SBL or IAA   increases steeply as SNR decreases while running time of CL-BCD appears unaffected by SNR.

\section{Source localization using covariance learning} \label{sec:source} 

\subsection{SSR-based DOA estimation} 

 In DOA (or source localization) estimation problem, the aim is  to estimate the DOA-s 
of narrowband sources that are present in the array's 
viewing field.  In SSR-based DOA estimation, one forms  a fine (equidistant)  {\bf grid}  that  covers the DOA space $\Theta = [-\pi/2,\pi/2)$: $-\pi/2 \leq \theta_1  < \cdots < \theta_M< \pi/2$, where the grid spacing $\Delta \theta = \theta_2-\theta_1$ determines the angular resolution.   The grid size $M$ is at least several hundreds in order to achieve high-resolution DOA estimates.  It is assumed that $K<N$ sources are present and  the  grid is fine enough ($M \gg K$) so that  true DOA parameters lie on or at least close to one of the grid points.   The array consists of $N$ sensors and let $\a(\theta) \in \mathbb{C}^N$ denote the array manifold  which  without any loss of generality is assumed to verify $\| \a(\theta) \|^2 = N$. We let  $\a_m = \a(\theta_m) \in \mathbb{C}^{N}$  denote  \emph{steering vector} corresponding to DOA $\theta_m$ on the grid. 

We assume to have $L$ array snapshots available. The $l^{\text{th}}$ snapshot $\y_l$ follows the model in \eqref{eq:MMV} where  $\A \in \mathbb{C}^{N \times M}$ contains steering vectors $\a_m$ as columns,  $\x_l$ contains the latent source signal waveforms  and $\e_l$  denotes the sensor noise vector. 
When represented in matrix form, the model follows  \eqref{eq:MMVmodel},  where the source signal matrix $\X$ is exactly $K$-rowsparse if the $K$ true sources are exactly on the grid.  
Exploiting sparsity of grid-based array model in DOA estimation was groud breaking  work by   \cite{malioutov2005sparse}. We refer the reader to \cite{yang2018sparse}   for a recent review. 

 If one assumes that snapshots are i.i.d.,  the noise vector $\e_l$ is spatially white ($\cov(\e_l)= \sigma^2  \mathbf{I}$)),    $\x_l$ and $\e_l$ are independent,  and source signals are uncorrelated, then  the covariance matrix of $\y_l$ has the same form as in   \eqref{eq:M}, with  $\Gam= \cov(\x_l)=\diag(\gamma_1,\ldots, \gamma_M)$, where {$\gamma_m = \var(x_{ml})$  is  the power of $m^{\text{th}}$ source signal.  Again if source DOAs are exactly on the grid, then $\gam=(\gamma_1,\ldots, \gamma_M)^\top$ is exactly $K$-sparse vector \cite{stoica2010spice,yardibi2010source}.  Under the assumption of Gaussian source signal and noise, the  negative LLF  in  \eqref{eq:SML} can be optimized for obtaining  estimate of the signal power $\gam$ and the noise $\sigma^2$.   Methods that follow this  CL-approach for DOA estimation have been developed in a series of papers, see \cite{wipf2007beamforming,stoica2010spice,yardibi2010source,abeida2012iterative,gerstoft2016multisnapshot,mecklenbrauker2022doa} to mention only a few. 

It is important to realise that even if the true DOAs are on-grid, the source power estimates $\hat{\gam}$ does not exhibit $K$-sparsity but $K$-peaksparsity.  Namely,  power from $i^{\text{th}}$ source is leaked or smeared to close by grids of the true DOA, and hence the estimate of the source power vector  $\gam$ is $K$-peaksparse. The width of the  peak around the true DOA will depend mainly on the SNR, the sample size $L$ and the grid size $M$.  
Hence the CL-algorithms, such as CL-BCD or  CL-OMP  will set the boolean variable \textsf{peak} as \textsf{true}  in order to locate $K$  significant peaks in  $\gam$.  Naturally, the same is true for other CL-algorithms (IAA, SPICE, SAMV, SBL). 

\subsection{Simulation}  \label{sec:source:simul}

\begin{figure}[!t]
\centerline{\input{tikz/onesource_snr_vs_MSEdoa_theta=-25.00_N=20_M=1801_L=25_LL=1000_nmethods=7_27-Nov-2023}}\vspace{-2pt}
\centerline{\hspace{6pt}\input{tikz/onesource_snr_vs_MSEgamma__theta=-25.00_N=20_M=1801_L=25_LL=1000_nmethods=7_27-Nov-2023.tex}}
\caption{RMSE of $\hat \theta$ vs SNR (top panel) and  root NMSE of $\hat \gamma$ vs SNR (bottom panel) in a single source case with DOA $\theta_1 = -25$;  $L=25$, $N=20$,  and $M=1801$ ($\Delta \theta=0.1^o$). }  \label{fig:1source_CRB_vs_SNR}
\end{figure}
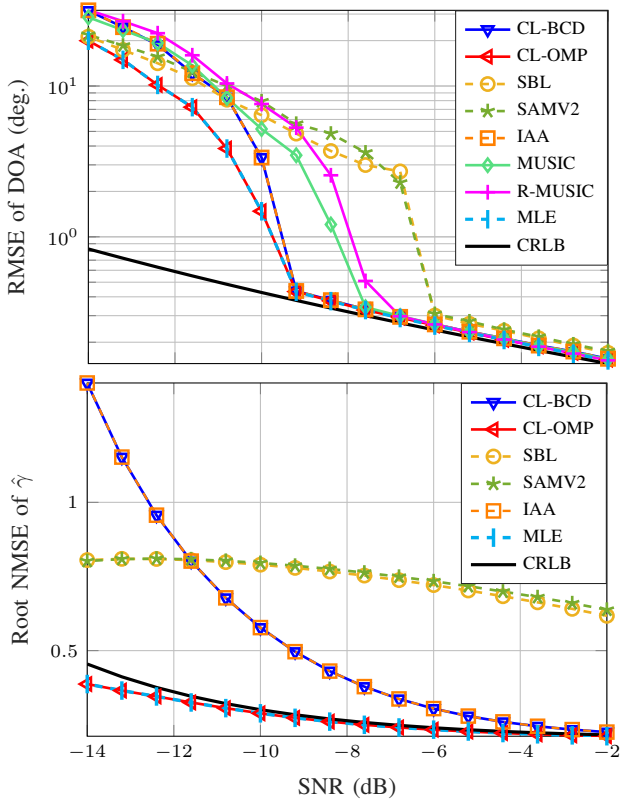

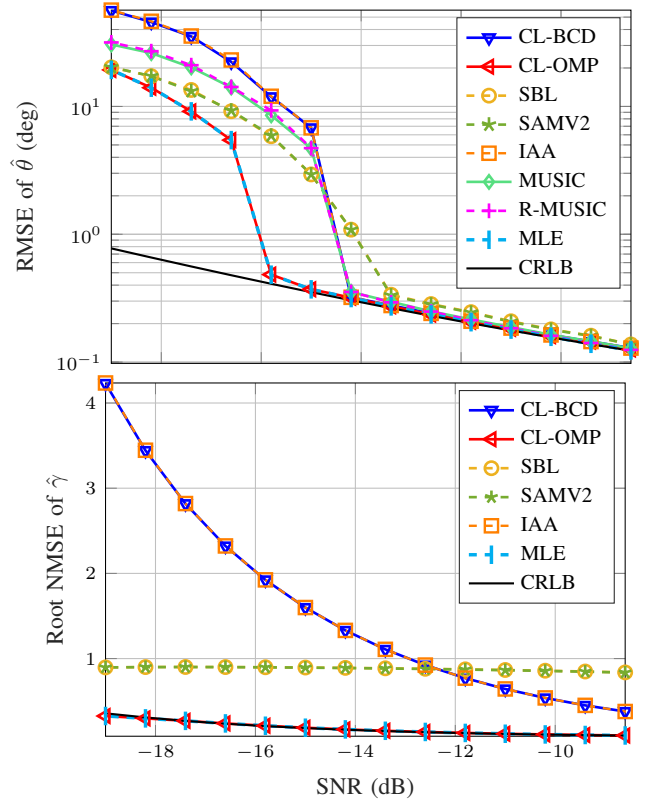
\begin{figure}[!t]
\centerline{\input{tikz/onesource_snr_vs_MSEdoa_theta=-25.00_N=20_M=1801_L=200_LL=1000_nmethods=7_27-Nov-2023}} \vspace{-2pt}
\centerline{\hspace{11pt}\input{tikz/onesource_snr_vs_MSEgamma__theta=-25.00_N=20_M=1801_L=200_LL=1000_nmethods=7_27-Nov-2023.tex}}
\caption{RMSE of $\hat \theta$ vs SNR (top panel) and root NMSE of $\hat \gamma$ vs SNR (bottom panel) in a single source case with DOA $\theta_1 = -25$;  $L=200$, $N=20$,  and $M=1801$ ($\Delta \theta=0.1^o$).} \label{fig:1source_CRB_vs_SNR_gamma}
\end{figure}

The simulation set-up is as follows. The array is Uniform Linear Array (ULA)  with half a wavelength interelement spacing and number of sensors is $N=20$. 
The grid size is $M=1801$, providing angular resolution  $\Delta \theta=0.1^o$. The number of MC trials is 1000. 
Let $\sigma^2_{s,k} =  \E[ |s_{k}|^2]$ denote the power of the $k^{\text{th}}$ source signal, and let   $\mathrm{SNR}_k = \sigma^2_{s,k}/\sigma^2$ denote the  SNR  of the  $k^{\text{th}}$ source. 
The SNR is defined as average of source SNR-s, i.e.,  $\mbox{SNR (dB)} = \frac{10}{K}  \sum_{k=1}^K  \log_{10} \mbox{SNR}_k$.  

We compare the proposed  CL-BCD and  CL-OMP against the Cram\'er-Rao lower bound and the  following CL-based methods: {\bf IAA}(-APES) \cite[Table~2]{yardibi2010source}, 
{\bf SAMV2} \cite[Table I]{abeida2012iterative}, 
and  {\bf SBL} variant of  \cite{nannuru2019sparse} (using $b=1$ in \eqref{eq:update_SAMV2} in signal power update).  For all methods (except for  greedy CL-OMP which terminates after $K$ iterations) we use \eqref{eq:terminate_iters} as the stopping criterion and set $500$ as maximum number of iterations.  We also compare with conventional high-resolution methods: {\bf MUSIC} \cite{schmidt:1986} and {\bf R-MUSIC} \cite{barabell1983improving} (aka Root-MUSIC).  We will seek the DOA peaks from MUSIC pseudospectrum using same DOA grid, and thus it has the same resolution as  CL-based algorithms.   In the case of a single source, we also compare with the {\bf MLE},  which can be obtained by solving
\beq \label{eq:DOA_MLE_1source}
\hat \theta_{\text{ML}} = \arg \max_{\theta} \,  \a(\theta)^\hop \S \a(\theta). 
\eeq
To find the maximizer in \eqref{eq:DOA_MLE_1source}, we sweep through a dense grid of 18,001 points, yielding a resolution of $0.01^o$.  
It is worthwhile to note that  SAMV2 uses the same signal power update as SBL. Thus we expect these two methods perform similarly in DOA estimation (i.e., identifying the support $\mathcal M$), while possibly having some differences in actual estimates of signal powers. As discussed in \autoref{subsec:FPiter}, the signal power update of IAA is similar to FP update \eqref{eq:gamma_update2} and thus these two methods are expected to perform similarly in identifying the support $\mathcal M$. There are also many grid-free methods, such as gridless SPICE \cite{yang2015gridless}, atomic norm minimization approaches \cite{wagner2021gridless}, etc.  Grid-free methods have their own benefits and disadvantages. In this study we only compare with gridded methods and use fine grid for better illustrating the attainable accuracy of the methods.

In the first simulation setting, we have a single ($K=1$) source at $\theta=-25^o$. 
\autoref{fig:1source_CRB_vs_SNR} and \autoref{fig:1source_CRB_vs_SNR_gamma} display the performance of DOA and signal power estimation vs the SNR in terms of root MSE (RMSE)  and root normalized MSE (NMSE), respectively, when  $L=25$ and $L=200$.
Concerning DOA estimation, the following observations can be made:
(a) The proposed CL-algorithms demonstrate the best performance;
(b) CL-OMP performs comparably to the MLE across all SNR and snapshot ranges;
(c) MUSIC and R-MUSIC exhibit essentially the same performance;
(d) In the high SNR regime, CL-OMP, CL-BCD, IAA, MUSIC, and R-MUSIC perform similarly, attaining the CRLB; 
(e) SBL and SAMV2 exhibit similar performance, as anticipated, but they do achieve the same level of performance as the other methods.
When $L$ increases to 200,  one can notice that CL-OMP and MLE  improve their performance the most. The threshold where CL-OMP and MLE break down is $<-15.8$ dB for $L=200$ while the 2nd best performing methods, CL-BCD and IAA, breakdown when SNR is $<-14.2$ dB.   SAMV2 and SBL show some improvement when the number of snapshots is larger. In terms of signal power estimation, CL-OMP and MLE exhibit impeccable performance, reaching the CRLB. On the other hand, CL-BCD and IAA provide accurate estimates of signal powers in the higher SNR regime, but their accuracy sharply deviates from the CRLB as SNR decreases. SAMV2 and SBL on the other hand provide clearly biased and inconsistent estimates of signal power estimates even in the high SNR regime. 

Next, we investigate the impact of the true DOA,  allowing the DOA of the source to range from $-30^o$ to $-80^o$ in 5 degree steps.  The results are presented for a sample size of $L=25$ and SNR of $-8$ dB in \autoref{fig:1source_CRB_vs_DOAsrc}.  We can notice that CL-BCD and IAA start to deviate from the performance of CL-OMP more when the true DOA deviates significantly of array broadside.  The performance of CL-OMP signal power estimate  $\hat \gamma$ is not affected by true DOA of the source. MUSIC and R-MUSIC exhibit erratic behaviour, and for some values of DOA, they encounter difficulties in identifying the true noise subspace due to the relatively small sample size.  This issue results in a notably large value for the root NMSE. 

\begin{figure}[!t]
\centerline{\input{tikz/uusiRMSE_vs_DOAsrc_SNR=-8_K=1_N=20_M=1801_L=25_LL=1500_doa=-80to-30_04-Dec-2023}} 
\caption{RMSE of $\hat \theta$ (top panel) and root NMSE of $\hat \gamma$  (bottom panel) vs true DOA in single source case; SNR is $-8$ dB,    $N=20$, $L=25$, and $M=1801$. The black line displays the CRLB.}  \label{fig:1source_CRB_vs_DOAsrc}
\end{figure}
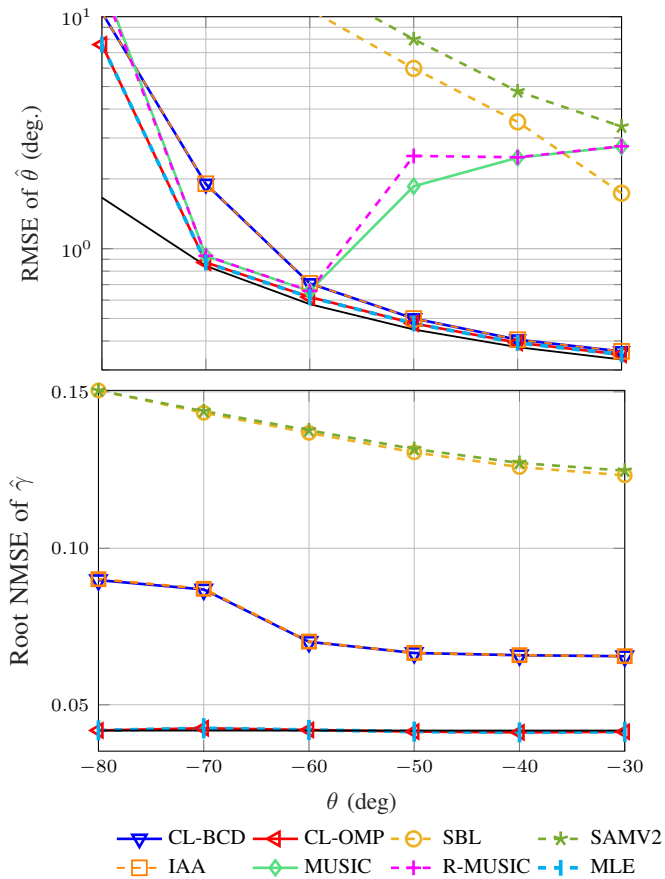

Next we consider the case of two source signals. The sources arrive from angles  $\theta_1 = -20.02^o$ and $\theta_2=3.02^o$, and are thus off the predefined grid.  The 2nd  source (closer to the array broadside) has 3dB higher power than the 1st source.  We then calculate the DOA estimate $\hat{\boldsymbol{\theta}} = (\hat \theta_1, \hat \theta_2)^\top$ and the respective  source power estimate $\hat{\boldsymbol{\sigma}}^2_s = \hat{\gam}_{\hat{\mathcal M}}$ and their empirical RMSE $\| \hat{\boldsymbol{\theta}}  - \boldsymbol{\theta}  \|$ and root NMSE $\| \hat{\boldsymbol{\sigma}}^2_s  - \boldsymbol{\sigma}^2_s  \|/\| \boldsymbol{\sigma}^2_s \|$, respectively, averaged over 1000 MC trials. The results for DOA estimation are displayed in  \autoref{fig:2source_MSEdoa} when $L=25$ (bottom panel) and $L=125$ (top panel).  
As observed, the overall performance of the methods remains highly consistent with the single-source scenario. CL-OMP continues to exhibit the best performance across all SNR levels and various sample lengths. The second-best performing methods are CL-BCD and IAA, which demonstrate very similar performance. On the other hand, SBL and SAMV2 exhibit comparable but slightly less satisfactory performance compared to other competing methods.
Top panel of \autoref{fig:2source_MSEgamma} displays the normalized RMSE of signal power estimates for snapshot size $L=25$. The conclusions are  equivalent with 1 source case: CL-OMP has impeccable performance for all SNR levels while IAA and CL-BCD provide consistent estimation only at high SNR regime, while SAMV2 and SBL estimates are heavily biased for all SNR ranges. 

Next we consider the effect of correlation of the source signals.  Earlier studies  have shown that the correlations of the source signals  have little or no effect on support recovery performance of CL-based methods (see e.g., \cite{stoica2010new,pote2020robustness}).  Let the signal covariance $\M_s$ be of the form
\[
\M_s = \bmat \sigma_1^2  & \rho \sigma_1 \sigma_2 \\   \rho \sigma_1 \sigma_2 & \sigma_2^2 \emat
\]
and  we set correlation $\rho$ as $\rho = 0.95$. Hence the signals have high correlation while their  correlation phase is equal to $0$. Otherwise the simulation set-up is as earlier. \autoref{fig:2sourceCC_MSEdoa2} displays the DOA estimation performance vs the SNR when  $L=25$ and $L=125$, respectively. We can notice that the correlation did not bring essential differences between methods compared to uncorrelated case displayed in \autoref{fig:2source_MSEdoa}. Thus all  CL methods seem  to be robust to correlation between the source signals. This is not the case for  MUSIC and R-MUSIC which encounter complete failure. This is because  their design is heavily dependent on  the assumption of uncorrelated sources. The bottom panel  of \autoref{fig:2source_MSEgamma} depicts the  root NMSE of signal power estimate $\hat{\boldsymbol{\sigma}}_s^2$ in $L=25$ case. A comparison with the top panel plot ($\rho=0$) reveals that the curves are essentially equivalent. This further confirms the robustness of CL methods to source correlation, consistent with the findings of  \cite{stoica2010new}. 

\begin{figure}[!t]
\centerline{\input{tikz/SNRvsRMSE_sources=2_theta1=-20.02_theta2=3.02_N=20_L=125_M=1801_LL=1000_nmethods=8_SNRplus=3_SNR-7to-21_20-Jan-2024.tex}}
\centerline{\input{tikz/SNRvsRMSE_sources=2_theta1=-20.02_theta2=3.02_N=20_L=25_M=1801_LL=1000_nmethods=8_SNRplus=3_SNR-2to-16_19-Jan-2024.tex}}
\caption{RMSE of $\hat{\boldsymbol{\theta}}$  vs SNR in two sources case for $L=125$ (top panel) and $L=25$ (bottom panel);  $N=20$,  $M=1801$ ($\Delta \theta=0.1^o)$, $\theta_1 = -20.02^o$ and $\theta_2=3.02^o$, and the  $2$nd  source has $3$ dB higher power than the $1$st source.}  \label{fig:2source_MSEdoa}
\end{figure}
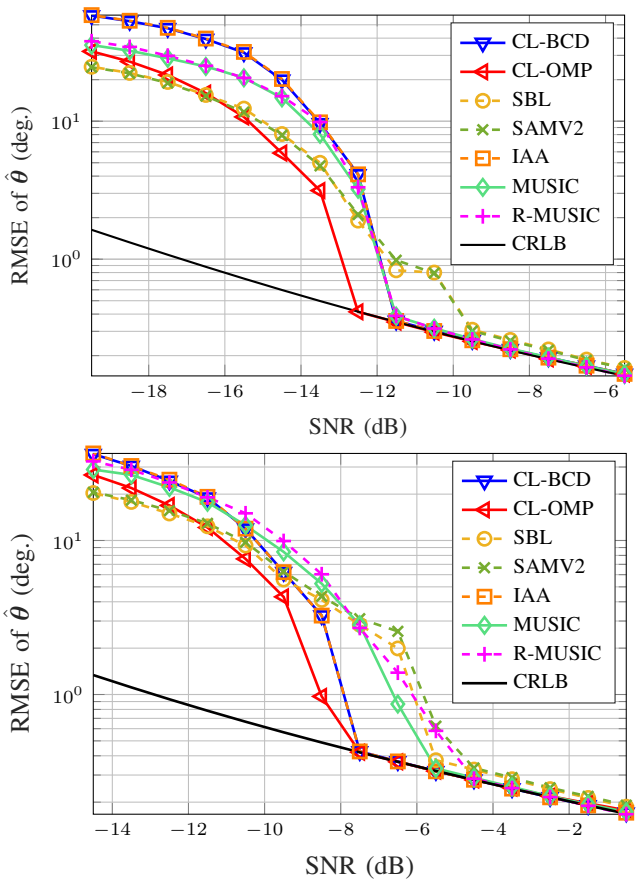

\begin{figure}[!t]
\centerline{\input{tikz/SNRvsRMSEgamma_sources=2_theta1=-20.02_theta2=3.02_N=20_L=25_M=1801_LL=1000_nmethods=8_SNRplus=3_SNR-2to-16_19-Jan-2024.tex}}
\centerline{\input{tikz/SNRvsRMSEgamma_2sources_theta1=-20.02_theta2=3.02_N=20_M=1801_L=25_LL=1000_SNRplus=3_SNR-2to-16_rho=0dot95_20-Jan-2024.tex}}
\caption{Root NMSE of  source signal power $\hat{\boldsymbol{\sigma}}^2_s$  vs SNR in two sources case for 
uncorrelated sources (top panel) and correlated  sources (bottom panel) for $\rho=0.95$;  $L=25$, $N=20$,  $M=1801$ ($\Delta \theta=0.1^o)$, $\theta_1 = -20.02^o$, $\theta_2=3.02^o$, and the $2$nd  source has a  $3$ dB higher power. }    \label{fig:2source_MSEgamma}
\end{figure}
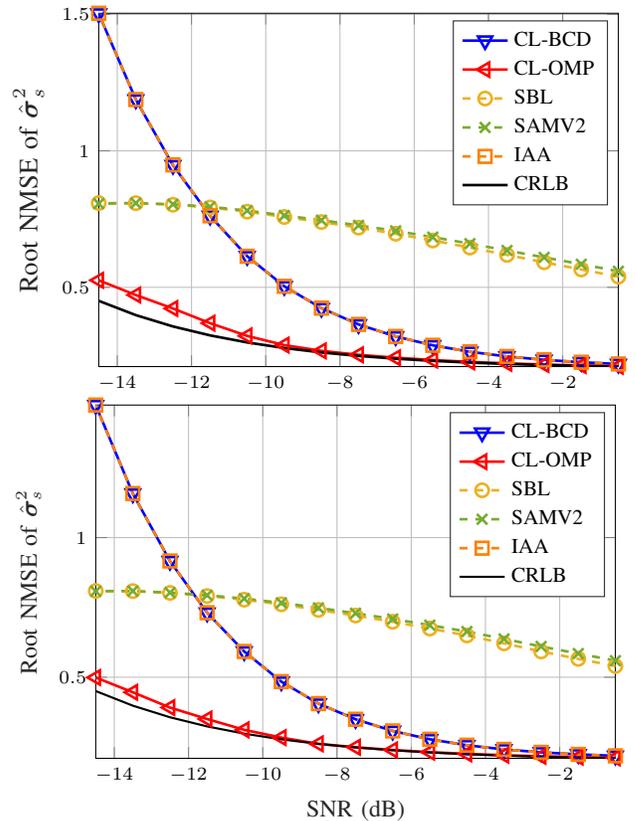

\begin{figure}[!t]
\centerline{\input{tikz/SNRvsRMSE_2sources_theta1=-20.02_theta2=3.02_N=20_M=1801_L=200_LL=1000_SNRplus=3_SNR-8to-21_rho=0dot95_04-Dec-2023}} \hspace{-2pt}
\centerline{\input{tikz/SNRvsRMSE_2sources_theta1=-20.02_theta2=3.02_N=20_M=1801_L=25_LL=1000_SNRplus=3_SNR-2to-16_rho=0dot95_20-Jan-2024.tex}} 
\caption{RMSE of $\hat{\boldsymbol{\theta}}$ in the correlated 2 sources ($\rho=0.95$) case when $L=200$ (top panel) and  $L=25$ (bottom panel); $N=20$,  $M=1801$ ($\Delta \theta=0.1^o)$,  $\theta_1 = -20.02^o$ and $\theta_2=3.02^o$, and the $2$nd  source has a $3$dB higher power.}    \label{fig:2sourceCC_MSEdoa2}
\end{figure}
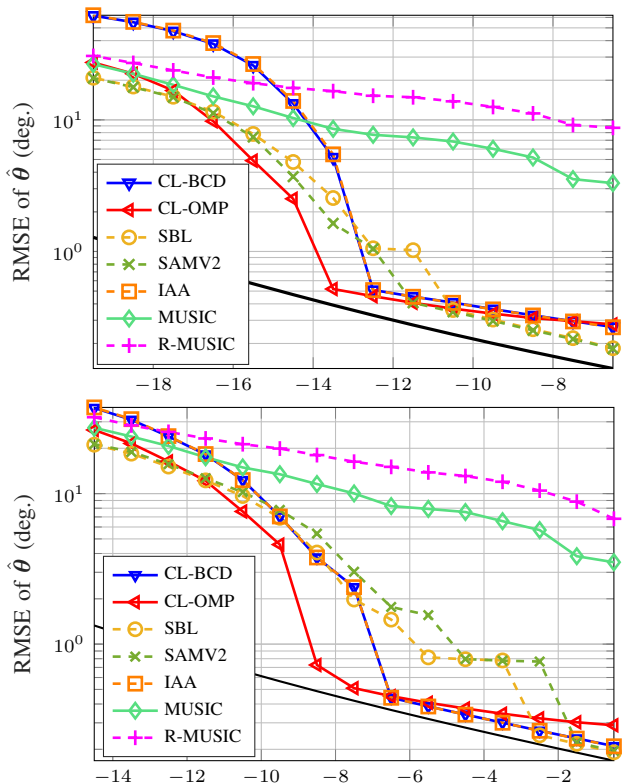

In our last  set-up we have $K=4$ sources at DOAs  $\theta_1 = -30.1^o$, $\theta_2 = -20.02^o$, $\theta_3=-10.02^o$ and $\theta_4=3.02^o$.  Note that all, except the  $1^{\text{st}}$  source, is off the predefined grid.  The SNR of the last 3 sources are $-1$, $-2$ and $-5$ dB  relative to the $1^{\text{st}}$ source.  \autoref{fig:4sources_CRB_vs_SNR} displays the RMSE $\| \hat{\boldsymbol{\theta}}  - \boldsymbol{\theta}  \|$  vs SNR when $N=20$ and $M=1801$ (as earlier) and snapshot size is $L=125$. Here we excluded SAMV2 and SBL due to their heavy computation time in this setting. As can be noted, CL-OMP outperforms other methods by a  large margin. Also CL-BCD and IAA can with stand low SNR better than MUSIC and R-MUSIC and at high SNR  all methods perform similarly.

\begin{figure}[!t]
\centerline{\input{tikz/SNRvsRMSE_sources=4_at-30.10_-20.02_-10.02_3.02_N=20_L=125_M=1801_LL=2000_SNR-6to-0.50to-12_19-Jun-2024}}
\caption{RMSE of $\hat{\boldsymbol{\theta}}$  vs SNR in four sources case; $L=125$,  $N=20$,  $M=1801$ ($\Delta \theta=0.1^o)$,$\boldsymbol{\theta} =(-30.1^o, -20.02^o, -10.02^o, 3.02^o)$.}  \label{fig:4sources_CRB_vs_SNR}
\end{figure}
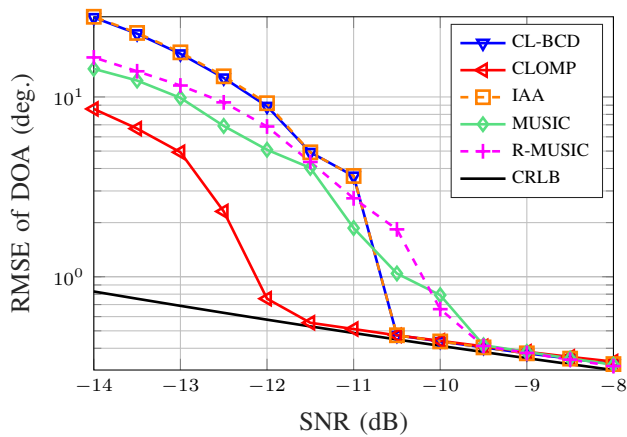

\section{Conclusions}  \label{sec:concl}

In this paper we revisited the covariance learning (CL) approach for SSR. First, relying on the Gaussian likelihood model and a known sensing matrix $\A$, we  leveraged on  FP equation outlined in Lemma~\ref{lem} when deriving the BCD algorithm which updates $\gam$  using FP-algorithm and $\sigma^2$ using MLE (cf. Lemma~2). 
 Next we proposed a  greedy pursuit CL algorithm, which we referred to CL-OMP due to its similarity with the (S)OMP algorithm.  In  the simulation study it was observed that CL-OMP outperforms  its greedy-pursuit counterparts  (SOMP or SNIHT) as well as other benchmark CL-based methods. 

Subsequently, we evaluated the performance of the methods in SSR-based source localization, where the $K$-sparsity constraint is replaced with a $K$-peaksparsity constraint. We compared against state-of-the-art methods  under a broad variaty of settings and the Cram\'er-Rao lower bound was used as a benchmark. The proposed CL-OMP was again the best performing method. Overall, the two proposed CL-algorithms outperformed competing methods especially in the low SNR regime. CL-OMP had the best breakdown behavior than other methods in low SNR.  
 Additionally, the proposed CL algorithms demonstrated solid  performance in accurately estimating the signal powers. 

There are many open questions to be addressed in future works. For example, (a) estimation of number of sources $K$; (b) deriving performance guarantees (e.g., bounds on the probability of support recovery of CL-based methods);  (c) addressing the non-Gaussian scenarios and grid-free methods; (d) developing other greedy pursuit CL methods analogous to CoSaMP \cite{needell_tropp:2009},  gOMP \cite{wang2012generalized}, etc.  Regarding items (b) and (c), works in \cite{tang2010performance} or \cite{mecklenbrauker2023robust} provide good starting points. Furthermore, there are several applications where the proposed CL-based sparse reconstructions could be useful. For example, they can be used as atom selection method in coupled dictionary learning algorithm  \cite{veshki2022multimodal} or  activity detection method for  massive machine type communications problems  \cite{liu2018sparse,haghighatshoar2018improved,fengler2021non}. These explorations are left as future works.

\input{sparseDOArev.bbl}

\end{document}

%% file: tikz_new/SAMV_simul2_N=16to128_K=4_L=32_M=256to1000_NRSIM=1_GaussianS=49_1-May-2024.tex
\definecolor{mycolor1}{rgb}{0.26600,0.27400,0.28800}%
\definecolor{mycolor2}{rgb}{0.00784,0.31373,0.26275}%
\definecolor{mycolor3}{rgb}{0.46600,0.67400,0.18800}%
\definecolor{mycolor4}{rgb}{0.54510,0.50196,0.00000}%
\definecolor{mycolor5}{rgb}{0.49412,0.11765,0.61176}%
\definecolor{mycolor6}{rgb}{0.99608,0.00392,0.60392}%

\begin{tikzpicture}
\begin{groupplot}[group style={group size=4 by 1,horizontal sep=0.29cm,vertical sep=0.27cm}]

\nextgroupplot[
width=0.47\columnwidth,
height=5.0cm,
tick label style={font=\scriptsize} , 
title style={font=\footnotesize},
xticklabel=\empty,
scale only axis,
xmin=1,
xmax=7,
xlabel style={font=\color{white!15!black}},
ymin=0,
ymax=1,
ylabel style={font=\color{white!15!black}},
axis background/.style={fill=white},
title={$N=16$},
axis x line*=bottom,
axis y line*=left,
xmajorgrids,
ymajorgrids,
legend style={legend cell align=left, align=left, draw=white!15!black}
]
\addplot [color=red, line width=0.8pt, mark size=2.5pt, mark=triangle, mark options={solid, rotate=90, red}]
  table[row sep=crcr]{%
1	0.138\\
2	0.296\\
3	0.487\\
4	0.691\\
5	0.861\\
6	0.956\\
7	0.986\\
8	0.997\\
9	0.999\\
};

\addplot [color=blue, line width=0.8pt, mark size=2.5pt, mark=triangle, mark options={solid, rotate=180, blue}]
  table[row sep=crcr]{%
1	0.052\\
2	0.111\\
3	0.186\\
4	0.29\\
5	0.374\\
6	0.474\\
7	0.562\\
8	0.662\\
9	0.739\\
};

\addplot [color=mycolor1, dashdotted, line width=0.8pt, mark size=2.9pt, mark=+, mark options={solid, mycolor1}]
  table[row sep=crcr]{%
1	0.04\\
2	0.096\\
3	0.183\\
4	0.321\\
5	0.492\\
6	0.658\\
7	0.793\\
8	0.889\\
9	0.947\\
};

\addplot [color=mycolor2, dashdotted, line width=0.8pt, mark size=2.9pt, mark=star, mark options={solid, mycolor2}]
  table[row sep=crcr]{%
1	0.071\\
2	0.171\\
3	0.303\\
4	0.491\\
5	0.696\\
6	0.829\\
7	0.928\\
8	0.972\\
9	0.989\\
};

\addplot [color=mycolor3, dashed, line width=0.8pt, mark size=2.5pt, mark=asterisk, mark options={solid, mycolor3}]
  table[row sep=crcr]{%
1	0.118\\
2	0.247\\
3	0.451\\
4	0.66\\
5	0.841\\
6	0.935\\
7	0.983\\
8	0.994\\
9	0.999\\
};

\addplot [color=mycolor4, dashed, line width=0.8pt, mark size=2.5pt, mark=o, mark options={solid, mycolor4}]
  table[row sep=crcr]{%
1	0.117\\
2	0.242\\
3	0.453\\
4	0.663\\
5	0.84\\
6	0.934\\
7	0.983\\
8	0.995\\
9	0.999\\
};

\addplot [color=mycolor5, dashdotted, line width=0.8pt, mark size=2.0pt, mark=square, mark options={solid, mycolor5}]
  table[row sep=crcr]{%
1	0.051\\
2	0.112\\
3	0.189\\
4	0.293\\
5	0.377\\
6	0.48\\
7	0.568\\
8	0.668\\
9	0.745\\
};

\addplot [color=mycolor6, dashed, line width=0.8pt, mark size=2.4pt, mark=diamond, mark options={solid, mycolor6}]
  table[row sep=crcr]{%
1	0.038\\
2	0.114\\
3	0.224\\
4	0.404\\
5	0.604\\
6	0.79\\
7	0.92\\
};

\nextgroupplot[
width=0.47\columnwidth,
height=5.0cm,
scale only axis,
title style={font=\footnotesize},
yticklabel=\empty,
xticklabel=\empty,
title style={font=\footnotesize},
tick label style={font=\scriptsize} , 
xmin=1,
xmax=7,
xlabel style={font=\color{white!15!black}},
ymin=0,
ymax=1,
axis background/.style={fill=white},
title={$N=32$},
axis x line*=bottom,
axis y line*=left,
xmajorgrids,
ymajorgrids,
legend style={at={(1.0,0.63)}, legend cell align=left, align=left, draw=white!15!black,font=\scriptsize} %
]
\addplot [color=red, line width=0.8pt, mark size=2.5pt, mark=triangle, mark options={solid, rotate=90, red}]
  table[row sep=crcr]{%
1	0.187\\
2	0.394\\
3	0.619\\
4	0.818\\
5	0.935\\
6	0.983\\
7	0.998\\
8	0.999\\
9	1\\
};
\addlegendentry{CL-OMP}

\addplot [color=blue, line width=0.8pt, mark size=2.5pt, mark=triangle, mark options={solid, rotate=180, blue}]
  table[row sep=crcr]{%
1	0.11\\
2	0.238\\
3	0.447\\
4	0.635\\
5	0.806\\
6	0.918\\
7	0.964\\
8	0.985\\
9	0.997\\
};
\addlegendentry{CL-BCD}

\addplot [color=mycolor1, dashdotted, line width=0.8pt, mark size=2.9pt, mark=+, mark options={solid, mycolor1}]
  table[row sep=crcr]{%
1	0.103\\
2	0.241\\
3	0.427\\
4	0.641\\
5	0.836\\
6	0.94\\
7	0.979\\
8	0.997\\
9	1\\
};
\addlegendentry{SOMP}

\addplot [color=mycolor2, dashdotted, line width=0.8pt, mark size=2.9pt, mark=star, mark options={solid, mycolor2}]
  table[row sep=crcr]{%
1	0.163\\
2	0.338\\
3	0.54\\
4	0.737\\
5	0.903\\
6	0.971\\
7	0.993\\
8	0.998\\
9	0.999\\
};
\addlegendentry{SNIHT}

\addplot [color=mycolor3, dashed, line width=0.8pt, mark size=2.9pt, mark=asterisk, mark options={solid, mycolor3}]
  table[row sep=crcr]{%
1	0.164\\
2	0.369\\
3	0.595\\
4	0.781\\
5	0.924\\
6	0.978\\
7	0.995\\
8	0.999\\
9	0.999\\
};
\addlegendentry{SAMV2}

\addplot [color=mycolor4, dashed, line width=0.8pt, mark size=2.5pt, mark=o, mark options={solid, mycolor4}]
  table[row sep=crcr]{%
1	0.175\\
2	0.367\\
3	0.607\\
4	0.785\\
5	0.932\\
6	0.98\\
7	0.996\\
8	0.999\\
9	0.999\\
};
\addlegendentry{SBL}

\addplot [color=mycolor5, dashdotted, line width=0.8pt, mark size=2.0pt, mark=square, mark options={solid, mycolor5}]
  table[row sep=crcr]{%
1	0.1\\
2	0.227\\
3	0.423\\
4	0.623\\
5	0.798\\
6	0.916\\
7	0.964\\
8	0.986\\
9	0.997\\
};
\addlegendentry{IAA}

\addplot [color=mycolor6, dashed, line width=0.8pt, mark size=2.4pt, mark=diamond, mark options={solid, mycolor6}]
  table[row sep=crcr]{%
1	10.038\\
2	10.114\\
3	10.224\\
4	10.404\\
5	10.604\\
6	10.79\\
7	10.92\\
};
\addlegendentry{SPICE}

\nextgroupplot[
width=0.47\columnwidth,
height=5.0cm,
scale only axis,
title style={font=\footnotesize},
yticklabel=\empty,
xticklabel=\empty,
tick label style={font=\scriptsize} , 
xmin=1,
xmax=7,
xlabel style={font=\color{white!15!black}},
ymin=0,
ymax=1,
axis background/.style={fill=white},
title={$N=64$},
axis x line*=bottom,
axis y line*=left,
xmajorgrids,
ymajorgrids,
legend style={legend cell align=left, align=left, draw=white!15!black}
]
\addplot [color=red, line width=0.8pt, mark size=2.5pt, mark=triangle, mark options={solid, rotate=90, red}]
  table[row sep=crcr]{%
1	0.239\\
2	0.436\\
3	0.681\\
4	0.851\\
5	0.942\\
6	0.988\\
7	0.998\\
8	1\\
9	1\\
};

\addplot [color=blue, line width=0.8pt, mark size=2.5pt, mark=triangle, mark options={solid, rotate=180, blue}]
  table[row sep=crcr]{%
1	0.126\\
2	0.284\\
3	0.499\\
4	0.715\\
5	0.878\\
6	0.954\\
7	0.989\\
8	0.999\\
9	1\\
};

\addplot [color=mycolor1, dashdotted, line width=0.8pt, mark size=2.9pt, mark=+, mark options={solid, mycolor1}]
  table[row sep=crcr]{%
1	0.156\\
2	0.338\\
3	0.573\\
4	0.777\\
5	0.917\\
6	0.97\\
7	0.996\\
8	0.999\\
9	1\\
};

\addplot [color=mycolor2, dashdotted, line width=0.8pt, mark size=2.9pt, mark=star, mark options={solid, mycolor2}]
  table[row sep=crcr]{%
1	0.223\\
2	0.415\\
3	0.642\\
4	0.823\\
5	0.931\\
6	0.976\\
7	0.998\\
8	0.999\\
9	1\\
};

\addplot [color=mycolor3, dashed, line width=0.8pt, mark size=2.9pt, mark=asterisk, mark options={solid, mycolor3}]
  table[row sep=crcr]{%
1	0.214\\
2	0.422\\
3	0.651\\
4	0.836\\
5	0.944\\
6	0.988\\
7	0.998\\
8	1\\
9	1\\
};

\addplot [color=mycolor4, dashed, line width=0.8pt, mark size=2.5pt, mark=o, mark options={solid, mycolor4}]
  table[row sep=crcr]{%
1	0.217\\
2	0.429\\
3	0.655\\
4	0.846\\
5	0.944\\
6	0.989\\
7	0.998\\
8	1\\
9	1\\
};

\addplot [color=mycolor5, dashdotted, line width=0.8pt, mark size=2.0pt, mark=square, mark options={solid, mycolor5}]
  table[row sep=crcr]{%
1	0.074\\
2	0.196\\
3	0.385\\
4	0.621\\
5	0.804\\
6	0.919\\
7	0.978\\
8	0.996\\
9	0.999\\
};


\nextgroupplot[
width=0.47\columnwidth,
height=5.0cm,
scale only axis,
yticklabel=\empty,
xticklabel=\empty,
tick label style={font=\scriptsize} , 
xmin=1,
xmax=7,
xlabel style={font=\color{white!15!black}},
ymin=0,
ymax=1,
axis background/.style={fill=white},
title style={font=\footnotesize},
title={$N=128$},
axis x line*=bottom,
axis y line*=left,
xmajorgrids,
ymajorgrids,
legend style={legend cell align=left, align=left, draw=white!15!black}
]
\addplot [color=red, line width=0.8pt, mark size=2.5pt, mark=triangle, mark options={solid, rotate=90, red}]
  table[row sep=crcr]{%
1	0.238\\
2	0.441\\
3	0.681\\
4	0.856\\
5	0.961\\
6	0.992\\
7	1\\
8	1\\
9	1\\
};

\addplot [color=blue, line width=0.8pt, mark size=2.5pt, mark=triangle, mark options={solid, rotate=180, blue}]
  table[row sep=crcr]{%
1	0.087\\
2	0.212\\
3	0.445\\
4	0.66\\
5	0.85\\
6	0.957\\
7	0.989\\
8	0.999\\
9	1\\
};

\addplot [color=mycolor1, dashdotted, line width=0.8pt, mark size=2.9pt, mark=+, mark options={solid, mycolor1}]
  table[row sep=crcr]{%
1	0.176\\
2	0.35\\
3	0.58\\
4	0.796\\
5	0.938\\
6	0.983\\
7	0.999\\
8	1\\
9	1\\
};

\addplot [color=mycolor2, dashdotted, line width=0.8pt, mark size=2.9pt, mark=star, mark options={solid, mycolor2}]
  table[row sep=crcr]{%
1	0.227\\
2	0.43\\
3	0.66\\
4	0.857\\
5	0.954\\
6	0.989\\
7	0.998\\
8	1\\
9	1\\
};

\addplot [color=mycolor3, dashed, line width=0.8pt, mark size=2.9pt, mark=asterisk, mark options={solid, mycolor3}]
  table[row sep=crcr]{%
1	0.221\\
2	0.419\\
3	0.664\\
4	0.849\\
5	0.952\\
6	0.991\\
7	0.998\\
8	1\\
9	1\\
};

\addplot [color=mycolor4, dashed, line width=0.8pt, mark size=2.9pt, mark=o, mark options={solid, mycolor4}]
  table[row sep=crcr]{%
1	0.222\\
2	0.423\\
3	0.665\\
4	0.855\\
5	0.954\\
6	0.989\\
7	0.999\\
8	1\\
9	1\\
};

\addplot [color=mycolor5, dashdotted, line width=0.8pt, mark size=2.0pt, mark=square, mark options={solid, mycolor5}]
  table[row sep=crcr]{%
1	0.005\\
2	0.021\\
3	0.071\\
4	0.183\\
5	0.381\\
6	0.604\\
7	0.805\\
8	0.937\\
9	0.987\\
};

\end{groupplot}

\end{tikzpicture}%

%% file: tikz_new/SAMV_simul2_N=32_K=4_L=32_M=64to512_NRSIM=500_GaussianS=1_11-May-2024.tex
\definecolor{mycolor1}{rgb}{0.26600,0.27400,0.28800}%
\definecolor{mycolor2}{rgb}{0.00784,0.31373,0.26275}%
\definecolor{mycolor3}{rgb}{0.46600,0.67400,0.18800}%
\definecolor{mycolor4}{rgb}{0.54510,0.50196,0.00000}%
\definecolor{mycolor5}{rgb}{0.49412,0.11765,0.61176}%
\begin{tikzpicture}
\begin{groupplot}[group style={group size=4 by 1,horizontal sep=0.29cm,vertical sep=0.27cm}]

\nextgroupplot[
width=0.47\columnwidth,
height=5.0cm,
tick label style={font=\scriptsize} , 
title style={font=\footnotesize},
xticklabel=\empty,
scale only axis,
xmin=1,
xmax=7,
xlabel style={font=\color{white!15!black}},
ymin=0,
ymax=1,
ylabel style={font=\color{white!15!black}},
axis background/.style={fill=white},
title={$M=64$},
axis x line*=bottom,
axis y line*=left,
xmajorgrids,
ymajorgrids,
legend style={at={(0.03,0.97)}, anchor=north west, legend cell align=left, align=left, draw=white!15!black,font=\scriptsize}
]
\addplot [color=red, line width=0.8pt, mark size=2.5pt, mark=triangle, mark options={solid, rotate=90, red}]
  table[row sep=crcr]{%
1	0.378\\
2	0.582\\
3	0.742\\
4	0.884\\
5	0.976\\
6	0.99\\
7	1\\
8	1\\
9	1\\
};

\addplot [color=blue, line width=0.8pt, mark size=2.5pt, mark=triangle, mark options={solid, rotate=180, blue}]
  table[row sep=crcr]{%
1	0.19\\
2	0.362\\
3	0.54\\
4	0.746\\
5	0.868\\
6	0.942\\
7	0.986\\
8	0.996\\
9	1\\
};

\addplot [color=mycolor1, dashdotted, line width=0.8pt, mark size=2.9pt, mark=+, mark options={solid, mycolor1}]
  table[row sep=crcr]{%
1	0.262\\
2	0.424\\
3	0.626\\
4	0.786\\
5	0.904\\
6	0.968\\
7	0.99\\
8	1\\
9	1\\
};

\addplot [color=mycolor2, dashdotted, line width=0.8pt, mark size=2.9pt, mark=star, mark options={solid, mycolor2}]
  table[row sep=crcr]{%
1	0.332\\
2	0.52\\
3	0.694\\
4	0.828\\
5	0.936\\
6	0.978\\
7	0.994\\
8	1\\
9	1\\
};

\addplot [color=mycolor3, dashed, line width=0.8pt, mark size=2.9pt, mark=asterisk, mark options={solid, mycolor3}]
  table[row sep=crcr]{%
1	0.366\\
2	0.574\\
3	0.728\\
4	0.874\\
5	0.958\\
6	0.988\\
7	0.998\\
8	1\\
9	1\\
};

\addplot [color=mycolor4, dashed, line width=0.8pt, mark size=2.0pt, mark=o, mark options={solid, mycolor4}]
  table[row sep=crcr]{%
1	0.364\\
2	0.566\\
3	0.732\\
4	0.876\\
5	0.96\\
6	0.988\\
7	0.998\\
8	1\\
9	1\\
};

\addplot [color=mycolor5, dashdotted, line width=0.8pt, mark size=1.9pt, mark=square, mark options={solid, mycolor5}]
  table[row sep=crcr]{%
1	0.022\\
2	0.058\\
3	0.142\\
4	0.28\\
5	0.486\\
6	0.696\\
7	0.84\\
8	0.93\\
9	0.974\\
};


\nextgroupplot[
width=0.47\columnwidth,
height=5.0cm,
title style={font=\footnotesize},
scale only axis,
yticklabel=\empty,
xticklabel=\empty,
tick label style={font=\scriptsize} , 
xmin=1,
xmax=7,
xlabel style={font=\color{white!15!black}},
ymin=0,
ymax=1,
axis background/.style={fill=white},
title={$M=128$},
axis x line*=bottom,
axis y line*=left,
xmajorgrids,
ymajorgrids,
legend style={at={(1.0,0.6)}, legend cell align=left, align=left, draw=white!15!black,font=\scriptsize} %
]
\addplot [color=red, line width=0.8pt, mark size=2.5pt, mark=triangle, mark options={solid, rotate=90, red}]
  table[row sep=crcr]{%
1	0.276\\
2	0.49\\
3	0.68\\
4	0.856\\
5	0.962\\
6	0.984\\
7	0.996\\
8	1\\
9	1\\
};
\addlegendentry{CL-OMP}

\addplot [color=blue, line width=0.8pt, mark size=2.5pt, mark=triangle, mark options={solid, rotate=180, blue}]
  table[row sep=crcr]{%
1	0.146\\
2	0.286\\
3	0.508\\
4	0.716\\
5	0.832\\
6	0.934\\
7	0.978\\
8	0.994\\
9	0.998\\
};
\addlegendentry{CL-BCD}

\addplot [color=mycolor1, dashdotted, line width=0.8pt, mark size=2.9pt, mark=+, mark options={solid, mycolor1}]
  table[row sep=crcr]{%
1	0.156\\
2	0.316\\
3	0.536\\
4	0.748\\
5	0.876\\
6	0.956\\
7	0.992\\
8	0.998\\
9	0.998\\
};
\addlegendentry{SOMP}

\addplot [color=mycolor2, dashdotted, line width=0.8pt, mark size=2.9pt, mark=star, mark options={solid, mycolor2}]
  table[row sep=crcr]{%
1	0.236\\
2	0.414\\
3	0.616\\
4	0.786\\
5	0.934\\
6	0.976\\
7	0.992\\
8	0.998\\
9	1\\
};
\addlegendentry{SNIHT}

\addplot [color=mycolor3, dashed, line width=0.8pt, mark size=2.9pt, mark=asterisk, mark options={solid, mycolor3}]
  table[row sep=crcr]{%
1	0.268\\
2	0.474\\
3	0.674\\
4	0.844\\
5	0.954\\
6	0.988\\
7	0.998\\
8	1\\
9	1\\
};
\addlegendentry{SAMV2}

\addplot [color=mycolor4, dashed, line width=0.8pt, mark size=2.0pt, mark=o, mark options={solid, mycolor4}]
  table[row sep=crcr]{%
1	0.27\\
2	0.474\\
3	0.676\\
4	0.842\\
5	0.956\\
6	0.986\\
7	0.998\\
8	1\\
9	1\\
};
\addlegendentry{SBL}

\addplot [color=mycolor5, dashdotted, line width=0.8pt, mark size=1.9pt, mark=square, mark options={solid, mycolor5}]
  table[row sep=crcr]{%
1	0.1\\
2	0.216\\
3	0.386\\
4	0.618\\
5	0.774\\
6	0.9\\
7	0.962\\
8	0.994\\
9	0.998\\
};
\addlegendentry{IAA}


\nextgroupplot[
width=0.47\columnwidth,
height=5.0cm,
scale only axis,
xmin=1,
xmax=7,
xlabel style={font=\color{white!15!black}},
ymin=0,
ymax=1,
yticklabel=\empty,
xticklabel=\empty,
tick label style={font=\scriptsize} , 
axis background/.style={fill=white},
title style={font=\footnotesize},
title={$M=256$},
axis x line*=bottom,
axis y line*=left,
xmajorgrids,
ymajorgrids,
legend style={legend cell align=left, align=left, draw=white!15!black}
]
\addplot [color=red, line width=0.8pt, mark size=2.5pt, mark=triangle, mark options={solid, rotate=90, red}]
  table[row sep=crcr]{%
1	0.192\\
2	0.414\\
3	0.628\\
4	0.824\\
5	0.934\\
6	0.982\\
7	1\\
8	1\\
9	1\\
};

\addplot [color=blue, line width=0.8pt, mark size=2.5pt, mark=triangle, mark options={solid, rotate=180, blue}]
  table[row sep=crcr]{%
1	0.116\\
2	0.25\\
3	0.466\\
4	0.644\\
5	0.804\\
6	0.916\\
7	0.958\\
8	0.984\\
9	0.998\\
};

\addplot [color=mycolor1, dashdotted, line width=0.8pt, mark size=2.9pt, mark=+, mark options={solid, mycolor1}]
  table[row sep=crcr]{%
1	0.11\\
2	0.254\\
3	0.456\\
4	0.67\\
5	0.846\\
6	0.928\\
7	0.978\\
8	0.998\\
9	1\\
};

\addplot [color=mycolor2, dashdotted, line width=0.8pt, mark size=2.9pt, mark=star, mark options={solid, mycolor2}]
  table[row sep=crcr]{%
1	0.154\\
2	0.352\\
3	0.556\\
4	0.754\\
5	0.892\\
6	0.966\\
7	0.994\\
8	1\\
9	1\\
};

\addplot [color=mycolor3, dashed, line width=0.8pt, mark size=2.9pt, mark=asterisk, mark options={solid, mycolor3}]
  table[row sep=crcr]{%
1	0.162\\
2	0.404\\
3	0.616\\
4	0.79\\
5	0.928\\
6	0.976\\
7	0.996\\
8	1\\
9	1\\
};

\addplot [color=mycolor4, dashed, line width=0.8pt, mark size=2.0pt, mark=o, mark options={solid, mycolor4}]
  table[row sep=crcr]{%
1	0.172\\
2	0.402\\
3	0.622\\
4	0.796\\
5	0.936\\
6	0.978\\
7	0.996\\
8	1\\
9	1\\
};

\addplot [color=mycolor5, dashdotted, line width=0.8pt, mark size=1.9pt, mark=square, mark options={solid, mycolor5}]
  table[row sep=crcr]{%
1	0.102\\
2	0.238\\
3	0.436\\
4	0.628\\
5	0.796\\
6	0.916\\
7	0.956\\
8	0.986\\
9	0.998\\
};


\nextgroupplot[
width=0.47\columnwidth,
height=5.0cm,
tick label style={font=\scriptsize} , 
yticklabel=\empty,
xticklabel=\empty,
scale only axis,
xmin=1,
xmax=7,
xlabel style={font=\color{white!15!black}},
ymin=0,
ymax=1,
axis background/.style={fill=white},
title style={font=\footnotesize},
title={$M=512$},
axis x line*=bottom,
axis y line*=left,
xmajorgrids,
ymajorgrids,
legend style={legend cell align=left, align=left, draw=white!15!black}
]
\addplot [color=red, line width=0.8pt, mark size=2.5pt, mark=triangle, mark options={solid, rotate=90, red}]
  table[row sep=crcr]{%
1	0.13\\
2	0.28\\
3	0.54\\
4	0.742\\
5	0.906\\
6	0.982\\
7	0.994\\
8	0.998\\
9	1\\
};

\addplot [color=blue, line width=0.8pt, mark size=2.5pt, mark=triangle, mark options={solid, rotate=180, blue}]
  table[row sep=crcr]{%
1	0.076\\
2	0.176\\
3	0.32\\
4	0.534\\
5	0.698\\
6	0.836\\
7	0.924\\
8	0.958\\
9	0.988\\
};

\addplot [color=mycolor1, dashdotted, line width=0.8pt, mark size=2.9pt, mark=+, mark options={solid, mycolor1}]
  table[row sep=crcr]{%
1	0.064\\
2	0.176\\
3	0.364\\
4	0.582\\
5	0.794\\
6	0.916\\
7	0.986\\
8	0.992\\
9	0.998\\
};

\addplot [color=mycolor2, dashdotted, line width=0.8pt, mark size=2.9pt, mark=star, mark options={solid, mycolor2}]
  table[row sep=crcr]{%
1	0.106\\
2	0.216\\
3	0.426\\
4	0.646\\
5	0.852\\
6	0.954\\
7	0.994\\
8	0.998\\
9	0.998\\
};

\addplot [color=mycolor3, dashed, line width=0.8pt, mark size=2.9pt, mark=asterisk, mark options={solid, mycolor3}]
  table[row sep=crcr]{%
1	0.098\\
2	0.248\\
3	0.504\\
4	0.71\\
5	0.892\\
6	0.968\\
7	0.992\\
8	0.998\\
9	1\\
};

\addplot [color=mycolor4, dashed, line width=0.8pt, mark size=2.0pt, mark=o, mark options={solid, mycolor4}]
  table[row sep=crcr]{%
1	0.1\\
2	0.252\\
3	0.514\\
4	0.718\\
5	0.892\\
6	0.974\\
7	0.992\\
8	0.998\\
9	1\\
};

\addplot [color=mycolor5, dashdotted, line width=0.8pt, mark size=1.9pt, mark=square, mark options={solid, mycolor5}]
  table[row sep=crcr]{%
1	0.072\\
2	0.174\\
3	0.324\\
4	0.534\\
5	0.702\\
6	0.84\\
7	0.928\\
8	0.96\\
9	0.99\\
};
\end{groupplot}
\end{tikzpicture}%

%% file: tikz_new/SAMV_simul2_N=32_K=4_L=16to128_M=256to1000_NRSIM=1_GaussianS=49_2-May-2024.tex
\definecolor{mycolor1}{rgb}{0.26600,0.27400,0.28800}%
\definecolor{mycolor2}{rgb}{0.00784,0.31373,0.26275}%
\definecolor{mycolor3}{rgb}{0.46600,0.67400,0.18800}%
\definecolor{mycolor4}{rgb}{0.54510,0.50196,0.00000}%
\definecolor{mycolor5}{rgb}{0.49412,0.11765,0.61176}%
\definecolor{mycolor6}{rgb}{0.99608,0.00392,0.60392}%
\begin{tikzpicture}
\begin{groupplot}[group style={group size=4 by 1,horizontal sep=0.29cm,vertical sep=0.27cm}]

\nextgroupplot[
width=0.47\columnwidth,
height=5.0cm,
scale only axis,
tick label style={font=\scriptsize} , 
title style={font=\footnotesize},
xmin=1,
xmax=7,
xlabel style={font=\color{white!15!black}},
xlabel={SNR},
ymin=0,
ymax=1,
ylabel style={font=\color{white!15!black}},
axis background/.style={fill=white},
title={$L=16$},
axis x line*=bottom,
axis y line*=left,
xmajorgrids,
ymajorgrids,
legend style={legend cell align=left, align=left, draw=white!15!black}
]
\addplot [color=red, line width=0.8pt, mark size=2.5pt, mark=triangle, mark options={solid, rotate=90, red}]
  table[row sep=crcr]{%
1	0.029\\
2	0.087\\
3	0.21\\
4	0.383\\
5	0.618\\
6	0.791\\
7	0.913\\
};

\addplot [color=blue, line width=0.8pt, mark size=2.5pt, mark=triangle, mark options={solid, rotate=180, blue}]
  table[row sep=crcr]{%
1	0.013\\
2	0.05\\
3	0.121\\
4	0.253\\
5	0.429\\
6	0.609\\
7	0.773\\
};

\addplot [color=mycolor1, dashed, line width=0.8pt, mark size=2.9pt, mark=+, mark options={solid, mycolor1}]
  table[row sep=crcr]{%
1	0.013\\
2	0.041\\
3	0.114\\
4	0.259\\
5	0.461\\
6	0.653\\
7	0.818\\
};

\addplot [color=mycolor2, dashed, line width=0.8pt, mark size=2.9pt, mark=star, mark options={solid, mycolor2}]
  table[row sep=crcr]{%
1	0.025\\
2	0.077\\
3	0.18\\
4	0.325\\
5	0.546\\
6	0.736\\
7	0.864\\
};

\addplot [color=mycolor3, dashed, line width=0.8pt, mark size=2.9pt, mark=asterisk, mark options={solid, mycolor3}]
  table[row sep=crcr]{%
1	0.027\\
2	0.071\\
3	0.175\\
4	0.345\\
5	0.55\\
6	0.753\\
7	0.895\\
};

\addplot [color=mycolor4, dashed, line width=0.8pt, mark size=2.5pt, mark=o, mark options={solid, mycolor4}]
  table[row sep=crcr]{%
1	0.027\\
2	0.073\\
3	0.185\\
4	0.354\\
5	0.568\\
6	0.77\\
7	0.901\\
};

\addplot [color=mycolor5, dashdotted, line width=0.8pt, mark size=2.0pt, mark=square, mark options={solid, mycolor5}]
  table[row sep=crcr]{%
1	0.012\\
2	0.04\\
3	0.113\\
4	0.232\\
5	0.416\\
6	0.597\\
7	0.769\\
};


\nextgroupplot[
width=0.47\columnwidth,
height=5.0cm,
scale only axis,
title style={font=\footnotesize},
yticklabel=\empty,
title style={font=\footnotesize},
tick label style={font=\scriptsize} , 
xmin=1,
xmax=7,
xlabel style={font=\color{white!15!black}},
xlabel={SNR},
ymin=0,
ymax=1,
axis background/.style={fill=white},
title={$L=32$},
axis x line*=bottom,
axis y line*=left,
xmajorgrids,
ymajorgrids,
legend style={legend cell align=left, align=left, draw=white!15!black}
]
\addplot [color=red, line width=0.8pt, mark size=2.5pt, mark=triangle, mark options={solid, rotate=90, red}]
  table[row sep=crcr]{%
1	0.187\\
2	0.394\\
3	0.619\\
4	0.818\\
5	0.935\\
6	0.983\\
7	0.998\\
};

\addplot [color=blue, line width=0.8pt, mark size=2.5pt, mark=triangle, mark options={solid, rotate=180, blue}]
  table[row sep=crcr]{%
1	0.11\\
2	0.238\\
3	0.447\\
4	0.635\\
5	0.806\\
6	0.918\\
7	0.964\\
};

\addplot [color=mycolor1, dashed, line width=0.8pt, mark size=2.9pt, mark=+, mark options={solid, mycolor1}]
  table[row sep=crcr]{%
1	0.103\\
2	0.241\\
3	0.427\\
4	0.641\\
5	0.836\\
6	0.94\\
7	0.979\\
};

\addplot [color=mycolor2, dashed, line width=0.8pt, mark size=2.9pt, mark=star, mark options={solid, mycolor2}]
  table[row sep=crcr]{%
1	0.163\\
2	0.338\\
3	0.54\\
4	0.737\\
5	0.903\\
6	0.971\\
7	0.993\\
};

\addplot [color=mycolor3, dashed, line width=0.8pt, mark size=2.9pt, mark=asterisk, mark options={solid, mycolor3}]
  table[row sep=crcr]{%
1	0.164\\
2	0.369\\
3	0.595\\
4	0.781\\
5	0.924\\
6	0.978\\
7	0.995\\
};

\addplot [color=mycolor4, dashed, line width=0.8pt, mark size=2.5pt, mark=o, mark options={solid, mycolor4}]
  table[row sep=crcr]{%
1	0.171\\
2	0.367\\
3	0.604\\
4	0.784\\
5	0.93\\
6	0.979\\
7	0.996\\
};

\addplot [color=mycolor5, dashdotted, line width=0.8pt, mark size=2.0pt, mark=square, mark options={solid, mycolor5}]
  table[row sep=crcr]{%
1	0.1\\
2	0.227\\
3	0.423\\
4	0.623\\
5	0.798\\
6	0.916\\
7	0.964\\
};


\nextgroupplot[
width=0.47\columnwidth,
height=5.0cm,
scale only axis,
title style={font=\footnotesize},
yticklabel=\empty,
title style={font=\footnotesize},
tick label style={font=\scriptsize} , 
xmin=1,
xmax=7,
xlabel style={font=\color{white!15!black}},
xlabel={SNR},
ymin=0,
ymax=1,
axis background/.style={fill=white},
title={$L=64$},
axis x line*=bottom,
axis y line*=left,
xmajorgrids,
ymajorgrids,
legend style={legend cell align=left, align=left, draw=white!15!black}
]
\addplot [color=red, line width=0.8pt, mark size=2.5pt, mark=triangle, mark options={solid, rotate=90, red}]
  table[row sep=crcr]{%
1	0.62\\
2	0.805\\
3	0.931\\
4	0.984\\
5	0.999\\
6	1\\
7	1\\
};

\addplot [color=blue, line width=0.8pt, mark size=2.5pt, mark=triangle, mark options={solid, rotate=180, blue}]
  table[row sep=crcr]{%
1	0.43\\
2	0.643\\
3	0.803\\
4	0.907\\
5	0.967\\
6	0.995\\
7	0.997\\
};

\addplot [color=mycolor1, dashed, line width=0.8pt, mark size=2.9pt, mark=+, mark options={solid, mycolor1}]
  table[row sep=crcr]{%
1	0.428\\
2	0.606\\
3	0.809\\
4	0.914\\
5	0.969\\
6	0.991\\
7	0.999\\
};

\addplot [color=mycolor2, dashed, line width=0.8pt, mark size=2.9pt, mark=star, mark options={solid, mycolor2}]
  table[row sep=crcr]{%
1	0.518\\
2	0.72\\
3	0.88\\
4	0.968\\
5	0.993\\
6	0.999\\
7	1\\
};

\addplot [color=mycolor3, dashed, line width=0.8pt, mark size=2.9pt, mark=asterisk, mark options={solid, mycolor3}]
  table[row sep=crcr]{%
1	0.599\\
2	0.781\\
3	0.928\\
4	0.982\\
5	0.999\\
6	1\\
7	1\\
};

\addplot [color=mycolor4, dashed, line width=0.8pt, mark size=2.2pt, mark=o, mark options={solid, mycolor4}]
  table[row sep=crcr]{%
1	0.599\\
2	0.781\\
3	0.929\\
4	0.982\\
5	0.999\\
6	1\\
7	1\\
};

\addplot [color=mycolor5, dashdotted, line width=0.8pt, mark size=2.0pt, mark=square, mark options={solid, mycolor5}]
  table[row sep=crcr]{%
1	0.403\\
2	0.622\\
3	0.799\\
4	0.905\\
5	0.964\\
6	0.994\\
7	0.998\\
};

\addplot [color=mycolor6, dashed, line width=0.8pt, mark size=2.4pt, mark=diamond, mark options={solid, mycolor6}]
  table[row sep=crcr]{%
1	0.391\\
2	0.629\\
3	0.82\\
4	0.928\\
5	0.988\\
6	0.999\\
7	1\\
};
%

\nextgroupplot[
width=0.47\columnwidth,
height=5.0cm,
scale only axis,
title style={font=\footnotesize},
yticklabel=\empty,
title style={font=\footnotesize},
tick label style={font=\scriptsize} , 
xmin=1,
xmax=7,
xlabel style={font=\color{white!15!black}},
xlabel={SNR},
ymin=0,
ymax=1,
axis background/.style={fill=white},
title={$L=128$},
axis x line*=bottom,
axis y line*=left,
xmajorgrids,
ymajorgrids,
legend style={at={(1.0,0.8)}, legend cell align=left, align=left, draw=white!15!black,font=\scriptsize} %
]
\addplot [color=red, line width=0.8pt, mark size=2.5pt, mark=triangle, mark options={solid, rotate=90, red}]
  table[row sep=crcr]{%
1	0.934\\
2	0.984\\
3	0.999\\
4	1\\
5	1\\
6	1\\
7	1\\
};
\addlegendentry{CL-OMP}

\addplot [color=blue, line width=0.8pt, mark size=2.5pt, mark=triangle, mark options={solid, rotate=180, blue}]
  table[row sep=crcr]{%
1	0.788\\
2	0.92\\
3	0.98\\
4	0.996\\
5	0.999\\
6	1\\
7	1\\
};
\addlegendentry{CL-BCD}

\addplot [color=mycolor1, dashed, line width=0.8pt, mark size=2.9pt, mark=+, mark options={solid, mycolor1}]
  table[row sep=crcr]{%
1	0.722\\
2	0.876\\
3	0.955\\
4	0.987\\
5	0.997\\
6	0.999\\
7	1\\
};
\addlegendentry{SOMP}

\addplot [color=mycolor2, dashed, line width=0.8pt, mark size=2.9pt, mark=star, mark options={solid, mycolor2}]
  table[row sep=crcr]{%
1	0.86\\
2	0.954\\
3	0.99\\
4	0.999\\
5	0.999\\
6	1\\
7	1\\
};
\addlegendentry{SNIHT}

\addplot [color=mycolor3, dashed, line width=0.8pt, mark size=2.9pt, mark=asterisk, mark options={solid, mycolor3}]
  table[row sep=crcr]{%
1	0.933\\
2	0.986\\
3	0.998\\
4	1\\
5	1\\
6	1\\
7	1\\
};
\addlegendentry{SAMV2}

\addplot [color=mycolor4, dashed, line width=0.8pt, mark size=2.2pt, mark=o, mark options={solid, mycolor4}]
  table[row sep=crcr]{%
1	0.933\\
2	0.985\\
3	0.999\\
4	1\\
5	1\\
6	1\\
7	1\\
};
\addlegendentry{SBL}

\addplot [color=mycolor5, dashdotted, line width=0.8pt, mark size=2.0pt, mark=square, mark options={solid, mycolor5}]
  table[row sep=crcr]{%
1	0.782\\
2	0.907\\
3	0.975\\
4	0.996\\
5	0.999\\
6	1\\
7	1\\
};
\addlegendentry{IAA}

\addplot [color=mycolor6, dashed, line width=0.8pt, mark size=2.4pt, mark=diamond, mark options={solid, mycolor6}]
  table[row sep=crcr]{%
1	0.877\\
2	0.973\\
3	0.999\\
4	1\\
5	1\\
6	1\\
7	1\\
};
\addlegendentry{SPICE}

\end{groupplot}

\end{tikzpicture}%

%% file: tikz_new/SAMV_simul2_cpu_N=16to128_K=4_L=32_M=256_NRSIM=500_GaussianS=1_04-Jun-2024.tex
\definecolor{mycolor1}{rgb}{0.26600,0.27400,0.28800}%
\definecolor{mycolor2}{rgb}{0.00784,0.31373,0.26275}%
\definecolor{mycolor3}{rgb}{0.46600,0.67400,0.18800}%
\definecolor{mycolor4}{rgb}{0.54510,0.50196,0.00000}%
\definecolor{mycolor5}{rgb}{0.49412,0.11765,0.61176}%
\definecolor{mycolor6}{rgb}{0.99608,0.00392,0.60392}%
\begin{tikzpicture}
\begin{groupplot}[group style={group size=4 by 1,horizontal sep=0.965cm,vertical sep=0.27cm}]

\nextgroupplot[
width=0.55\columnwidth,
height=5.0cm,
tick label style={font=\scriptsize} , 
title style={font=\footnotesize},
xticklabel=\empty,
scale only axis,
xmin=1,
xmax=7,
xlabel style={font=\color{white!15!black}},
ymin=0.000634068796,
ymax=0.140443496458,
ylabel style={font=\color{white!15!black}},
ylabel={running time [s]},
axis background/.style={fill=white},
title={$N=16$},
axis x line*=bottom,
axis y line*=left,
xmajorgrids,
ymajorgrids,
yticklabel=\pgfkeys{/pgf/number format/.cd,fixed,precision=2,zerofill}\pgfmathprintnumber{\tick}, 
legend style={legend cell align=left, align=left, draw=white!15!black}
]
\addplot [color=red, line width=0.8pt, mark size=2.5pt, mark=triangle, mark options={solid, rotate=90, red}]
  table[row sep=crcr]{%
1	0.001727988472\\
2	0.00162681799\\
3	0.001611174792\\
4	0.001688151814\\
5	0.001629687498\\
6	0.001644898842\\
7	0.001680348524\\
};

\addplot [color=blue, line width=0.8pt, mark size=2.5pt, mark=triangle, mark options={solid, rotate=180, blue}]
  table[row sep=crcr]{%
1	0.002997164692\\
2	0.002962322742\\
3	0.002871960482\\
4	0.00293240084\\
5	0.002998537954\\
6	0.00314978828\\
7	0.003186146082\\
};

\addplot [color=mycolor1, dashdotted, line width=0.8pt, mark size=2.9pt, mark=+, mark options={solid, mycolor1}]
  table[row sep=crcr]{%
1	0.000698829878\\
2	0.000662391766\\
3	0.000636335532\\
4	0.000634068796\\
5	0.000665550682\\
6	0.00066517998\\
7	0.000661897054\\
};

\addplot [color=mycolor2, dashdotted, line width=0.8pt, mark size=2.9pt, mark=star, mark options={solid, mycolor2}]
  table[row sep=crcr]{%
1	0.00844229043600001\\
2	0.007401898358\\
3	0.006373940652\\
4	0.005877114508\\
5	0.00573464537\\
6	0.005176676594\\
7	0.004774094582\\
};

\addplot [color=mycolor3, dashed, line width=0.8pt, mark size=2.5pt, mark=asterisk, mark options={solid, mycolor3}]
  table[row sep=crcr]{%
1	0.08177562477\\
2	0.072987353006\\
3	0.0625723907979999\\
4	0.05669838867\\
5	0.0545225998699999\\
6	0.052309291622\\
7	0.047153375036\\
};

\addplot [color=mycolor4, dashed, line width=0.8pt, mark size=2.5pt, mark=o, mark options={solid, mycolor4}]
  table[row sep=crcr]{%
1	0.140443496458\\
2	0.124931676578\\
3	0.108157585716\\
4	0.0968431227859999\\
5	0.0934943308360001\\
6	0.088599184078\\
7	0.0792486556699999\\
};

\addplot [color=mycolor5, dashdotted, line width=0.8pt, mark size=2.0pt, mark=square, mark options={solid, mycolor5}]
  table[row sep=crcr]{%
1	0.07495431691\\
2	0.0671639718140001\\
3	0.060023316202\\
4	0.055315885678\\
5	0.0533372119540001\\
6	0.050986568436\\
7	0.045227097198\\
};

\nextgroupplot[
width=0.55\columnwidth,
height=5.0cm,
scale only axis,
xticklabel=\empty,
tick label style={font=\scriptsize} , 
xmin=1,
xmax=7,
xlabel style={font=\color{white!15!black}},
ymin=0.000782610554,
ymax=0.280339412648,
axis background/.style={fill=white},
title style={font=\footnotesize},
title={$N=64$},
axis x line*=bottom,
axis y line*=left,
xmajorgrids,
ymajorgrids,
yticklabel=\pgfkeys{/pgf/number format/.cd,fixed,precision=2,zerofill}\pgfmathprintnumber{\tick}, 
legend style={legend cell align=left, align=left, draw=white!15!black}
]
\addplot [color=red, line width=0.8pt, mark size=2.5pt, mark=triangle, mark options={solid, rotate=90, red}]
  table[row sep=crcr]{%
1	0.003914669942\\
2	0.003972239086\\
3	0.004003947448\\
4	0.003851966006\\
5	0.003992213148\\
6	0.00390961472\\
7	0.003822425228\\
};

\addplot [color=blue, line width=0.8pt, mark size=2.5pt, mark=triangle, mark options={solid, rotate=180, blue}]
  table[row sep=crcr]{%
1	0.009312682456\\
2	0.008177984452\\
3	0.007973796148\\
4	0.00673040343000001\\
5	0.00798637037\\
6	0.006933372368\\
7	0.006901404984\\
};

\addplot [color=mycolor1, dashdotted, line width=0.8pt, mark size=2.9pt, mark=+, mark options={solid, mycolor1}]
  table[row sep=crcr]{%
1	0.000830877838\\
2	0.000816929022\\
3	0.000834896786\\
4	0.000805494782\\
5	0.000832140534\\
6	0.000811931265999999\\
7	0.000782610554\\
};

\addplot [color=mycolor2, dashdotted, line width=0.8pt, mark size=2.9pt, mark=star, mark options={solid, mycolor2}]
  table[row sep=crcr]{%
1	0.0048285471\\
2	0.004383769636\\
3	0.003930089258\\
4	0.003589466406\\
5	0.003518580008\\
6	0.003184832442\\
7	0.003082153746\\
};

\addplot [color=mycolor3, dashed, line width=0.8pt, mark size=2.9pt, mark=asterisk, mark options={solid, mycolor3}]
  table[row sep=crcr]{%
1	0.228235980476\\
2	0.200254287476\\
3	0.179026188518\\
4	0.156647463168\\
5	0.145537298604\\
6	0.126008064576\\
7	0.109322039802\\
};

\addplot [color=mycolor4, dashed, line width=0.8pt, mark size=2.9pt, mark=o, mark options={solid, mycolor4}]
  table[row sep=crcr]{%
1	0.280339412648\\
2	0.249006032934\\
3	0.22340926109\\
4	0.194524957156\\
5	0.177980167634\\
6	0.1564508987\\
7	0.13764765975\\
};

\addplot [color=mycolor5, dashdotted, line width=0.8pt, mark size=2.0pt, mark=square, mark options={solid, mycolor5}]
  table[row sep=crcr]{%
1	0.280339412648\\
2	0.249006032934\\
3	0.22340926109\\
4	0.194524957156\\
5	0.177980167634\\
6	0.1564508987\\
7	0.13764765975\\
};

\nextgroupplot[
width=0.55\columnwidth,
height=5.0cm,
scale only axis,
title style={font=\footnotesize},
xticklabel=\empty,
tick label style={font=\scriptsize} , 
xmin=1,
xmax=7,
xlabel style={font=\color{white!15!black}},
ymin=0.001085796162,
ymax=0.655958469792,
axis background/.style={fill=white},
title={$N=128$},
axis x line*=bottom,
axis y line*=left,
xmajorgrids,
ymajorgrids,
yticklabel=\pgfkeys{/pgf/number format/.cd,fixed,precision=2,zerofill}\pgfmathprintnumber{\tick}, 
legend style={at={(0.4,0.6)}, legend cell align=left, align=left, draw=white!15!black,font=\scriptsize} %
]
\addplot [color=red, line width=0.8pt, mark size=2.5pt, mark=triangle, mark options={solid, rotate=90, red}]
  table[row sep=crcr]{%
1	0.00862191575799999\\
2	0.00888369890600001\\
3	0.008744733476\\
4	0.0088636268\\
5	0.008961177202\\
6	0.00841621748199999\\
7	0.008117016504\\
};
\addlegendentry{CL-OMP}

\addplot [color=blue, line width=0.8pt, mark size=2.5pt, mark=triangle, mark options={solid, rotate=180, blue}]
  table[row sep=crcr]{%
1	0.017105932656\\
2	0.024724270886\\
3	0.017058457434\\
4	0.019313273872\\
5	0.019053346636\\
6	0.014956314238\\
7	0.014175403738\\
};
\addlegendentry{CL-BCD}

\addplot [color=mycolor1, dashdotted, line width=0.8pt, mark size=2.9pt, mark=+, mark options={solid, mycolor1}]
  table[row sep=crcr]{%
1	0.001361722172\\
2	0.001266432044\\
3	0.001272199288\\
4	0.001371466738\\
5	0.001387651478\\
6	0.001125294848\\
7	0.001085796162\\
};
\addlegendentry{SOMP}

\addplot [color=mycolor2, dashdotted, line width=0.8pt, mark size=2.9pt, mark=star, mark options={solid, mycolor2}]
  table[row sep=crcr]{%
1	0.00512928531\\
2	0.00474454694\\
3	0.004307116674\\
4	0.003941686906\\
5	0.00378487659\\
6	0.0032770617\\
7	0.003082790846\\
};
\addlegendentry{SNIHT}

\addplot [color=mycolor3, dashed, line width=0.8pt, mark size=2.9pt, mark=asterisk, mark options={solid, mycolor3}]
  table[row sep=crcr]{%
1	0.655958469792\\
2	0.596388165939999\\
3	0.523101783188\\
4	0.474351238754\\
5	0.427646600836\\
6	0.360559344522\\
7	0.310463910956\\
};
\addlegendentry{SAMV2}

\addplot [color=mycolor4, dashed, line width=0.8pt, mark size=2.9pt, mark=o, mark options={solid, mycolor4}]
  table[row sep=crcr]{%
1	0.644103199754\\
2	0.58363828267\\
3	0.511654864938\\
4	0.463111607076\\
5	0.416690988612\\
6	0.352848751956\\
7	0.304950702604\\
};
\addlegendentry{SBL}

\addplot [color=mycolor5, dashdotted, line width=0.8pt, mark size=2.0pt, mark=square, mark options={solid, mycolor5}]
  table[row sep=crcr]{%
1	0.644103199754\\
2	0.58363828267\\
3	0.511654864938\\
4	0.463111607076\\
5	0.416690988612\\
6	0.352848751956\\
7	0.304950702604\\
};
\addlegendentry{IAA}

\end{groupplot}
\end{tikzpicture}%

%% file: tikz_new/SAMV_simul2_cpu_N=32_K=4_L=32_M=64to512_NRSIM=500_GaussianS=1_04-Jun-2024.tex
\definecolor{mycolor1}{rgb}{0.26600,0.27400,0.28800}%
\definecolor{mycolor2}{rgb}{0.00784,0.31373,0.26275}%
\definecolor{mycolor3}{rgb}{0.46600,0.67400,0.18800}%
\definecolor{mycolor4}{rgb}{0.54510,0.50196,0.00000}%
\definecolor{mycolor5}{rgb}{0.49412,0.11765,0.61176}%
\begin{tikzpicture}
\begin{groupplot}[group style={group size=4 by 1,horizontal sep=0.965cm,vertical sep=0.27cm}]

\nextgroupplot[
width=0.55\columnwidth,
height=5.0cm,
tick label style={font=\scriptsize} , 
title style={font=\footnotesize},
xticklabel=\empty,
scale only axis,
xmin=1,
xmax=7,
xlabel style={font=\color{white!15!black}},
xlabel={SNR},
ymin=0.000403015136,
ymax=0.070292719806,
ylabel style={font=\color{white!15!black}},
ylabel={running time [s]},
axis background/.style={fill=white},
title style={font=\footnotesize},
title={$M=64$},
axis x line*=bottom,
axis y line*=left,
xmajorgrids,
ymajorgrids,
yticklabel=\pgfkeys{/pgf/number format/.cd,fixed,precision=1,zerofill}\pgfmathprintnumber{\tick}, 
legend style={legend cell align=left, align=left, draw=white!15!black}
]
\addplot [color=red, line width=0.8pt, mark size=2.5pt, mark=triangle, mark options={solid, rotate=90, red}]
  table[row sep=crcr]{%
1	0.001572204848\\
2	0.001472969142\\
3	0.001401481024\\
4	0.001410442568\\
5	0.001399221714\\
6	0.001394152828\\
7	0.00142186899\\
};

\addplot [color=blue, line width=0.8pt, mark size=2.5pt, mark=triangle, mark options={solid, rotate=180, blue}]
  table[row sep=crcr]{%
1	0.00327969758\\
2	0.0027029446\\
3	0.00221146073\\
4	0.002458671574\\
5	0.002104109202\\
6	0.002051409784\\
7	0.002133583824\\
};

\addplot [color=mycolor1, dashdotted, line width=0.8pt, mark size=2.5pt, mark=+, mark options={solid, mycolor1}]
  table[row sep=crcr]{%
1	0.000481686861999999\\
2	0.000433755426\\
3	0.000413055076\\
4	0.000416471176\\
5	0.000409720796\\
6	0.000403015136\\
7	0.000420798676\\
};

\addplot [color=mycolor2, dashdotted, line width=0.8pt, mark size=2.5pt, mark=star, mark options={solid, mycolor2}]
  table[row sep=crcr]{%
1	0.002024821878\\
2	0.001750336918\\
3	0.001583851204\\
4	0.001486633576\\
5	0.001410578578\\
6	0.001356577204\\
7	0.001363017484\\
};

\addplot [color=mycolor3, dashed, line width=0.8pt, mark size=2.5pt, mark=asterisk, mark options={solid, mycolor3}]
  table[row sep=crcr]{%
1	0.052652633958\\
2	0.044061513976\\
3	0.038336551638\\
4	0.034446280734\\
5	0.030929969832\\
6	0.027590328606\\
7	0.025241081672\\
};

\addplot [color=mycolor4, dashed, line width=0.8pt, mark size=2.5pt, mark=o, mark options={solid, mycolor4}]
  table[row sep=crcr]{%
1	0.070292719806\\
2	0.059714420038\\
3	0.051961893458\\
4	0.04688289367\\
5	0.041699433102\\
6	0.03741857708\\
7	0.034365286462\\
};

\addplot [color=mycolor5, dashdotted, line width=0.8pt, mark size=3.0pt, mark=square, mark options={solid, mycolor5}]
  table[row sep=crcr]{%
1	0.070292719806\\
2	0.059714420038\\
3	0.051961893458\\
4	0.04688289367\\
5	0.041699433102\\
6	0.03741857708\\
7	0.034365286462\\
};

\nextgroupplot[
width=0.55\columnwidth,
height=5.0cm,
scale only axis,
xticklabel=\empty,
tick label style={font=\scriptsize} , 
xmin=1,
xmax=7,
xlabel style={font=\color{white!15!black}},
xlabel={SNR},
ymin=0.000591034582,
ymax=0.13018409924,
axis background/.style={fill=white},
title style={font=\footnotesize},
title={$M=256$},
axis x line*=bottom,
axis y line*=left,
xmajorgrids,
ymajorgrids,
yticklabel=\pgfkeys{/pgf/number format/.cd,fixed,precision=2,zerofill}\pgfmathprintnumber{\tick}, 
legend style={legend cell align=left, align=left, draw=white!15!black}
]
\addplot [color=red, line width=0.8pt, mark size=2.9pt, mark=triangle, mark options={solid, rotate=90, red}]
  table[row sep=crcr]{%
1	0.001918635924\\
2	0.001854620318\\
3	0.001851362058\\
4	0.001871881378\\
5	0.001909721852\\
6	0.00190856054\\
7	0.002019350328\\
};

\addplot [color=blue, line width=0.8pt, mark size=2.9pt, mark=triangle, mark options={solid, rotate=180, blue}]
  table[row sep=crcr]{%
1	0.00344108895\\
2	0.004272347986\\
3	0.00333792369\\
4	0.003407402824\\
5	0.003481620832\\
6	0.003446297208\\
7	0.003762018164\\
};

\addplot [color=mycolor1, dashdotted, line width=0.8pt, mark size=2.5pt, mark=+, mark options={solid, mycolor1}]
  table[row sep=crcr]{%
1	0.000632915716000001\\
2	0.000592093354\\
3	0.000592764814\\
4	0.000591034582\\
5	0.000630338958\\
6	0.000617934995999999\\
7	0.000700045712\\
};

\addplot [color=mycolor2, dashdotted, line width=0.8pt, mark size=2.5pt, mark=star, mark options={solid, mycolor2}]
  table[row sep=crcr]{%
1	0.004591928434\\
2	0.004098435024\\
3	0.003778800006\\
4	0.003478487678\\
5	0.003322719864\\
6	0.00318692104\\
7	0.003542744624\\
};

\addplot [color=mycolor3, dashed, line width=0.8pt, mark size=2.5pt, mark=asterisk, mark options={solid, mycolor3}]
  table[row sep=crcr]{%
1	0.106365461856\\
2	0.090837112218\\
3	0.0810502112279999\\
4	0.073221080568\\
5	0.0683611962100001\\
6	0.061203749344\\
7	0.061872510628\\
};

\addplot [color=mycolor4, dashed, line width=0.8pt, mark size=2.5pt, mark=o, mark options={solid, mycolor4}]
  table[row sep=crcr]{%
1	0.13018409924\\
2	0.111026678628\\
3	0.098538470726\\
4	0.0888345531780001\\
5	0.0832884704679999\\
6	0.07422396939\\
7	0.076679315438\\
};

\addplot [color=mycolor5, dashdotted, line width=0.8pt, mark size=3.0pt, mark=square, mark options={solid, mycolor5}]
  table[row sep=crcr]{%
1	0.13018409924\\
2	0.111026678628\\
3	0.098538470726\\
4	0.0888345531780001\\
5	0.0832884704679999\\
6	0.07422396939\\
7	0.076679315438\\
};


\nextgroupplot[
width=0.55\columnwidth,
height=5.0cm,
scale only axis,
xticklabel=\empty,
tick label style={font=\scriptsize} , 
xmin=1,
xmax=7,
xlabel style={font=\color{white!15!black}},
xlabel={SNR},
ymin=0.000844874374,
ymax=0.222555002342,
axis background/.style={fill=white},
title style={font=\footnotesize},
title={$M=512$},
axis x line*=bottom,
axis y line*=left,
xmajorgrids,
ymajorgrids,
yticklabel=\pgfkeys{/pgf/number format/.cd,fixed,precision=2,zerofill}\pgfmathprintnumber{\tick}, 
legend style={at={(0.4,0.6)}, legend cell align=left, align=left, draw=white!15!black,font=\scriptsize} %
]
\addplot [color=red, line width=0.8pt, mark size=3.0pt, mark=triangle, mark options={solid, rotate=90, red}]
  table[row sep=crcr]{%
1	0.002379252438\\
2	0.002349015216\\
3	0.00236361743\\
4	0.002392963436\\
5	0.002403007686\\
6	0.0024961001\\
7	0.00243256348\\
};
\addlegendentry{CL-OMP}

\addplot [color=blue, line width=0.8pt, mark size=3.0pt, mark=triangle, mark options={solid, rotate=180, blue}]
  table[row sep=crcr]{%
1	0.004745818034\\
2	0.005268392092\\
3	0.004641755406\\
4	0.004685163862\\
5	0.004746936922\\
6	0.004993930626\\
7	0.004809210302\\
};
\addlegendentry{CL-BCD}

\addplot [color=mycolor1, dashdotted, line width=0.8pt, mark size=2.5pt, mark=+, mark options={solid, mycolor1}]
  table[row sep=crcr]{%
1	0.000877999932\\
2	0.000847470002\\
3	0.000844874374\\
4	0.000854487847999999\\
5	0.00086588707\\
6	0.000919253597999999\\
7	0.00087233167\\
};
\addlegendentry{SOMP}

\addplot [color=mycolor2, dashdotted, line width=0.8pt, mark size=2.5pt, mark=star, mark options={solid, mycolor2}]
  table[row sep=crcr]{%
1	0.00811422024200001\\
2	0.007214023488\\
3	0.006776622464\\
4	0.005952626834\\
5	0.005617724718\\
6	0.005567896424\\
7	0.005098542452\\
};
\addlegendentry{SNIHT}

\addplot [color=mycolor3, dashed, line width=0.8pt, mark size=2.5pt, mark=asterisk, mark options={solid, mycolor3}]
  table[row sep=crcr]{%
1	0.177485383872\\
2	0.15387025227\\
3	0.137995772944\\
4	0.125988329108\\
5	0.115680137386\\
6	0.111146628422\\
7	0.0967823130600001\\
};
\addlegendentry{SAMV2}

\addplot [color=mycolor4, dashed, line width=0.8pt, mark size=2.5pt, mark=o, mark options={solid, mycolor4}]
  table[row sep=crcr]{%
1	0.222555002342\\
2	0.194527252912\\
3	0.172971662518\\
4	0.156682091222\\
5	0.142094038394\\
6	0.132964367274\\
7	0.115930773674\\
};
\addlegendentry{SBL}

\addplot [color=mycolor5, dashdotted, line width=0.8pt, mark size=3.0pt, mark=square, mark options={solid, mycolor5}]
  table[row sep=crcr]{%
1	0.222555002342\\
2	0.194527252912\\
3	0.172971662518\\
4	0.156682091222\\
5	0.142094038394\\
6	0.132964367274\\
7	0.115930773674\\
};
\addlegendentry{IAA}
\end{groupplot}

\end{tikzpicture}%

%% file: tikz/onesource_snr_vs_MSEdoa_theta=-25.00_N=20_M=1801_L=25_LL=1000_nmethods=7_27-Nov-2023.tex
\definecolor{mycolor1}{rgb}{0.30100,0.74500,0.93300}%
\definecolor{mycolor2}{rgb}{0.92900,0.69400,0.12500}%
\definecolor{mycolor3}{rgb}{0.46600,0.67400,0.18800}%
\definecolor{mycolor4}{rgb}{0.34200,0.87200,0.50300}%
\definecolor{mycolor5}{rgb}{1.00000,0.00000,1.00000}%

\begin{tikzpicture}

\begin{axis}[%
width= 0.78\columnwidth, 
height= 4.7cm, 
 tick label style={font=\scriptsize} , 
scale only axis,
xmin=-14,
xmax=-2,
xlabel style={font=\small \color{white!15!black}},
xlabel={},
ymode=log,
ymin=0.144331517421058,
ymax=31.7654560804658,
yminorticks=true,
ylabel style={font=\small \color{white!15!black}},
ylabel={RMSE of DOA (deg.)},
axis background/.style={fill=white},
xmajorgrids,
ymajorgrids,
yminorgrids,
xticklabel=\empty,
legend style={legend cell align=left, align=left, draw=white!15!black,font=\scriptsize,at={(1.0,1.0)}}
]
\addplot [color=blue, line width=1.0pt, mark size=2.5pt, mark=triangle, mark options={solid, rotate=180, blue}]
  table[row sep=crcr]{%
-2	0.15533834040571\\
-2.8	0.173407035612746\\
-3.6	0.19028925350634\\
-4.4	0.212132034355964\\
-5.2	0.2346699810372\\
-6	0.262507142759964\\
-6.8	0.294686273857472\\
-7.6	0.330439101802437\\
-8.4	0.379631400176541\\
-9.2	0.436875268240262\\
-10	3.35460578906078\\
-10.8	8.42721009587395\\
-11.6	12.1511542661593\\
-12.4	18.7708028597607\\
-13.2	24.5569002522712\\
-14	31.6136722004894\\
};
\addlegendentry{CL-BCD}


\addplot [color=red, line width=1.0pt, mark size=2.8pt, mark=triangle, mark options={solid, rotate=90, red}]
  table[row sep=crcr]{%
-2	0.153525242224203\\
-2.8	0.170205757834452\\
-3.6	0.187962762269552\\
-4.4	0.207557221025914\\
-5.2	0.231905153025973\\
-6	0.258689002472081\\
-6.8	0.291376045686671\\
-7.6	0.32671088136149\\
-8.4	0.376164857476082\\
-9.2	0.433347435667964\\
-10	1.47906727365594\\
-10.8	3.85023765500263\\
-11.6	7.24013190487577\\
-12.4	10.1673693746219\\
-13.2	14.9112913592351\\
-14	19.9746073803717\\
};
\addlegendentry{CL-OMP}

\addplot [color=mycolor2, dashed, line width=1.0pt, mark size=2.5pt, mark=o, mark options={solid, mycolor2}]
  table[row sep=crcr]{%
-2	0.170058813355851\\
-2.8	0.190184121314058\\
-3.6	0.212344060430237\\
-4.4	0.238390436049771\\
-5.2	0.265084892062902\\
-6	0.300582767303784\\
-6.8	2.7150653767451\\
-7.6	2.99333593169895\\
-8.4	3.70750455158184\\
-9.2	4.8358453242427\\
-10	6.40227147190746\\
-10.8	8.11358798559553\\
-11.6	11.2556990009506\\
-12.4	14.1591048445868\\
-13.2	17.3459286289319\\
-14	21.250614579348\\
};
\addlegendentry{SBL}

\addplot [color=mycolor3, dashed, line width=1.0pt, mark size=2.8pt, mark=star, mark options={solid, mycolor3}]
  table[row sep=crcr]{%
-2	0.174527934726794\\
-2.8	0.193597520645282\\
-3.6	0.217347647790355\\
-4.4	0.242404620418011\\
-5.2	0.274972725920228\\
-6	0.308074666274267\\
-6.8	2.29003056748158\\
-7.6	3.62645281232226\\
-8.4	4.83586807098788\\
-9.2	5.62028735920149\\
-10	7.92041539314701\\
-10.8	10.0262021723083\\
-11.6	12.4256017962914\\
-12.4	15.6204593402371\\
-13.2	18.5514837681518\\
-14	21.8414463806773\\
};
\addlegendentry{SAMV2}

\addplot [color=orange, dashed, line width=1.0pt, mark size=2.5pt, mark=square, mark options={solid, orange}]
  table[row sep=crcr]{%
-2	0.155273951453553\\
-2.8	0.173464693814045\\
-3.6	0.190052624291274\\
-4.4	0.212320512433444\\
-5.2	0.23473389188611\\
-6	0.262564277844493\\
-6.8	0.294618397253126\\
-7.6	0.330393704540509\\
-8.4	0.380328805114734\\
-9.2	0.437264222181509\\
-10	3.3546505034057\\
-10.8	8.42721306245427\\
-11.6	12.151983788666\\
-12.4	19.0885994771749\\
-13.2	24.5592626110802\\
-14	31.7654560804658\\
};
\addlegendentry{IAA}

\addplot [color=mycolor4, line width=1.0pt, mark size=2.5pt, mark=diamond, mark options={solid, mycolor4}]
  table[row sep=crcr]{%
-2	0.153687995627506\\
-2.8	0.172104619345328\\
-3.6	0.188042548376691\\
-4.4	0.211849002829846\\
-5.2	0.233452350598575\\
-6	0.263761255683999\\
-6.8	0.298127489507427\\
-7.6	0.341408845813933\\
-8.4	1.20696313116847\\
-9.2	3.47383793519503\\
-10	5.18856145766821\\
-10.8	8.34206029707289\\
-11.6	13.6999963503645\\
-12.4	19.3000474092682\\
-13.2	23.4768635043099\\
-14	28.5580277680375\\
};
\addlegendentry{MUSIC}

\addplot [color=mycolor5, line width=1.0pt, mark size=2.8pt, mark=+, mark options={solid, mycolor5}]
  table[row sep=crcr]{%
-2	0.151667900422665\\
-2.8	0.168168803461097\\
-3.6	0.186917721472199\\
-4.4	0.208395675731806\\
-5.2	0.233258536897253\\
-6	0.262432027837696\\
-6.8	0.297260404250697\\
-7.6	0.509852304808686\\
-8.4	2.55909005603818\\
-9.2	5.29854366545462\\
-10	7.60047538630686\\
-10.8	10.3967624182315\\
-11.6	15.9956626685649\\
-12.4	22.3980347710983\\
-13.2	26.9004266739098\\
-14	31.8931182361707\\
};
\addlegendentry{R-MUSIC}

\addplot [color=cyan, dashed, line width=1.2pt, mark size=3.5pt, mark=|, mark options={solid, cyan}]
  table[row sep=crcr]{%
-2	0.151542403306797\\
-2.8	0.167735804168341\\
-3.6	0.186134360073577\\
-4.4	0.20691592495504\\
-5.2	0.230592714542328\\
-6	0.25783870927384\\
-6.8	0.289427192917321\\
-7.6	0.326959018838753\\
-8.4	0.374479238409822\\
-9.2	0.432354368545063\\
-10	1.47730233872421\\
-10.8	3.84973522466156\\
-11.6	7.23956823988833\\
-12.4	10.1671696700704\\
-13.2	14.9125653829246\\
-14	19.9745623481467\\
};
\addlegendentry{MLE}

\addplot [color=black, line width=1.2pt]
  table[row sep=crcr]{%
-2	0.144331517421058\\
-2.8	0.159427201006278\\
-3.6	0.17633954353049\\
-4.4	0.195351740527268\\
-5.2	0.216805141283621\\
-6	0.241112940431716\\
-6.8	0.268776989117953\\
-7.6	0.300408307564667\\
-8.4	0.336751949337351\\
-9.2	0.378716944044976\\
-10	0.427412139459791\\
-10.8	0.484188891653789\\
-11.6	0.5506917329059\\
-12.4	0.628918403887093\\
-13.2	0.72129098857241\\
-14	0.830740350339569\\
};
\addlegendentry{CRLB}

\end{axis}
\end{tikzpicture}%

%% file: tikz/onesource_snr_vs_MSEgamma__theta=-25.00_N=20_M=1801_L=25_LL=1000_nmethods=7_27-Nov-2023.tex
\definecolor{mycolor1}{rgb}{0.30100,0.74500,0.93300}%
\definecolor{mycolor2}{rgb}{0.92900,0.69400,0.12500}%
\definecolor{mycolor3}{rgb}{0.46600,0.67400,0.18800}%
\definecolor{mycolor4}{rgb}{1.00000,0.64706,0.00000}%
\begin{tikzpicture}

\begin{axis}[%
width= 0.78\columnwidth, 
height= 4.7cm, 
 tick label style={font=\scriptsize} , 
scale only axis,
xmin=-14,
xmax=-2,
xlabel style={font=\small \color{white!15!black}},
xlabel={SNR (dB)},
ymin=0.211404539287313,
ymax=1.40246884506856,
ylabel style={font=\small \color{white!15!black}},
ylabel={Root NMSE of $\hat \gamma$},
axis background/.style={fill=white},
xmajorgrids,
ymajorgrids,
legend style={legend cell align=left, align=left, draw=white!15!black,font=\scriptsize,at={(1.0,1.0)}}
]
\addplot [color=blue, line width=1.0pt, mark size=2.5pt, mark=triangle, mark options={solid, rotate=180, blue}]
  table[row sep=crcr]{%
-2	0.225273571706401\\
-2.8	0.233724011929479\\
-3.6	0.244969803142989\\
-4.4	0.259781700296373\\
-5.2	0.27923392912607\\
-6	0.304510669209365\\
-6.8	0.337075146085164\\
-7.6	0.378580156940595\\
-8.4	0.430937668178998\\
-9.2	0.496372680895655\\
-10	0.577598092357972\\
-10.8	0.677890413330678\\
-11.6	0.801827905273843\\
-12.4	0.956582528351925\\
-13.2	1.1522406876814\\
-14	1.40190551213633\\
};
\addlegendentry{CL-BCD}

\addplot [color=red, line width=1.0pt, mark size=2.8pt, mark=triangle, mark options={solid, rotate=90, red}]
  table[row sep=crcr]{%
-2	0.211404539287313\\
-2.8	0.214335666268326\\
-3.6	0.217876121609759\\
-4.4	0.222145818996785\\
-5.2	0.227289166151866\\
-6	0.23348779538153\\
-6.8	0.240944675321894\\
-7.6	0.249900206802786\\
-8.4	0.260628241601095\\
-9.2	0.273415843315441\\
-10	0.288615142961534\\
-10.8	0.306140655352345\\
-11.6	0.325258968271488\\
-12.4	0.345462699606557\\
-13.2	0.365374729986318\\
-14	0.387314683312987\\
};
\addlegendentry{CL-OMP}

\addplot [color=mycolor2, dashed, line width=1.0pt, mark size=2.5pt, mark=o, mark options={solid, mycolor2}]
  table[row sep=crcr]{%
-2	0.617435426971236\\
-2.8	0.640917265514221\\
-3.6	0.66281154764782\\
-4.4	0.683385030884299\\
-5.2	0.702835624446067\\
-6	0.721084099310109\\
-6.8	0.737604612531123\\
-7.6	0.753106228794092\\
-8.4	0.766742357414984\\
-9.2	0.779027037718557\\
-10	0.789854146538981\\
-10.8	0.798472346358566\\
-11.6	0.805076392857122\\
-12.4	0.809140044833767\\
-13.2	0.809495233266393\\
-14	0.805340989575452\\
};
\addlegendentry{SBL}

\addplot [color=mycolor3, dashed, line width=1.0pt, mark size=2.8pt, mark=star, mark options={solid, mycolor3}]
  table[row sep=crcr]{%
-2	0.637830471886107\\
-2.8	0.659847195079435\\
-3.6	0.680330483247524\\
-4.4	0.69930437594705\\
-5.2	0.717191539976546\\
-6	0.73401740963137\\
-6.8	0.749211027981336\\
-7.6	0.763198788289634\\
-8.4	0.775502599069268\\
-9.2	0.786308890082303\\
-10	0.795662195743692\\
-10.8	0.802887896294969\\
-11.6	0.808071510274618\\
-12.4	0.810181195364849\\
-13.2	0.80862885445993\\
-14	0.802133305405494\\
};
\addlegendentry{SAMV2}

\addplot [color=orange, dashed, line width=0.9pt, mark size=2.5pt, mark=square, mark options={solid, orange}]
  table[row sep=crcr]{%
-2	0.225320159121887\\
-2.8	0.233775722350656\\
-3.6	0.245027334762188\\
-4.4	0.259843978878145\\
-5.2	0.279301282588453\\
-6	0.304582175621347\\
-6.8	0.337149396658526\\
-7.6	0.378656212962597\\
-8.4	0.431014474337194\\
-9.2	0.496449910476266\\
-10	0.577674182862698\\
-10.8	0.677970434310246\\
-11.6	0.801924978012594\\
-12.4	0.956745249137457\\
-13.2	1.15253635099886\\
-14	1.40246884506856\\
};
\addlegendentry{IAA}

\addplot [color=cyan, dashed, line width=1.3pt, mark size=3.5pt, mark=|, mark options={solid,cyan}]
  table[row sep=crcr]{%
-2	0.211419237813491\\
-2.8	0.214352037310961\\
-3.6	0.217891760920614\\
-4.4	0.222160941254322\\
-5.2	0.227306540650329\\
-6	0.233503258645829\\
-6.8	0.240958691289994\\
-7.6	0.24991561338434\\
-8.4	0.260643709545586\\
-9.2	0.273435777316111\\
-10	0.288638692594633\\
-10.8	0.306167527102401\\
-11.6	0.325288732260629\\
-12.4	0.345494113846945\\
-13.2	0.365414783090422\\
-14	0.387357588368255\\
};
\addlegendentry{MLE}

\addplot [color=black, line width=1.2pt]
  table[row sep=crcr]{%
-2	0.215879554075642\\
-2.8	0.219098220663534\\
-3.6	0.222970624784067\\
-4.4	0.227630001370447\\
-5.2	0.233236859081474\\
-6	0.239984573349587\\
-6.8	0.248106113039155\\
-7.6	0.25788212083418\\
-8.4	0.269650607888969\\
-9.2	0.283818572536995\\
-10	0.300875914272767\\
-10.8	0.321412091972646\\
-11.6	0.346136073937595\\
-12.4	0.375900250901607\\
-13.2	0.411729134512984\\
-14	0.454853844478911\\
};
\addlegendentry{CRLB}

\end{axis}
\end{tikzpicture}%

%% file: tikz/onesource_snr_vs_MSEdoa_theta=-25.00_N=20_M=1801_L=200_LL=1000_nmethods=7_27-Nov-2023.tex
\definecolor{mycolor1}{rgb}{0.30100,0.74500,0.93300}%
\definecolor{mycolor2}{rgb}{0.92900,0.69400,0.12500}%
\definecolor{mycolor3}{rgb}{0.46600,0.67400,0.18800}%
\definecolor{mycolor4}{rgb}{0.34200,0.87200,0.50300}%
\definecolor{mycolor5}{rgb}{1.00000,0.00000,1.00000}%

\begin{tikzpicture}

\begin{axis}[%
width= 0.78\columnwidth, 
height= 4.7cm, 
 tick label style={font=\scriptsize} , 
 scale only axis,
xmin=-19,
xmax=-8.6,
xlabel style={font=\small\color{white!15!black}},
ymode=log,
ymin=0.0976819529256765,
ymax=56.7131448607816,
yminorticks=true,
ylabel style={font=\small \color{white!15!black}},
ylabel={RMSE of $\hat \theta$ (deg)},
axis background/.style={fill=white},
xmajorgrids,
ymajorgrids,
yminorgrids,
xticklabel=\empty,
legend style={legend cell align=left, align=left, draw=white!15!black,font=\footnotesize}
]
\addplot [color=blue, line width=1.0pt, mark size=2.5pt, mark=triangle, mark options={solid, rotate=180, blue}]
  table[row sep=crcr]{%
-7	0.103440804327886\\
-7.8	0.113929802949009\\
-8.6	0.128802173894698\\
-9.4	0.144844744467999\\
-10.2	0.162788205960998\\
-11	0.184499322491981\\
-11.8	0.208686367547092\\
-12.6	0.240748831772866\\
-13.4	0.277452698671324\\
-14.2	0.32138761643847\\
-15	6.80569100091976\\
-15.8	11.9072452733619\\
-16.6	22.3283615610281\\
-17.4	35.3160586702424\\
-18.2	45.7312440241899\\
-19	56.2949562571994\\
};
\addlegendentry{CL-BCD}


\addplot [color=red, line width=1.0pt, mark size=2.9pt, mark=triangle, mark options={solid, rotate=90, red}]
  table[row sep=crcr]{%
-7	0.103392456204503\\
-7.8	0.113357840487547\\
-8.6	0.125698050899766\\
-9.4	0.144706599711279\\
-10.2	0.159687194226714\\
-11	0.182071414560332\\
-11.8	0.206784912409006\\
-12.6	0.235584379787795\\
-13.4	0.273441035691428\\
-14.2	0.319014106271183\\
-15	0.371119926708337\\
-15.8	0.483580396625007\\
-16.6	5.45340077382912\\
-17.4	9.1545529656013\\
-18.2	14.017190874066\\
-19	19.3157919847983\\
};
\addlegendentry{CL-OMP}

\addplot [color=mycolor2, dashed, line width=1.0pt, mark size=2.5pt, mark=o, mark options={solid, mycolor2}]
  table[row sep=crcr]{%
-7	0.110589330407595\\
-7.8	0.122392810246354\\
-8.6	0.137331715200824\\
-9.4	0.160934769394312\\
-10.2	0.180249826629598\\
-11	0.208614476966485\\
-11.8	0.246008129946959\\
-12.6	0.284517134809138\\
-13.4	0.336585798868583\\
-14.2	1.0871292471459\\
-15	2.93992516911571\\
-15.8	5.84890759031121\\
-16.6	9.17392009993547\\
-17.4	13.4612257985668\\
-18.2	17.2604539337759\\
-19	20.2293267312583\\
};
\addlegendentry{SBL}

\addplot [color=mycolor3, dashed, line width=1.0pt, mark size=2.8pt, mark=star, mark options={solid, mycolor3}]
  table[row sep=crcr]{%
-7	0.110498868772491\\
-7.8	0.122800651464071\\
-8.6	0.137113092008022\\
-9.4	0.160187390265277\\
-10.2	0.178969271105406\\
-11	0.206421898063166\\
-11.8	0.243782690115603\\
-12.6	0.282063822565038\\
-13.4	0.332114438108315\\
-14.2	1.08527415891101\\
-15	2.93924139872859\\
-15.8	5.91620655487958\\
-16.6	9.28641642400339\\
-17.4	13.1772690645672\\
-18.2	17.2799872685138\\
-19	20.2787696865466\\
};
\addlegendentry{SAMV2}

\addplot [color=orange, dashed, line width=0.9pt, mark size=2.5pt, mark=square, mark options={solid, orange}]
  table[row sep=crcr]{%
-7	0.103489129863962\\
-7.8	0.11379806676741\\
-8.6	0.128685663537163\\
-9.4	0.144775688566831\\
-10.2  0.163033738839543\\
-11	0.18452642087246\\
-11.8	 0.208710325571113\\
-12.6	 0.24054105678657\\
-13.4	 0.277092042469647\\
-14.2	 0.321247568084182\\
-15	6.80570349045565\\
-15.8	 12.0130212686068\\
-16.6	 23.097937570268\\
-17.4 	35.5755331653653\\
-18.2	 46.3755882981553\\
-19	 56.7131448607816\\
};
\addlegendentry{IAA}

\addplot [color=mycolor4, line width=1.0pt, mark size=2.5pt, mark=diamond, mark options={solid, mycolor4}]
  table[row sep=crcr]{%
-7	0.104211323760905\\
-7.8	0.114017542509914\\
-8.6	0.127906215642556\\
-9.4	0.145567853594124\\
-10.2	0.163462533933621\\
-11	0.187723200484118\\
-11.8	0.215569014471004\\
-12.6	0.249319072675959\\
-13.4	0.296293773137407\\
-14.2	0.353765459026177\\
-15	4.69431145110761\\
-15.8	8.52335028025951\\
-16.6	14.0877052780075\\
-17.4	20.2118007609416\\
-18.2	26.15937078754\\
-19	30.5767344561188\\
};
\addlegendentry{MUSIC}

\addplot [color=mycolor5, dashed, line width=1.0pt, mark size=2.8pt, mark=+, mark options={solid, mycolor5}]
  table[row sep=crcr]{%
-7	0.0985040592985926\\
-7.8	0.110520971055548\\
-8.6	0.124526863817934\\
-9.4	0.140996103515637\\
-10.2	0.160557300927738\\
-11	0.184062656948516\\
-11.8	0.212703793139129\\
-12.6	0.248222461194876\\
-13.4	0.293369004331033\\
-14.2	0.353676758382948\\
-15	4.72065975646681\\
-15.8	9.31433095583727\\
-16.6	14.2027616776688\\
-17.4	21.0318181395102\\
-18.2	27.0639423085323\\
-19	31.6888300187869\\
};
\addlegendentry{R-MUSIC}

\addplot [color=cyan, dashed, line width=1.3pt, mark size=3.5pt, mark=|, mark options={solid,cyan}]
  table[row sep=crcr]{%
-7	0.0980545766397468\\
-7.8	0.109891309938502\\
-8.6	0.123412317051419\\
-9.4	0.139011869996774\\
-10.2	0.157674347945377\\
-11	0.179462252298359\\
-11.8	0.20514848281184\\
-12.6	0.236048935604463\\
-13.4	0.27268223264452\\
-14.2	0.317102822440924\\
-15	0.370822329424753\\
-15.8	0.482888599989687\\
-16.6	5.45392759394547\\
-17.4	9.15274056225783\\
-18.2	14.0162555984114\\
-19	19.3159824937796\\
};
\addlegendentry{MLE}

\addplot [color=black, line width=0.8pt]
  table[row sep=crcr]{%
-7	0.0976819529256765\\
-7.8	0.109255063266972\\
-8.6	0.122568687069943\\
-9.4	0.137960599991531\\
-10.2	0.15584275607671\\
-11	0.176717045666451\\
-11.8	0.201193986704537\\
-12.6	0.230014867372018\\
-13.4	0.264077991651045\\
-14.2	0.304469851957514\\
-15	0.352502265137302\\
-15.8	0.409756758680133\\
-16.6	0.478137782165957\\
-17.4	0.55993664713149\\
-18.2	0.657908474466995\\
-19	0.775364865909318\\
};
\addlegendentry{CRLB}

\end{axis}
\end{tikzpicture}%

%% file: tikz/onesource_snr_vs_MSEgamma__theta=-25.00_N=20_M=1801_L=200_LL=1000_nmethods=7_27-Nov-2023.tex
\definecolor{mycolor1}{rgb}{0.30100,0.74500,0.93300}%
\definecolor{mycolor2}{rgb}{0.92900,0.69400,0.12500}%
\definecolor{mycolor3}{rgb}{0.46600,0.67400,0.18800}%
\definecolor{mycolor4}{rgb}{1.00000,0.64706,0.00000}%
\begin{tikzpicture}

\begin{axis}[%
width= 0.78\columnwidth, 
height= 4.7cm, 
 tick label style={font=\scriptsize} , 
scale only axis,
xmin=-19,
xmax=-8.6,
xlabel style={font=\small \color{white!15!black}},
xlabel={SNR (dB)},
ymin=0.0885237117426931,
ymax=4.23516818794506,
ylabel style={font=\small\color{white!15!black}},
ylabel={Root NMSE of $\hat \gamma$},
axis background/.style={fill=white},
xmajorgrids,
ymajorgrids,
legend style={legend cell align=left, align=left, draw=white!15!black,font=\footnotesize}
]
\addplot [color=blue, line width=1.0pt, mark size=2.5pt, mark=triangle, mark options={solid, rotate=180, blue}]
  table[row sep=crcr]{%
-7	0.269922255883377\\
-7.8	0.319029597242011\\
-8.6	0.378628351330418\\
-9.4	0.450798000532048\\
-10.2	0.538032225237281\\
-11	0.643354445318443\\
-11.8	0.770410612957527\\
-12.6	0.923607132569928\\
-13.4	1.10826567826858\\
-14.2	1.33078171818818\\
-15	1.59907633765058\\
-15.8	1.92376514168299\\
-16.6	2.32146430045274\\
-17.4	2.817938580838\\
-18.2	3.44220680437327\\
-19	4.22926122475551\\
};
\addlegendentry{CL-BCD}

\addplot [color=red, line width=1.0pt, mark size=2.9pt, mark=triangle, mark options={solid, rotate=90, red}]
  table[row sep=crcr]{%
-7	0.0889597931155574\\
-7.8	0.0928175655724508\\
-8.6	0.0974520042160049\\
-9.4	0.103016499188423\\
-10.2	0.109689240874284\\
-11	0.117684522350575\\
-11.8	0.127265456894546\\
-12.6	0.138732731174125\\
-13.4	0.152439791101815\\
-14.2	0.168819233282521\\
-15	0.188382255951675\\
-15.8	0.211646648468412\\
-16.6	0.239085515307267\\
-17.4	0.268897522383682\\
-18.2	0.298145311094112\\
-19	0.325811013408601\\
};
\addlegendentry{CL-OMP}

\addplot [color=mycolor2, dashed, line width=1.0pt, mark size=2.5pt, mark=o, mark options={solid, mycolor2}]
  table[row sep=crcr]{%
-7	0.813064935748295\\
-7.8	0.825982547667653\\
-8.6	0.837765082691442\\
-9.4	0.848295541550042\\
-10.2	0.857713446921974\\
-11	0.866185247966104\\
-11.8	0.87361686685951\\
-12.6	0.880226543609761\\
-13.4	0.885978441786756\\
-14.2	0.890937549785301\\
-15	0.895061119729682\\
-15.8	0.898360797377769\\
-16.6	0.900627626846296\\
-17.4	0.901483491541602\\
-18.2	0.90039498222725\\
-19	0.896594398289204\\
};
\addlegendentry{SBL}

\addplot [color=mycolor3, dashed, line width=1.0pt, mark size=2.5pt, mark=star, mark options={solid, mycolor3}]
  table[row sep=crcr]{%
-7	0.811921750962867\\
-7.8	0.824374826405855\\
-8.6	0.835774200027057\\
-9.4	0.846011057419327\\
-10.2	0.85525877081314\\
-11	0.863684365200406\\
-11.8	0.871188860797363\\
-12.6	0.877951618218234\\
-13.4	0.883947227123917\\
-14.2	0.889182366923972\\
-15	0.893619768224972\\
-15.8	0.897219220063464\\
-16.6	0.899745300886403\\
-17.4	0.900777162324722\\
-18.2	0.899793523166742\\
-19	0.89606247121731\\
};
\addlegendentry{SAMV2}

\addplot [color=orange, dashed, line width=0.9pt, mark size=2.5pt, mark=square, mark options={solid,orange}]
  table[row sep=crcr]{%
-7	0.270023933652605\\
-7.8	0.319132807858628\\
-8.6	0.37873262896302\\
-9.4	0.450903002894432\\
-10.2	0.538138196928096\\
-11	0.643461845289568\\
-11.8	0.770520010218817\\
-12.6	0.923720116423918\\
-13.4	1.10838436405029\\
-14.2	1.33090853676729\\
-15	1.59923743352581\\
-15.8	1.92404148849007\\
-16.6	2.32219752896939\\
-17.4	2.81973832161055\\
-18.2	3.44564826053516\\
-19	4.23516818794506\\
};
\addlegendentry{IAA}

\addplot [color=cyan, dashed, line width=1.3pt, mark size=3.5pt, mark=|, mark options={solid, cyan}]
  table[row sep=crcr]{%
-7	0.0889655045347249\\
-7.8	0.0928249550497132\\
-8.6	0.0974601741363501\\
-9.4	0.103023692160921\\
-10.2	0.109696767935854\\
-11	0.117695212677961\\
-11.8	0.127274952397604\\
-12.6	0.138740060346817\\
-13.4	0.152451036014186\\
-14.2	0.168834866498324\\
-15	0.188396222513\\
-15.8	0.211666692675751\\
-16.6	0.239102752821261\\
-17.4	0.268921373285085\\
-18.2	0.29817399713969\\
-19	0.325849583259096\\
};
\addlegendentry{MLE}

\addplot [color=black, line width=0.8pt]
  table[row sep=crcr]{%
-7	0.0885237117426931\\
-7.8	0.0921440839824458\\
-8.6	0.0965024141267074\\
-9.4	0.101749429177645\\
-10.2	0.108066540653729\\
-11	0.115672005422768\\
-11.8	0.124828305224803\\
-12.6	0.135850994599268\\
-13.4	0.149119323001461\\
-14.2	0.165089003420929\\
-15	0.184307577266058\\
-15.8	0.207432914861524\\
-16.6	0.235255495248491\\
-17.4	0.268725232422034\\
-18.2	0.308983763515482\\
-19	0.357403294552935\\
};
\addlegendentry{CRLB}

\end{axis}
\end{tikzpicture}%

%% file: tikz/uusiRMSE_vs_DOAsrc_SNR=-8_K=1_N=20_M=1801_L=25_LL=1500_doa=-80to-30_04-Dec-2023.tex
\definecolor{mycolor1}{rgb}{0.92900,0.69400,0.12500}%
\definecolor{mycolor2}{rgb}{0.46600,0.67400,0.18800}%
\definecolor{mycolor3}{rgb}{1.00000,0.64706,0.00000}%
\definecolor{mycolor4}{rgb}{0.34200,0.87200,0.50300}%
\definecolor{mycolor5}{rgb}{1.00000,0.00000,1.00000}%

\begin{tikzpicture}
\begin{groupplot}[group style={group size=1 by 2,vertical sep=0.27cm}]
\nextgroupplot[
width= 0.78\columnwidth, 
height= 4.7cm, 
scale only axis,
 tick label style={font=\scriptsize} , 
xmin=-80,
xmax=-30,
xlabel style={font=\small \color{white!15!black}},
ymode=log,
ymin=0.3,
ymax=10,
yminorticks=true,
ylabel style={font=\small \color{white!15!black}},
ylabel={RMSE of $\hat \theta$ (deg.)},
axis background/.style={fill=white},
 tick label style={font=\scriptsize} , 
 xticklabel=\empty,
xmajorgrids,
ymajorgrids,
yminorgrids,
legend style={anchor=north west, legend cell align=left, align=left,font=\footnotesize,at={(0.01,-1.27)}, legend columns=4,draw=none} 
]
\addplot [color=blue,  line width=1.0pt, mark size=3.1pt, mark=triangle, mark options={solid, rotate=180}]
  table[row sep=crcr]{%
-80	10.4264218854472\\
-70	1.90046134749785\\
-60	0.708114868271148\\
-50	0.500033332222296\\
-40	0.405240669232493\\
-30	0.36053663706573\\
};
\addlegendentry{CL-BCD}

\addplot [color=red, line width=1.0pt, mark size=3.1pt, mark=triangle, mark options={solid, rotate=90, red}]
  table[row sep=crcr]{%
-80	7.58181816365089\\
-70	0.872066511224918\\
-60	0.618832233592703\\
-50	0.476046216243758\\
-40	0.393293783322339\\
-30	0.349704637277421\\
};
\addlegendentry{CL-OMP}

\addplot [color=mycolor1, dashed, line width=1.0pt, mark size=2.8pt, mark=o, mark options={solid, mycolor1}]
  table[row sep=crcr]{%
-80	40.1257961914776\\
-70	18.3615627512112\\
-60	10.7280684810143\\
-50	5.98552643187437\\
-40	3.52122895213211\\
-30	1.73416838859438\\
};
\addlegendentry{SBL}

\addplot [color=mycolor2, dashed, line width=1.0pt, mark size=2.8pt, mark=star, mark options={solid, mycolor2}]
  table[row sep=crcr]{%
-80	48.7215035345449\\
-70	25.4851855006002\\
-60	14.1007295319545\\
-50	7.99624078343483\\
-40	4.76851129808874\\
-30	3.36082430364935\\
};
\addlegendentry{SAMV2}

\addplot [color=orange, dashed, line width=0.7pt, mark size=2.8pt, mark=square, mark options={solid,orange}]
  table[row sep=crcr]{%
-80	10.4234341749732\\
-70	1.90538010206188\\
-60	0.708891152904778\\
-50	0.499873317284823\\
-40	0.405117266973402\\
-30	0.360841608835048\\
};
\addlegendentry{IAA}

\addplot [color=mycolor4,  line width=1.0pt, mark size=2.8pt, mark=diamond, mark options={solid, mycolor4}]
  table[row sep=crcr]{%
-80	14.3819382096666\\
-70	0.930082433622598\\
-60	0.656561751348137\\
-50	1.85753600234289\\
-40	2.46740214260535\\
-30	2.75778534335075\\
};
\addlegendentry{MUSIC}

\addplot [color=mycolor5, dashed, line width=1.0pt, mark size=2.8pt, mark=+, mark options={solid, mycolor5}]
  table[row sep=crcr]{%
-80	13.7119146088959\\
-70	0.931786188352941\\
-60	0.657794944741776\\
-50	2.50974279634016\\
-40	2.47756705704373\\
-30	2.76467608981676\\
};
\addlegendentry{R-MUSIC}

\addplot [color=cyan, dashed, line width=1.4pt, mark size=3pt, mark=|, mark options={solid, cyan}]
  table[row sep=crcr]{%
-80	7.58242233062759\\
-70	0.871071562310852\\
-60	0.617481443715788\\
-50	0.475698433884326\\
-40	0.392537641507156\\
-30	0.349686335640004\\
};
\addlegendentry{MLE}

\addplot [color=black, line width=0.7pt]
  table[row sep=crcr]{%
-80	1.65939054241786\\
-70	0.842494658141308\\
-60	0.576300287457201\\
-50	0.448282044311836\\
-40	0.376153298045459\\
-30	0.332727126097474\\
};

\nextgroupplot[
width= 0.79\columnwidth, 
height= 4.8cm, 
scale only axis,
tick label style={font=\scriptsize} , 
xmin=-80,
xmax=-30,
xlabel style={font=\small\color{white!15!black}},
xlabel={$\theta\text{ (deg)}$},
ymin=0.0351759348147907,
ymax=0.150455690272155,
ylabel style={font=\small\color{white!15!black}},
ylabel={Root NMSE of  $\hat \gamma$},
yminorticks=true,
axis background/.style={fill=white},
xmajorgrids,
ymajorgrids,
yminorgrids,
ylabel style={font=\color{white!15!black}},
             yticklabel style={%
                 /pgf/number format/.cd,
                     fixed,
                     fixed zerofill,
                     precision=2,
                     },
legend style={anchor=north west, legend cell align=left, align=left,draw=none,font=\scriptsize,at={(0.44,0.73)}, legend columns=2} 
]

\addplot [color=blue, line width=1.0pt, mark size=3.5pt, mark=triangle, mark options={solid, rotate=180, blue}]
  table[row sep=crcr]{%
-80	0.0898344100942886\\
-70	0.086836832079358\\
-60	0.0701186772475869\\
-50	0.0665498726482214\\
-40	0.0658285269923627\\
-30	0.0655134573986153\\
};

\addplot [color=red, line width=1.0pt, mark size=3.0pt, mark=triangle, mark options={solid, rotate=90, red}]
  table[row sep=crcr]{%
-80	0.0418081832449748\\
-70	0.0425008906394486\\
-60	0.0419566718826229\\
-50	0.0414257787664828\\
-40	0.0411759348147907\\
-30	0.0413670623128282\\
};

\addplot [color=mycolor1, dashed, line width=1.0pt, mark size=2.5pt, mark=o, mark options={solid, mycolor1}]
  table[row sep=crcr]{%
-80	0.150455690272155\\
-70	0.143288930656015\\
-60	0.136876260310345\\
-50	0.130643707302419\\
-40	0.12595725844725\\
-30	0.123306920139187\\
};

\addplot [color=mycolor2, dashed, line width=1.0pt, mark size=2.8pt, mark=star, mark options={solid, mycolor2}]
  table[row sep=crcr]{%
-80	0.150391816293909\\
-70	0.143820986175084\\
-60	0.137700304576946\\
-50	0.131718181995784\\
-40	0.127266175332536\\
-30	0.124797992065236\\
};

\addplot [color=orange, dashed, line width=1.0pt, mark size=2.5pt, mark=square, mark options={solid, orange}]
  table[row sep=crcr]{%
-80	0.0901246671592491\\
-70	0.0870518702679986\\
-60	0.0701853925834908\\
-50	0.066578594026026\\
-40	0.0658471925193633\\
-30	0.0655271886677369\\
};

\addplot [color=cyan, dashed, line width=1.4pt, mark size=3.8pt, mark=|, mark options={solid,cyan}]
  table[row sep=crcr]{%
-80	0.0418083033906211\\
-70	0.0425013755407697\\
-60	0.0419577421404959\\
-50	0.0414274226919146\\
-40	0.041178418019675\\
-30	0.0413694567558496\\
};

\addplot [color=black, line width=0.7pt]
  table[row sep=crcr]{%
-80	0.0417609268034489\\
-70	0.0417609268034489\\
-60	0.0417609268034489\\
-50	0.0417609268034489\\
-40	0.0417609268034489\\
-30	0.0417609268034489\\
};

\end{groupplot}
\end{tikzpicture}%

%% file: tikz/SNRvsRMSE_sources=2_theta1=-20.02_theta2=3.02_N=20_L=125_M=1801_LL=1000_nmethods=8_SNRplus=3_SNR-7to-21_20-Jan-2024.tex
\definecolor{mycolor1}{rgb}{0.92900,0.69400,0.12500}%
\definecolor{mycolor2}{rgb}{0.46600,0.67400,0.18800}%
\definecolor{mycolor3}{rgb}{1.00000,0.64706,0.00000}%
\definecolor{mycolor4}{rgb}{0.34200,0.87200,0.50300}%
\definecolor{mycolor5}{rgb}{1.00000,0.00000,1.00000}%
\begin{tikzpicture}

\begin{axis}[%
width= 0.8\columnwidth, 
height= 4.8cm, 
 tick label style={font=\scriptsize} ,
scale only axis,
xmin=-19.5,
xmax=-5.5,
xlabel style={font=\small\color{white!15!black}},
xlabel={SNR (dB)},
ymode=log,
ymin=0.142579681048756,
ymax=58.7364946860127,
yminorticks=true,
ylabel style={font=\small\color{white!15!black}},
ylabel={RMSE of $\hat{\boldsymbol{\theta}}$ (deg.)},
axis background/.style={fill=white},
xmajorgrids,
ymajorgrids,
yminorgrids,
legend style={legend cell align=left, align=left, draw=white!15!black,font=\footnotesize}
]
\addplot [color=blue, line width=1.0pt, mark size=3.1pt, mark=triangle, mark options={solid, rotate=180, blue}]
  table[row sep=crcr]{%
-5.5	0.147349923651152\\
-6.5	0.169050288375975\\
-7.5	0.193447667341842\\
-8.5	0.22247247020699\\
-9.5	0.258561404699154\\
-10.5	0.301403384188034\\
-11.5	0.357015405830056\\
-12.5	4.11868425592446\\
-13.5	9.65619687040399\\
-14.5	19.9643554366275\\
-15.5	31.436917469752\\
-16.5	39.4092850480696\\
-17.5	47.1601637613781\\
-18.5	53.0713343529254\\
-19.5	58.3523888285647\\
};
\addlegendentry{CL-BCD}

\addplot [color=red, line width=1.0pt, mark size=3.1pt, mark=triangle, mark options={solid, rotate=90, red}]
  table[row sep=crcr]{%
-5.5	0.147234506824998\\
-6.5	0.167469400190006\\
-7.5	0.189662858778413\\
-8.5	0.219918166598396\\
-9.5	0.253740024434459\\
-10.5	0.295972971738974\\
-11.5	0.350202798389733\\
-12.5	0.415405825669309\\
-13.5	3.14372358835824\\
-14.5	5.87707886623959\\
-15.5	10.7407535117421\\
-16.5	15.9847424126884\\
-17.5	21.6321938323417\\
-18.5	27.0513844747362\\
-19.5	32.0790706847938\\
};
\addlegendentry{CL-OMP}

\addplot [color=mycolor1, dashed, line width=1.0pt, mark size=2.5pt, mark=o, mark options={solid, mycolor1}]
  table[row sep=crcr]{%
-5.5	0.162861904692289\\
-6.5	0.187387299462904\\
-7.5	0.221900878772482\\
-8.5	0.259923065540556\\
-9.5	0.308295312971183\\
-10.5	0.800356170714014\\
-11.5	0.829837333457588\\
-12.5	1.90328295321531\\
-13.5	4.99850677702852\\
-14.5	8.07557576894675\\
-15.5	12.3464692928788\\
-16.5	15.4929045049661\\
-17.5	19.2423687731007\\
-18.5	22.3851865303821\\
-19.5	24.8016883699477\\
};
\addlegendentry{SBL}

\addplot [color=mycolor2, dashed, line width=1.0pt, mark size=2.8pt, mark=x, mark options={solid, mycolor2}]
  table[row sep=crcr]{%
-5.5	0.162511538051919\\
-6.5	0.185332134288687\\
-7.5	0.218842409052724\\
-8.5	0.255968748092418\\
-9.5	0.301761495224291\\
-10.5	0.793790904457843\\
-11.5	0.982842815510192\\
-12.5	2.10692287471563\\
-13.5	4.74799157539269\\
-14.5	7.90729713619009\\
-15.5	11.6474223757877\\
-16.5	15.2174696319723\\
-17.5	18.9142866109193\\
-18.5	22.2391908575829\\
-19.5	24.6021841306824\\
};
\addlegendentry{SAMV2}

\addplot [color=orange, dashed, line width=1.0pt, mark size=2.5pt, mark=square, mark options={solid, orange}]
  table[row sep=crcr]{%
-5.5	0.147465250143891\\
-6.5	0.169150820275871\\
-7.5	0.19349935400409\\
-8.5	0.222535390443857\\
-9.5	0.258329247279513\\
-10.5	0.301768122902337\\
-11.5	0.356864119799119\\
-12.5	4.12030799819625\\
-13.5	9.81395424892535\\
-14.5	20.2346991082151\\
-15.5	31.7828220899277\\
-16.5	39.7307101119525\\
-17.5	47.513220349709\\
-18.5	53.4599028618646\\
-19.5	58.7364946860127\\
};
\addlegendentry{IAA}

\addplot [color=mycolor4, line width=1.0pt, mark size=3.1pt, mark=diamond, mark options={solid, mycolor4}]
  table[row sep=crcr]{%
-5.5	0.147275252503604\\
-6.5	0.16803571049036\\
-7.5	0.194535343832426\\
-8.5	0.224307824205933\\
-9.5	0.263749881516561\\
-10.5	0.315860728803059\\
-11.5	0.386437575812705\\
-12.5	3.21629973727574\\
-13.5	8.05221385707062\\
-14.5	14.7618408743625\\
-15.5	20.6184690023289\\
-16.5	24.949322034877\\
-17.5	28.5563541790614\\
-18.5	32.4051604840957\\
-19.5	35.6270365312637\\
};
\addlegendentry{MUSIC}

\addplot [color=mycolor5, dashed, line width=1.0pt, mark size=2.8pt, mark=+, mark options={solid, mycolor5}]
  table[row sep=crcr]{%
-5.5	0.143067489380246\\
-6.5	0.164415155332826\\
-7.5	0.190183925165675\\
-8.5	0.221888528066421\\
-9.5	0.26193728912096\\
-10.5	0.314543733641064\\
-11.5	0.388392204317641\\
-12.5	3.32151052941911\\
-13.5	9.85067734789521\\
-14.5	15.1949889166089\\
-15.5	20.6188370427766\\
-16.5	25.223351873188\\
-17.5	29.806507368881\\
-18.5	34.5710848998759\\
-19.5	37.9900627615974\\
};
\addlegendentry{R-MUSIC}

\addplot [color=black, line width=0.8pt]
  table[row sep=crcr]{%
-5.5	0.142579681048756\\
-6.5	0.163562472548216\\
-7.5	0.188462571612736\\
-8.5	0.218239551395516\\
-9.5	0.254123516066876\\
-10.5	0.297689174025468\\
-11.5	0.350947591457344\\
-12.5	0.416459625651587\\
-13.5	0.497476362561353\\
-14.5	0.598113671779777\\
-15.5	0.723570248048631\\
-16.5	0.880401221995841\\
-17.5	1.07686262873614\\
-18.5	1.32334584841393\\
-19.5	1.63292581136152\\
};
\addlegendentry{CRLB}

\end{axis}
\end{tikzpicture}%

%% file: tikz/SNRvsRMSE_sources=2_theta1=-20.02_theta2=3.02_N=20_L=25_M=1801_LL=1000_nmethods=8_SNRplus=3_SNR-2to-16_19-Jan-2024.tex
\definecolor{mycolor1}{rgb}{0.34200,0.87200,0.50300}%
\definecolor{mycolor2}{rgb}{1.00000,0.00000,1.00000}%
\definecolor{mycolor3}{rgb}{0.92900,0.69400,0.12500}%
\definecolor{mycolor4}{rgb}{0.46600,0.67400,0.18800}%
\definecolor{mycolor5}{rgb}{1.00000,0.64706,0.00000}%
\begin{tikzpicture}

\begin{axis}[%
width= 0.8\columnwidth, 
height= 4.8cm, 
 tick label style={font=\scriptsize} , 
scale only axis,
xmin=-14.5,
xmax=-0.499999999999999,
xlabel style={font=\color{white!15!black}},
xlabel={SNR (dB)},
ymode=log,
ymin=0.167002559756267,
ymax=36.7506271511113,
yminorticks=true,
ylabel style={font=\color{white!15!black}},
ylabel={RMSE of $\hat{\boldsymbol{\theta}}$ (deg.)},
axis background/.style={fill=white},
xmajorgrids,
ymajorgrids,
yminorgrids,
legend style={legend cell align=left, align=left, draw=white!15!black,font=\footnotesize}
]
\addplot [color=blue, line width=1.0pt, mark size=3.1pt, mark=triangle, mark options={solid, rotate=180, blue}]
  table[row sep=crcr]{%
-0.499999999999999	0.170762993649093\\
-1.5	0.190667249416359\\
-2.5	0.21496046148071\\
-3.5	0.245059176526813\\
-4.5	0.279774909525497\\
-5.5	0.317323179109248\\
-6.5	0.367676488233882\\
-7.5	0.426940276853801\\
-8.5	3.25630250437517\\
-9.5	6.22831871374611\\
-10.5	11.9357804101785\\
-11.5	18.983383418137\\
-12.5	24.3638006887267\\
-13.5	30.4623612019817\\
-14.5	36.3595753825591\\
};
\addlegendentry{CL-BCD}

\addplot [color=red, line width=1.0pt, mark size=3.1pt, mark=triangle, mark options={solid, rotate=90, red}]
  table[row sep=crcr]{%
-0.499999999999999	0.176968923825626\\
-1.5	0.196529895944612\\
-2.5	0.218796709298838\\
-3.5	0.247159867292407\\
-4.5	0.278829697127118\\
-5.5	0.317968551904115\\
-6.5	0.3654859778432\\
-7.5	0.421962083604677\\
-8.5	0.971861101186791\\
-9.5	4.3060067347834\\
-10.5	7.59370094222836\\
-11.5	12.1424496704742\\
-12.5	16.8918823107432\\
-13.5	21.9192744861686\\
-14.5	26.6377942780554\\
};
\addlegendentry{CL-OMP}

\addplot [color=mycolor3, dashed, line width=1.0pt, mark size=2.5pt, mark=o, mark options={solid, mycolor3}]
  table[row sep=crcr]{%
-0.499999999999999	0.186917093921343\\
-1.5	0.213044596270358\\
-2.5	0.243055549206349\\
-3.5	0.281886501982624\\
-4.5	0.325542623937327\\
-5.5	0.374395512793623\\
-6.5	1.99134477175601\\
-7.5	2.83708089415864\\
-8.5	4.11339883794412\\
-9.5	5.57841518712976\\
-10.5	9.26807725474923\\
-11.5	12.3572067232041\\
-12.5	14.970211888948\\
-13.5	17.7917490427445\\
-14.5	20.3331051735833\\
};
\addlegendentry{SBL}

\addplot [color=mycolor4, dashed, line width=1.0pt, mark size=2.8pt, mark=x, mark options={solid, mycolor4}]
  table[row sep=crcr]{%
-0.499999999999999	0.191149156419797\\
-1.5	0.218096309001322\\
-2.5	0.247774090655177\\
-3.5	0.288541158242632\\
-4.5	0.333310665895947\\
-5.5	0.62288201129909\\
-6.5	2.57134361764429\\
-7.5	3.10580392169242\\
-8.5	4.33963109031171\\
-9.5	6.28933827997826\\
-10.5	9.7438972695734\\
-11.5	12.9202503071729\\
-12.5	15.6358788688068\\
-13.5	18.2911985938593\\
-14.5	20.5334777862884\\
};
\addlegendentry{SAMV2}

\addplot [color=orange, dashed, line width=1.0pt, mark size=2.5pt, mark=square, mark options={solid, orange}]
  table[row sep=crcr]{%
-0.499999999999999	0.170821544308674\\
-1.5	0.190525588832577\\
-2.5	0.215011627592556\\
-3.5	0.24486731100741\\
-4.5	0.279599713876821\\
-5.5	0.317228624181362\\
-6.5	0.367790701350646\\
-7.5	0.426926223134631\\
-8.5	3.25435093375008\\
-9.5	6.22917602897847\\
-10.5	11.9353656835474\\
-11.5	19.2155201334754\\
-12.5	24.8255068024804\\
-13.5	30.4577463381649\\
-14.5	36.7506271511113\\
};
\addlegendentry{IAA}

\addplot [color=mycolor1, line width=1.0pt, mark size=3.1pt, mark=diamond, mark options={solid, mycolor1}]
  table[row sep=crcr]{%
-0.499999999999999	0.172933513235579\\
-1.5	0.193736935043373\\
-2.5	0.218641258686461\\
-3.5	0.251081660023188\\
-4.5	0.287749891398763\\
-5.5	0.331212922453216\\
-6.5	0.867198939113745\\
-7.5	2.83254726350682\\
-8.5	5.24681827396375\\
-9.5	8.43614663220122\\
-10.5	12.5382659885648\\
-11.5	17.7146726190466\\
-12.5	22.1323655310498\\
-13.5	26.6527253765914\\
-14.5	28.8894817537456\\
};
\addlegendentry{MUSIC}

\addplot [color=mycolor2, dashed, line width=1.0pt, mark size=2.8pt, mark=+, mark options={solid, mycolor2}]
  table[row sep=crcr]{%
-0.499999999999999	0.167002559756267\\
-1.5	0.189502529391903\\
-2.5	0.215752353346363\\
-3.5	0.246736613853881\\
-4.5	0.283914995118726\\
-5.5	0.581049526122686\\
-6.5	1.38630850670296\\
-7.5	2.70615654983434\\
-8.5	6.03749818302258\\
-9.5	9.92180989946245\\
-10.5	15.0151009931177\\
-11.5	19.0736702820715\\
-12.5	23.8115865506162\\
-13.5	28.7869509633057\\
-14.5	32.6375187271589\\
};
\addlegendentry{R-MUSIC}

\addplot [color=black, line width=1.0pt]
  table[row sep=crcr]{%
-0.499999999999999	0.168212276618162\\
-1.5	0.190267687726193\\
-2.5	0.215626165468828\\
-3.5	0.244929073488398\\
-4.5	0.278984532030145\\
-5.5	0.318817859035257\\
-6.5	0.365736807185754\\
-7.5	0.4214151213405\\
-8.5	0.487998472299433\\
-9.5	0.568237456606795\\
-10.5	0.665653229286712\\
-11.5	0.784742671038446\\
-12.5	0.931232032841063\\
-13.5	1.11239096388652\\
-14.5	1.33742282837158\\
};
\addlegendentry{CRLB}

\end{axis}
\end{tikzpicture}%

%% file: tikz/SNRvsRMSEgamma_sources=2_theta1=-20.02_theta2=3.02_N=20_L=25_M=1801_LL=1000_nmethods=8_SNRplus=3_SNR-2to-16_19-Jan-2024.tex
\definecolor{mycolor1}{rgb}{0.92900,0.69400,0.12500}%
\definecolor{mycolor2}{rgb}{0.46600,0.67400,0.18800}%
\definecolor{mycolor3}{rgb}{1.00000,0.64706,0.00000}%
\begin{tikzpicture}

\begin{axis}[%
width= 0.78\columnwidth, 
height= 4.7cm, 
 tick label style={font=\scriptsize} , 
scale only axis,
xmin=-14.5,
xmax=-0.499999999999999,
xlabel style={font=\color{white!15!black}},
ymin=0.209569362996591,
ymax=1.50397676507501,
yminorticks=true,
ylabel style={font=\color{white!15!black}},
ylabel={Root NMSE of $\hat{\boldsymbol{\sigma}}^2_s$},
axis background/.style={fill=white},
xmajorgrids,
ymajorgrids,
yminorgrids,
legend style={legend cell align=left, align=left, draw=white!15!black,font=\footnotesize}
]
\addplot [color=blue, line width=1.0pt, mark size=3.1pt, mark=triangle, mark options={solid, rotate=180, blue}]
  table[row sep=crcr]{%
-0.499999999999999	0.219075634635024\\
-1.5	0.225438444935296\\
-2.5	0.234264679834373\\
-3.5	0.246575397488519\\
-4.5	0.263695730876903\\
-5.5	0.287397490890979\\
-6.5	0.319906359125578\\
-7.5	0.363931682196256\\
-8.5	0.423728643999224\\
-9.5	0.503796077979119\\
-10.5	0.613885345459713\\
-11.5	0.761527305376752\\
-12.5	0.947423055232439\\
-13.5	1.18783151428013\\
-14.5	1.50292796526782\\
};
\addlegendentry{CL-BCD}

\addplot [color=red, line width=1.0pt, mark size=3.1pt, mark=triangle, mark options={solid, rotate=90, red}]
  table[row sep=crcr]{%
-0.499999999999999	0.210863916855349\\
-1.5	0.213603980107622\\
-2.5	0.217048336950752\\
-3.5	0.221349321005266\\
-4.5	0.226749446892765\\
-5.5	0.233512866230965\\
-6.5	0.241984038094559\\
-7.5	0.252690021548543\\
-8.5	0.267211261541493\\
-9.5	0.288247697112676\\
-10.5	0.321914649880364\\
-11.5	0.36856010373513\\
-12.5	0.422398195048759\\
-13.5	0.471683325006789\\
-14.5	0.525746723428601\\
};
\addlegendentry{CL-OMP}

\addplot [color=mycolor1, dashed, line width=1.0pt, mark size=2.5pt, mark=o, mark options={solid, mycolor1}]
  table[row sep=crcr]{%
-0.499999999999999	0.538891868471689\\
-1.5	0.565308489746048\\
-2.5	0.591732597481262\\
-3.5	0.618691849990562\\
-4.5	0.645496298154972\\
-5.5	0.671581448468015\\
-6.5	0.696206206633849\\
-7.5	0.718586085453142\\
-8.5	0.739201326498302\\
-9.5	0.757663778761032\\
-10.5	0.777110272234118\\
-11.5	0.79254165195088\\
-12.5	0.802444624084474\\
-13.5	0.808368694383029\\
-14.5	0.808743571076896\\
};
\addlegendentry{SBL}

\addplot [color=mycolor2, dashed, line width=1.0pt, mark size=2.8pt, mark=x, mark options={solid, mycolor2}]
  table[row sep=crcr]{%
-0.499999999999999	0.558749542004811\\
-1.5	0.58460844177736\\
-2.5	0.609622625646077\\
-3.5	0.63495018886602\\
-4.5	0.659980435028589\\
-5.5	0.683840054837618\\
-6.5	0.706829310362303\\
-7.5	0.726927401531616\\
-8.5	0.74587709911103\\
-9.5	0.763177388380863\\
-10.5	0.781232976923681\\
-11.5	0.795026414452841\\
-12.5	0.804304500293435\\
-13.5	0.808583162912327\\
-14.5	0.806451658616813\\
};
\addlegendentry{SAMV2}

\addplot [color=orange, dashed, line width=1.0pt, mark size=2.4pt, mark=square, mark options={solid, orange}]
  table[row sep=crcr]{%
-0.499999999999999	0.219111321064401\\
-1.5	0.225478693075153\\
-2.5	0.234309780789725\\
-3.5	0.246625480667405\\
-4.5	0.263750333984394\\
-5.5	0.287456426512931\\
-6.5	0.319967629842649\\
-7.5	0.363992903967246\\
-8.5	0.423786839471285\\
-9.5	0.503850946233926\\
-10.5	0.613942288961214\\
-11.5	0.762183971578173\\
-12.5	0.948525341625369\\
-13.5	1.1880847696013\\
-14.5	1.50397676507501\\
};
\addlegendentry{IAA}

\addplot [color=black, line width=1.0pt]
  table[row sep=crcr]{%
-0.499999999999999	0.209569362996591\\
-1.5	0.212058552151707\\
-2.5	0.215198472992464\\
-3.5	0.219160822214229\\
-4.5	0.224163253323801\\
-5.5	0.230481850501755\\
-6.5	0.238466951974838\\
-7.5	0.248563138448791\\
-8.5	0.261334348654993\\
-9.5	0.277495254493666\\
-10.5	0.297950255383875\\
-11.5	0.323841796065847\\
-12.5	0.356610254858569\\
-13.5	0.398068459756376\\
-14.5	0.450494984205369\\
};
\addlegendentry{CRLB}

\end{axis}
\end{tikzpicture}%

%% file: tikz/SNRvsRMSEgamma_2sources_theta1=-20.02_theta2=3.02_N=20_M=1801_L=25_LL=1000_SNRplus=3_SNR-2to-16_rho=0dot95_20-Jan-2024.tex
\definecolor{mycolor1}{rgb}{0.92900,0.69400,0.12500}%
\definecolor{mycolor2}{rgb}{0.46600,0.67400,0.18800}%
\definecolor{mycolor3}{rgb}{1.00000,0.64706,0.00000}%
\begin{tikzpicture}

\begin{axis}[%
width= 0.78\columnwidth, 
height= 4.7cm, 
 tick label style={font=\scriptsize} , 
 scale only axis,
xmin=-14.5,
xmax=-0.499999999999999,
xlabel style={font=\small\color{white!15!black}},
xlabel={SNR (dB)},
ymin=0.209569362996591,
ymax=1.47458722603997,
ylabel style={font=\small\color{white!15!black}},
ylabel={Root NMSE of $\hat{\boldsymbol{\sigma}}^2_s$},
axis background/.style={fill=white},
xmajorgrids,
ymajorgrids,
legend style={legend cell align=left, align=left, draw=white!15!black,font=\footnotesize}
]
\addplot [color=blue, line width=1.0pt, mark size=3.1pt, mark=triangle, mark options={solid, rotate=180, blue}]
  table[row sep=crcr]{%
-0.499999999999999	0.218496627651914\\
-1.5	0.223639825193967\\
-2.5	0.230997741563245\\
-3.5	0.241502502979147\\
-4.5	0.256474989437246\\
-5.5	0.277746897045324\\
-6.5	0.307618137479862\\
-7.5	0.349429643065747\\
-8.5	0.405523447243669\\
-9.5	0.484246409780907\\
-10.5	0.592176512365628\\
-11.5	0.730912197348299\\
-12.5	0.915814288533881\\
-13.5	1.157485930739\\
-14.5	1.47418538041037\\
};
\addlegendentry{CL-BCD}

\addplot [color=red, line width=1.0pt, mark size=3.1pt, mark=triangle, mark options={solid, rotate=90, red}]
  table[row sep=crcr]{%
-0.499999999999999	0.211525565439853\\
-1.5	0.214026394139425\\
-2.5	0.216947160100509\\
-3.5	0.220771365381292\\
-4.5	0.225597683508018\\
-5.5	0.231774519418286\\
-6.5	0.23949356876984\\
-7.5	0.249250246716623\\
-8.5	0.261478845775611\\
-9.5	0.283737590467341\\
-10.5	0.311589889345929\\
-11.5	0.350352161594472\\
-12.5	0.391360343472384\\
-13.5	0.446061330561143\\
-14.5	0.498978062715913\\
};
\addlegendentry{CL-OMP}

\addplot [color=mycolor1, dashed, line width=1.0pt, mark size=2.5pt, mark=o, mark options={solid, mycolor1}]
  table[row sep=crcr]{%
-0.499999999999999	0.54000117810613\\
-1.5	0.565678449314794\\
-2.5	0.592303254697768\\
-3.5	0.621256588941069\\
-4.5	0.649631790922571\\
-5.5	0.674394050721744\\
-6.5	0.698604747460859\\
-7.5	0.720859948125031\\
-8.5	0.74104757584226\\
-9.5	0.761025599260156\\
-10.5	0.777670543591742\\
-11.5	0.790652169021518\\
-12.5	0.801919025960602\\
-13.5	0.808361580835318\\
-14.5	0.809050690205219\\
};
\addlegendentry{SBL}

\addplot [color=mycolor2, dashed, line width=1.0pt, mark size=2.8pt, mark=x, mark options={solid, mycolor2}]
  table[row sep=crcr]{%
-0.499999999999999	0.559212760468967\\
-1.5	0.584785043241458\\
-2.5	0.610117862335192\\
-3.5	0.637166797190712\\
-4.5	0.663540542111602\\
-5.5	0.686448337401553\\
-6.5	0.708570042834985\\
-7.5	0.729129293369182\\
-8.5	0.748450689836374\\
-9.5	0.766411697227064\\
-10.5	0.781546507662262\\
-11.5	0.793527217155043\\
-12.5	0.802925264378727\\
-13.5	0.808670542997998\\
-14.5	0.807808825706335\\
};
\addlegendentry{SAMV2}

\addplot [color=orange, dashed, line width=1.0pt, mark size=2.4pt, mark=square, mark options={solid, orange}]
  table[row sep=crcr]{%
-0.499999999999999	0.218533299214543\\
-1.5	0.223680175166385\\
-2.5	0.231041064970267\\
-3.5	0.241551532433412\\
-4.5	0.256525958283945\\
-5.5	0.277796341673809\\
-6.5	0.307664365719193\\
-7.5	0.349466475639981\\
-8.5	0.405548773983028\\
-9.5	0.48425705348468\\
-10.5	0.592169743143692\\
-11.5	0.730911364020034\\
-12.5	0.916264424414823\\
-13.5	1.15780069996333\\
-14.5	1.47458722603997\\
};
\addlegendentry{IAA}

\addplot [color=black, line width=0.8pt]
  table[row sep=crcr]{%
-0.499999999999999	0.209569362996591\\
-1.5	0.212058552151707\\
-2.5	0.215198472992464\\
-3.5	0.219160822214229\\
-4.5	0.224163253323801\\
-5.5	0.230481850501755\\
-6.5	0.238466951974838\\
-7.5	0.248563138448791\\
-8.5	0.261334348654993\\
-9.5	0.277495254493666\\
-10.5	0.297950255383875\\
-11.5	0.323841796065847\\
-12.5	0.356610254858569\\
-13.5	0.398068459756376\\
-14.5	0.450494984205369\\
};
\addlegendentry{CRLB}

\end{axis}
\end{tikzpicture}%

%% file: tikz/SNRvsRMSE_2sources_theta1=-20.02_theta2=3.02_N=20_M=1801_L=200_LL=1000_SNRplus=3_SNR-8to-21_rho=0dot95_04-Dec-2023.tex
\definecolor{mycolor1}{rgb}{0.92900,0.69400,0.12500}%
\definecolor{mycolor2}{rgb}{0.46600,0.67400,0.18800}%
\definecolor{mycolor3}{rgb}{1.00000,0.64706,0.00000}%
\definecolor{mycolor4}{rgb}{0.34200,0.87200,0.50300}%
\definecolor{mycolor5}{rgb}{1.00000,0.00000,1.00000}%

\begin{tikzpicture}

\begin{axis}[%
width= 0.78\columnwidth, 
height= 4.7cm, 
 tick label style={font=\scriptsize} , 
scale only axis,
xmin=-19.5,
xmax=-6.5,
xlabel style={font=\small\color{white!15!black}},
xlabel={},
ymode=log,
ymin=0.129066440410971,
ymax=62.3412563075208,
yminorticks=true,
ylabel style={font=\small\color{white!15!black}},
ylabel={RMSE of $\hat{\boldsymbol{\theta}}$ (deg.)},
axis background/.style={fill=white},
xmajorgrids,
ymajorgrids,
yminorgrids,
legend style={legend cell align=left, align=left, draw=white!15!black,font=\scriptsize, at={(0.3,0.58)}}
]
\addplot [color=blue, line width=1.0pt, mark size=2.5pt, mark=triangle, mark options={solid, rotate=180, blue}]
  table[row sep=crcr]{%
-6.5	0.266881996395406\\
-7.5	0.295181300220723\\
-8.5	0.327475189899932\\
-9.5	0.364790350749573\\
-10.5	0.408820253901389\\
-11.5	0.45401541824039\\
-12.5	0.509886261827082\\
-13.5	5.18411496786095\\
-14.5	13.2922502233444\\
-15.5	25.8990280512609\\
-16.5	37.8086029099199\\
-17.5	47.1684380491871\\
-18.5	55.1917228758081\\
-19.5	61.9296481178442\\
};
\addlegendentry{CL-BCD}

\addplot [color=red, line width=1.0pt, mark size=2.5pt, mark=triangle, mark options={solid, rotate=90, red}]
  table[row sep=crcr]{%
-6.5	0.279306283495371\\
-7.5	0.295553717621682\\
-8.5	0.312160215274139\\
-9.5	0.337031155829838\\
-10.5	0.367937494691689\\
-11.5	0.4075487700877\\
-12.5	0.457117052843138\\
-13.5	0.518721505241489\\
-14.5	2.51300218861823\\
-15.5	4.89694537441454\\
-16.5	9.76657616567852\\
-17.5	16.7007741137948\\
-18.5	22.3985891966436\\
-19.5	27.2971223391771\\
};
\addlegendentry{CL-OMP}

\addplot [color=mycolor1, dashed, line width=1.0pt, mark size=2.5pt, mark=o, mark options={solid, mycolor1}]
  table[row sep=crcr]{%
-6.5	0.184081503687901\\
-7.5	0.219517653048675\\
-8.5	0.256651514704274\\
-9.5	0.303940783706296\\
-10.5	0.359354977703102\\
-11.5	1.02292717238325\\
-12.5	1.05755850901971\\
-13.5	2.54796742522348\\
-14.5	4.78708094771751\\
-15.5	7.81779380643925\\
-16.5	11.5369772470955\\
-17.5	14.9912296360239\\
-18.5	17.9196545725636\\
-19.5	20.8958398730465\\
};
\addlegendentry{SBL}

\addplot [color=mycolor2, dashed, line width=1.0pt, mark size=2.8pt, mark=x, mark options={solid, mycolor2}]
  table[row sep=crcr]{%
-6.5	0.18257053431482\\
-7.5	0.216037033862249\\
-8.5	0.25150745515789\\
-9.5	0.297106041675355\\
-10.5	0.347882163957851\\
-11.5	0.407450610503896\\
-12.5	1.04956848275851\\
-13.5	1.63580255532261\\
-14.5	3.70289562369776\\
-15.5	7.41051874027723\\
-16.5	11.2219201565507\\
-17.5	14.8694231226366\\
-18.5	17.7139534830596\\
-19.5	20.752934924969\\
};
\addlegendentry{SAMV2}

\addplot [color=orange, dashed, line width=1.0pt, mark size=2.5pt, mark=square, mark options={solid, orange}]
  table[row sep=crcr]{%
-6.5	0.265514594702433\\
-7.5	0.293969386161207\\
-8.5	0.326609859006118\\
-9.5	0.363089520641944\\
-10.5	0.40755613110343\\
-11.5	0.453014348558626\\
-12.5	0.509185624306103\\
-13.5	5.4708193170676\\
-14.5	13.9227193464495\\
-15.5	26.5560175478177\\
-16.5	38.3190414285118\\
-17.5	47.6756506195773\\
-18.5	55.5184555620922\\
-19.5	62.3412563075208\\
};
\addlegendentry{IAA}

\addplot [color=mycolor4, line width=1.0pt, mark size=2.9pt, mark=diamond, mark options={solid, mycolor4}]
  table[row sep=crcr]{%
-6.5	3.3070149682153\\
-7.5	3.54592357503656\\
-8.5	5.1774874215202\\
-9.5	6.02678388529073\\
-10.5	6.84861533450376\\
-11.5	7.3584191236977\\
-12.5	7.70242013915107\\
-13.5	8.54492972469638\\
-14.5	10.3264911756123\\
-15.5	12.6765078787496\\
-16.5	15.0855823222042\\
-17.5	18.4511119990097\\
-18.5	22.3886218423556\\
-19.5	26.5042582616454\\
};
\addlegendentry{MUSIC}

\addplot [color=mycolor5, dashed, line width=1.0pt, mark size=2.8pt, mark=+, mark options={solid, mycolor5}]
  table[row sep=crcr]{%
-6.5	8.70286074572803\\
-7.5	9.14025369079303\\
-8.5	11.1953195761331\\
-9.5	12.5366884900691\\
-10.5	13.8316557040697\\
-11.5	14.8500837424432\\
-12.5	15.2500180310455\\
-13.5	16.5300450826907\\
-14.5	17.493146511383\\
-15.5	19.0034434126272\\
-16.5	20.9915981876401\\
-17.5	23.761094217712\\
-18.5	27.0610908859055\\
-19.5	30.5999483963772\\
};
\addlegendentry{R-MUSIC}

\addplot [color=black, line width=1.2pt]
  table[row sep=crcr]{%
-6.5	0.129066440410971\\
-7.5	0.14865999600756\\
-8.5	0.172075895362683\\
-9.5	0.200275472436002\\
-10.5	0.23448911867264\\
-11.5	0.276287954171987\\
-12.5	0.327672482876883\\
-13.5	0.391182427744657\\
-14.5	0.470033334911152\\
-15.5	0.568287277340699\\
-16.5	0.691067077180434\\
-17.5	0.844825944908574\\
-18.5	1.03768741727168\\
-19.5	1.27987415187425\\
};

\end{axis}
\end{tikzpicture}%

%% file: tikz/SNRvsRMSE_2sources_theta1=-20.02_theta2=3.02_N=20_M=1801_L=25_LL=1000_SNRplus=3_SNR-2to-16_rho=0dot95_20-Jan-2024.tex
\definecolor{mycolor1}{rgb}{0.92900,0.69400,0.12500}%
\definecolor{mycolor2}{rgb}{0.46600,0.67400,0.18800}%
\definecolor{mycolor3}{rgb}{1.00000,0.64706,0.00000}%
\definecolor{mycolor4}{rgb}{0.34200,0.87200,0.50300}%
\definecolor{mycolor5}{rgb}{1.00000,0.00000,1.00000}%
\begin{tikzpicture}

\begin{axis}[%
width= 0.78\columnwidth, 
height= 4.7cm, 
 tick label style={font=\scriptsize} , 
 scale only axis,
xmin=-14.5,
xmax=-0.499999999999999,
xlabel style={font=\small \color{white!15!black}},
ymode=log,
ymin=0.168119443425159,
ymax=37.3773331579448,
yminorticks=true,
ylabel style={font=\small\color{white!15!black}},
ylabel={RMSE of $\hat{\boldsymbol{\theta}}$ (deg.)},
axis background/.style={fill=white},
xmajorgrids,
ymajorgrids,
yminorgrids,
legend style={legend cell align=left, align=left, draw=white!15!black, font=\scriptsize,at={(0.31,0.58)}}
]
\addplot [color=blue, line width=1.0pt, mark size=2.2pt, mark=triangle, mark options={solid, rotate=180, blue}]
  table[row sep=crcr]{%
-0.499999999999999	0.210076176659798\\
-1.5	0.23628795991332\\
-2.5	0.265827011419077\\
-3.5	0.300596074491997\\
-4.5	0.340373324454191\\
-5.5	0.38587821913137\\
-6.5	0.440520147098856\\
-7.5	2.38916805603959\\
-8.5	3.76258661029883\\
-9.5	7.04918903137091\\
-10.5	12.3023104334105\\
-11.5	18.198901560259\\
-12.5	23.9675668769276\\
-13.5	30.9389142020207\\
-14.5	37.0803136718124\\
};
\addlegendentry{CL-BCD}

\addplot [color=red, line width=1.0pt, mark size=2.5pt, mark=triangle, mark options={solid, rotate=90, red}]
  table[row sep=crcr]{%
-0.499999999999999	0.28903978964841\\
-1.5	0.302608658170907\\
-2.5	0.321238229356341\\
-3.5	0.344801392108554\\
-4.5	0.373432724864864\\
-5.5	0.407899497425527\\
-6.5	0.451216134463294\\
-7.5	0.509780344854524\\
-8.5	0.726618194101962\\
-9.5	4.58435666151751\\
-10.5	7.60828508929574\\
-11.5	12.1631391507291\\
-12.5	16.3763585085329\\
-13.5	21.5471969870793\\
-14.5	26.4750131255869\\
};
\addlegendentry{CL-OMP}

\addplot [color=mycolor1, dashed, line width=1.0pt, mark size=2.5pt, mark=o, mark options={solid, mycolor1}]
  table[row sep=crcr]{%
-0.499999999999999	0.191096834092037\\
-1.5	0.217651096941871\\
-2.5	0.24788707106261\\
-3.5	0.779231672867577\\
-4.5	0.793492280995853\\
-5.5	0.8164324834302\\
-6.5	1.4550292093288\\
-7.5	1.97880418434973\\
-8.5	4.06005492573685\\
-9.5	6.90508204151117\\
-10.5	9.64276754879011\\
-11.5	12.2401063720868\\
-12.5	15.1245717955914\\
-13.5	18.4341496684821\\
-14.5	21.1362743642298\\
};
\addlegendentry{SBL}

\addplot [color=mycolor2, dashed, line width=1.0pt, mark size=2.5pt, mark=x, mark options={solid, mycolor2}]
  table[row sep=crcr]{%
-0.499999999999999	0.195417501775045\\
-1.5	0.222027025382046\\
-2.5	0.766061355245128\\
-3.5	0.776462491045126\\
-4.5	0.795653190781011\\
-5.5	1.55792811130681\\
-6.5	1.76502407915586\\
-7.5	3.02644610062694\\
-8.5	5.42502958517278\\
-9.5	7.79309271085619\\
-10.5	10.2248334949768\\
-11.5	12.6432675365192\\
-12.5	15.6751302386934\\
-13.5	18.9236392905804\\
-14.5	21.3262664805634\\
};
\addlegendentry{SAMV2}

\addplot [color=orange, dashed, line width=1.0pt, mark size=2.5pt, mark=square, mark options={solid,orange}]
  table[row sep=crcr]{%
-0.499999999999999	0.209513722700922\\
-1.5	0.235673502965434\\
-2.5	0.265702841535423\\
-3.5	0.300346466601487\\
-4.5	0.340035292285959\\
-5.5	0.385141532426714\\
-6.5	0.440145430511323\\
-7.5	2.38899811636594\\
-8.5	3.76384457702494\\
-9.5	7.04995205657457\\
-10.5	12.3017953161317\\
-11.5	18.3118191887098\\
-12.5	24.080002284053\\
-13.5	31.2316955671638\\
-14.5	37.3773331579448\\
};
\addlegendentry{IAA}

\addplot [color=mycolor4, line width=1.0pt, mark size=2.9pt, mark=diamond, mark options={solid, mycolor4}]
  table[row sep=crcr]{%
-0.499999999999999	3.49998828569468\\
-1.5	3.83797863464611\\
-2.5	5.72409136894232\\
-3.5	6.55983765043007\\
-4.5	7.54534015137819\\
-5.5	7.91382537083046\\
-6.5	8.2653431870673\\
-7.5	10.0281106894569\\
-8.5	11.578511389639\\
-9.5	13.4488077538494\\
-10.5	14.92246045396\\
-11.5	17.4348544014569\\
-12.5	20.7108276512553\\
-13.5	23.9391085882495\\
-14.5	27.3683472646779\\
};
\addlegendentry{MUSIC}

\addplot [color=mycolor5, dashed, line width=1.0pt, mark size=2.8pt, mark=+, mark options={solid, mycolor5}]
  table[row sep=crcr]{%
-0.499999999999999	6.80463303040512\\
-1.5	8.84073427334335\\
-2.5	10.4719637159616\\
-3.5	11.9861875101782\\
-4.5	13.0844864060019\\
-5.5	13.8347682378829\\
-6.5	15.0583080873639\\
-7.5	16.291524627314\\
-8.5	18.0240666415937\\
-9.5	19.9197015474798\\
-10.5	21.2965378473326\\
-11.5	23.2478305000951\\
-12.5	25.7012013437781\\
-13.5	28.3923948317339\\
-14.5	32.1761754912029\\
};
\addlegendentry{R-MUSIC}

\addplot [color=black, line width=0.8pt]
  table[row sep=crcr]{%
-0.499999999999999	0.168119443425159\\
-1.5	0.190137562263664\\
-2.5	0.215444079457874\\
-3.5	0.244674794046316\\
-4.5	0.278630273499131\\
-5.5	0.318325648562051\\
-6.5	0.365055020952828\\
-7.5	0.420473965072443\\
-8.5	0.486704129958801\\
-9.5	0.566464578659347\\
-10.5	0.663235383711524\\
-11.5	0.78146034382068\\
-12.5	0.926797738601908\\
-13.5	1.10643098935706\\
-14.5	1.32945503399761\\
};

\end{axis}
\end{tikzpicture}%

%% file: tikz/SNRvsRMSE_sources=4_at-30.10_-20.02_-10.02_3.02_N=20_L=125_M=1801_LL=2000_SNR-6to-0.50to-12_19-Jun-2024.tex
\definecolor{mycolor1}{rgb}{0.34200,0.87200,0.50300}%
\definecolor{mycolor2}{rgb}{1.00000,0.00000,1.00000}%
\definecolor{mycolor3}{rgb}{1.00000,0.64706,0.00000}%
\begin{tikzpicture}

\begin{axis}[%
width= 0.78\columnwidth, 
height= 4.7cm, 
 tick label style={font=\scriptsize} , 
scale only axis,
xmin=-14,
xmax=-8,
xlabel style={font=\color{white!15!black}},
xlabel={SNR (dB)},
ymode=log,
ymin=0.302734348545646,
ymax=27.9262727373346,
yminorticks=true,
ylabel style={font=\color{white!15!black}},
ylabel={RMSE of DOA (deg.)},
axis background/.style={fill=white},
xmajorgrids,
ymajorgrids,
yminorgrids,
legend style={legend cell align=left, align=left, draw=white!15!black,font=\scriptsize}
]
\addplot [color=blue, line width=1.0pt, mark size=2.5pt, mark=triangle, mark options={solid, rotate=180, blue}]
  table[row sep=crcr]{%
-8	0.326788922700876\\
-8.5	0.350402625560938\\
-9	0.375017332932758\\
-9.5	0.404622045865026\\
-10	0.436713865133681\\
-10.5	0.472422480413452\\
-11	3.63331185008939\\
-11.5	4.92270332236263\\
-12	8.88121326171148\\
-12.5	12.7667772362488\\
-13	17.3865941748233\\
-13.5	22.4883968303657\\
-14	27.6742542988967\\
};
\addlegendentry{CL-BCD}

\addplot [color=red, line width=1.0pt, mark size=2.8pt, mark=triangle, mark options={solid, rotate=90, red}]
  table[row sep=crcr]{%
-8	0.337170579973995\\
-8.5	0.358134053114193\\
-9	0.38245391879284\\
-9.5	0.409962193378853\\
-10	0.439881802306029\\
-10.5	0.471908889511525\\
-11	0.511818327143529\\
-11.5	0.554478133022394\\
-12	0.755129128030431\\
-12.5	2.30799718370712\\
-13	4.92532780634955\\
-13.5	6.66630294841151\\
-14	8.57881979062388\\
};
\addlegendentry{CLOMP}

\addplot [color=orange, dashed, line width=1.0pt, mark size=2.5pt, mark=square, mark options={solid,orange}]
  table[row sep=crcr]{%
-8	0.32669710742521\\
-8.5	0.35032841734578\\
-9	0.374975999231949\\
-9.5	0.404605981171805\\
-10	0.436601649103618\\
-10.5	0.472591790026021\\
-11	3.63328749206555\\
-11.5	4.92270545533653\\
-12	9.22845306646786\\
-12.5	13.0264041085788\\
-13	17.7105258250567\\
-13.5	22.7060259182447\\
-14	27.9262727373346\\
};
\addlegendentry{IAA}

\addplot [color=mycolor1, line width=1.0pt, mark size=2.5pt, mark=diamond, mark options={solid, mycolor1}]
  table[row sep=crcr]{%
-8	0.324166623821762\\
-8.5	0.350968659569483\\
-9	0.381727651605172\\
-9.5	0.415977162834692\\
-10	0.790018987113599\\
-10.5	1.04248357301206\\
-11	1.86502144759785\\
-11.5	4.03180505481106\\
-12	5.08478141123097\\
-12.5	6.89771302679373\\
-13	9.88460166116975\\
-13.5	12.3399005263414\\
-14	14.3361887543377\\
};
\addlegendentry{MUSIC}

\addplot [color=mycolor2, dashed, line width=1.0pt, mark size=2.8pt, mark=+, mark options={solid, mycolor2}]
  table[row sep=crcr]{%
-8	0.31840567669074\\
-8.5	0.345746915866651\\
-9	0.376726032715818\\
-9.5	0.412253580944465\\
-10	0.660143684280998\\
-10.5	1.83441307250994\\
-11	2.7302437236324\\
-11.5	4.34924959731924\\
-12	6.84862509244044\\
-12.5	9.33112785283665\\
-13	11.5782480957215\\
-13.5	13.8948475283365\\
-14	16.5941174180238\\
};
\addlegendentry{R-MUSIC}

\addplot [color=black, line width=1.0pt]
  table[row sep=crcr]{%
-8	0.302734348545646\\
-8.5	0.326413432578675\\
-9	0.352447368950927\\
-9.5	0.381121974064003\\
-10	0.412760088446054\\
-10.5	0.447726141825284\\
-11	0.486431236697659\\
-11.5	0.529338812425413\\
-12	0.576970961121026\\
-12.5	0.629915476887484\\
-13	0.68883373140509\\
-13.5	0.754469481405001\\
-14	0.827658727316119\\
};
\addlegendentry{CRLB}

\end{axis}
\end{tikzpicture}%